\pdfoutput=1

\RequirePackage[l2tabu,orthodox]{nag}

\documentclass[11pt,letterpaper,reqno,oneside]{article}

\DeclareTextFontCommand{\emph}{\slshape}

\usepackage[T1]{fontenc}
\usepackage{lmodern}
\usepackage[utf8]{inputenc}

\usepackage[american,german,french,british]{babel}

\usepackage[activate={true,nocompatibility}]{microtype}

\usepackage{csquotes}

\usepackage[backend=biber,autolang=other,style=apa,maxcitenames=5,sortcites=false]{biblatex}
\DeclareLanguageMapping{british}{british-apa}
\DeclareLanguageMapping{american}{american-apa}
\DeclareLanguageMapping{german}{german-apa}
\DeclareLanguageMapping{french}{french-apa}
\DefineBibliographyExtras{french}{\restorecommand\mkbibnamefamily}

\usepackage{amsmath}
\usepackage{amssymb}
\usepackage{amsfonts}
\usepackage{amsthm}
\usepackage[retainorgcmds]{IEEEtrantools}
\usepackage{mathtools}
\usepackage{stmaryrd}
\usepackage{mleftright}
\usepackage{centernot}

\let\originalleft\left
\let\originalright\right
\renewcommand{\left}{\mathopen{}\mathclose\bgroup\originalleft}
\renewcommand{\right}{\aftergroup\egroup\originalright}

\usepackage{graphicx}
\usepackage{tikz,pgfplots}
\pgfplotsset{compat=1.10}
\usetikzlibrary{shapes.multipart}
\usetikzlibrary{decorations.markings}
\usepackage{etoolbox}
\AtBeginEnvironment{tikzpicture}{\shorthandoff{;}}

\usepackage[title]{appendix}

\usepackage{datetime}
\newdateformat{datestyle}{ {\THEDAY} \monthname[\THEMONTH] \THEYEAR }

\usepackage{enumerate}
\usepackage{enumitem}
\setlist[enumerate,1]{label=(\arabic*)}
\setlist[itemize,1]{label=--}
\setlist[itemize,2]{label=--}
\setlist[itemize,3]{label=--}
\setlist[itemize,4]{label=--}

\usepackage{xcolor}
\usepackage{verbatim}
\usepackage{gensymb}
\usepackage{rotating}
\usepackage{subcaption}
\usepackage{floatpag}
\floatpagestyle{plain}
\usepackage{marvosym}

\usepackage{hyphenat}
\theoremstyle{definition}

\newtheorem{theorem}{Theorem}
\newtheorem{proposition}{Proposition}
\newtheorem{lemma}{Lemma}
\newtheorem{corollary}{Corollary}
\newtheorem{observation}{Observation}
\newtheorem{example}{Example}
\newtheorem{fact}{Fact}
\newtheorem{definition}{Definition}

\newtheoremstyle{named}
	{\topsep}					
	{\topsep}					
	{}							
	{0pt}						
	{\bfseries}					
	{}							
	{5pt plus 1pt minus 1pt}	
	{\thmnote{#3}}				
\theoremstyle{named}
\newtheorem{namedthm}{}

\renewcommand{\qedsymbol}{$\blacksquare$}

\usepackage{xpatch}
\xpatchcmd{\proof}{\itshape}{\proofheadfont}{}{}
\newcommand{\proofheadfont}{\slshape}

\usepackage{hyperref}
\hypersetup{
	colorlinks=true,
	linkcolor=black,
	citecolor=black,
	filecolor=black,
	urlcolor=black,
}

\usepackage[nameinlink]{cleveref}
\crefname{page}{p.}{pp.}
\crefname{equation}{equation}{equations}
\crefname{section}{section}{sections}
\crefname{subsection}{section}{sections}
\crefname{subsubsection}{section}{sections}
\crefname{appsec}{appendix}{appendices}
\crefname{supplsec}{supplemental appendix}{supplemental appendices}
\crefname{footnote}{footnote}{footnotes}
\crefname{figure}{figure}{figures}
\crefname{table}{table}{tables}
\crefname{theorem}{theorem}{theorems}
\crefname{proposition}{proposition}{propositions}
\crefname{lemma}{lemma}{lemmata}
\crefname{corollary}{corollary}{corollaries}
\crefname{remark}{remark}{remarks}
\crefname{observation}{observation}{observations}
\crefname{example}{example}{examples}
\crefname{fact}{fact}{facts}
\crefname{definition}{definition}{definitions}
\crefname{assumption}{assumption}{assumptions}
\crefname{exercise}{exercise}{exercises}
\crefname{notation}{notation}{notation}
\crefname{claim}{claim}{claims}
\crefname{conjecture}{conjecture}{conjectures}

\DeclareMathOperator*{\tr}{tr}

\newcommand{\N}{\mathbf{N}}

\newcommand{\union}{\cup}

\newcommand{\intersect}{\cap}

\newcommand{\Union}{\bigcup}

\DeclarePairedDelimiter\abs{\lvert}{\rvert}

\DeclarePairedDelimiter\floor{\lfloor}{\rfloor}

\makeatletter
\let\save@mathaccent\mathaccent
\newcommand*\if@single[3]{%
	\setbox0\hbox{${\mathaccent"0362{#1}}^H$}%
	\setbox2\hbox{${\mathaccent"0362{\kern0pt#1}}^H$}%
	\ifdim\ht0=\ht2 #3\else #2\fi
	}
\newcommand*\rel@kern[1]{\kern#1\dimexpr\macc@kerna}
\newcommand*\widebar[1]{\@ifnextchar^{{\wide@bar{#1}{0}}}{\wide@bar{#1}{1}}}
\newcommand*\wide@bar[2]{\if@single{#1}{\wide@bar@{#1}{#2}{1}}{\wide@bar@{#1}{#2}{2}}}
\newcommand*\wide@bar@[3]{%
	\begingroup
	\def\mathaccent##1##2{%
	  \let\mathaccent\save@mathaccent
	  \if#32 \let\macc@nucleus\first@char \fi
	  \setbox\z@\hbox{$\macc@style{\macc@nucleus}_{}$}%
	  \setbox\tw@\hbox{$\macc@style{\macc@nucleus}{}_{}$}%
	  \dimen@\wd\tw@
	  \advance\dimen@-\wd\z@
	  \divide\dimen@ 3
	  \@tempdima\wd\tw@
	  \advance\@tempdima-\scriptspace
	  \divide\@tempdima 10
	  \advance\dimen@-\@tempdima
	  \ifdim\dimen@>\z@ \dimen@0pt\fi
	  \rel@kern{0.6}\kern-\dimen@
	  \if#31
	    \overline{\rel@kern{-0.6}\kern\dimen@\macc@nucleus\rel@kern{0.4}\kern\dimen@}%
	    \advance\dimen@0.4\dimexpr\macc@kerna
	    \let\final@kern#2%
	    \ifdim\dimen@<\z@ \let\final@kern1\fi
	    \if\final@kern1 \kern-\dimen@\fi
	  \else
	    \overline{\rel@kern{-0.6}\kern\dimen@#1}%
	  \fi
	}%
	\macc@depth\@ne
	\let\math@bgroup\@empty \let\math@egroup\macc@set@skewchar
	\mathsurround\z@ \frozen@everymath{\mathgroup\macc@group\relax}%
	\macc@set@skewchar\relax
	\let\mathaccentV\macc@nested@a
	\if#31
	  \macc@nested@a\relax111{#1}%
	\else
	  \def\gobble@till@marker##1\endmarker{}%
	  \futurelet\first@char\gobble@till@marker#1\endmarker
	  \ifcat\noexpand\first@char A\else
	    \def\first@char{}%
	  \fi
	  \macc@nested@a\relax111{\first@char}%
	\fi
	\endgroup
}
\makeatother

\usepackage{sidecap}
\sidecaptionvpos{figure}{c}

\makeatletter
\def\namedlabel#1#2{\begingroup
	#2%
	\def\@currentlabel{#2}%
	\phantomsection\label{#1}\endgroup
}
\makeatother

\DeclareMathOperator*{\str}{str}
\newcommand{\rel}{\mathrel{Q}}
	\newcommand{\relp}{\mathrel{Q'}}
	\newcommand{\nrel}{\centernot{\rel}}
	
\newcommand{\pref}{\succ}
	\newcommand{\prefi}{\mathrel{\mathord{\pref}_i}}
	\newcommand{\prefeq}{\succeq}
	
\newcommand{\apref}{\prec}
\newcommand{\rank}{\mathrel{R}}
	\newcommand{\rankp}{\mathrel{\rank'}}
	\newcommand{\rankpp}{\mathrel{\rank''}}
	\newcommand{\rankarg}[1]{\mathrel{\mathord{\rank}_{#1}}}
	\newcommand{\rankparg}[1]{\mathrel{\mathord{\rank}'_{#1}}}
	\newcommand{\rankzero}{\rankarg{0}}
	\newcommand{\rankt}{\rankarg{t}}
	\newcommand{\ranktone}{\rankarg{t-1}}
	\newcommand{\ranktplusone}{\rankarg{t+1}}
	\newcommand{\rankT}{\rankarg{T}}
	\newcommand{\rankTone}{\rankarg{T-1}}

	\newcommand{\rankpT}{\rankparg{T}}

	\newcommand{\nrank}{\centernot{\rank}}
	\newcommand{\nrankp}{\mathrel{\centernot{\rank}\hspace{-0.27em}'}}
	\newcommand{\nrankpp}{\mathrel{\centernot{\rank}\hspace{-0.27em}''}}
	\newcommand{\nrankarg}[1]{\mathrel{\mathord{\nrank}_{#1}}}
	
	\newcommand{\nrankt}{\nrankarg{t}}
	\newcommand{\nranktone}{\nrankarg{t-1}}

	\newcommand{\nrankTone}{\nrankarg{T-1}}
	
	\newcommand{\rankfn}[2]{\mathrel{\mathord{\rank}^{\mathord{#1}}(\mathord{#2})}}
	\newcommand{\ranksup}[1]{\mathrel{\mathord{\rank}^{#1}}}
	\newcommand{\nranksup}[1]{\mathrel{\mathord{\nrank}^{#1}}}
	\newcommand{\rankh}{\ranksup{h}}
	\newcommand{\rankhp}{\ranksup{h'}}
	\newcommand{\rankhpp}{\ranksup{h''}}
	\newcommand{\nrankh}{\nranksup{h}}
\newcommand{\perm}{\mathrel{W}}
	\newcommand{\permp}{\mathrel{\perm'}}
	\newcommand{\permpp}{\mathrel{\perm''}}
	\newcommand{\permsup}[1]{\mathrel{\mathord{\perm}^{#1}}}
	\newcommand{\permarg}[1]{\mathrel{\mathord{\perm}_{#1}}}
	\newcommand{\permparg}[1]{\mathrel{\mathord{\perm}'_{#1}}}
\newcommand{\votei}{\mathrel{V_i}}
	\newcommand{\votepi}{\mathrel{V'_i}}

\newcommand{\sayh}{\mathrel{S^h}}
\newcommand{\strat}{\sigma}
	\newcommand{\stratp}{\strat'}
	\newcommand{\stratv}[1]{\strat_{#1}}
	\newcommand{\stratvp}[1]{\strat_{#1}'}
	\newcommand{\strati}{\stratv{i}}
	\newcommand{\stratip}{\stratvp{i}}

\title{\scshape Agenda-manipulation in ranking%
\thanks{We are deeply grateful to Eddie Dekel for many detailed comments.
We thank him,
Péter Es\H{o},
Alessandro Pavan,
John Quah and
Bruno Strulovici
for their guidance and support.
We have profited from comments and suggestions provided by
Georgy Egorov, Benny Moldovanu and Wojciech Olszewski (particularly fruitful),
and by
Nageeb Ali,
Nemanja Antić,
David Austen-Smith,
Marco Battaglini,
Steve Coate,
Tommaso Denti,
Francesc Dilmé,
Piotr Dworczak,
Jeff Ely,
Matteo Escudé,
Ben Golub,
Sergiu Hart,
Claudia Herresthal,
Matt Jackson,
Mathijs Janssen,
Stephan Lauermann,
Jiangtao Li,
Shengwu Li,
Elliot Lipnowski,
Francisco Poggi,
Marek Pycia,
Sven Rady,
Ariel Rubinstein,
Eddie Schlee,
Marciano Siniscalchi,
Rani Spiegler,
Francesco Squintani,
Jean-Marc Tallon,
Asher Wolinsky,
Andrea Galeotti (as editor),
four anonymous referees,
and audiences at Bocconi, Bonn, Cornell, Hebrew/Tel Aviv, Hong Kong, Northwestern, Oxford, Paris, Penn State, Royal Holloway, UCL, Warwick and the $18^{\text{th}}$ IO Theory Conference at UC Berkeley.
Curello gratefully acknowledges funding from the German Research Foundation (DFG) through CRC TR 224 (Project B02), as well as from the Hausdorff Center for Mathematics.}}

\author{Gregorio Curello \\ University of Bonn
\and Ludvig Sinander \\ University of Oxford}

\date{28 September 2022}

\makeatletter
	\AtBeginDocument{ \hypersetup{
		pdftitle = {Agenda-manipulation in ranking},
		pdfauthor = {Gregorio Curello and Ludvig Sinander}
		} }
\makeatother

\begin{document}

\maketitle

\begin{abstract}
	We study the susceptibility of committee governance (e.g. by boards of directors),
	modelled as the collective determination of a ranking of a set of alternatives,
	to manipulation of the order in which pairs of alternatives are voted on---\emph{agenda-manipulation.}
	We exhibit an agenda strategy called \emph{insertion sort}
	that allows a self-interested committee chair
	with no knowledge of how votes will be cast
	to do
	as well as if she had complete knowledge.
	Strategies with this `regret-freeness' property
	are characterised by their \emph{efficiency,}
	and by their avoidance of two intuitive errors.
	What distinguishes regret-free strategies from each other
	is how they prioritise among alternatives;
	insertion sort prioritises \emph{lexicographically.}
\end{abstract}

\section{Introduction}
\label{sec:intro}

In many modern organisations, governance is entrusted not to a single individual, but to a committee:
a group that reaches decisions by aggregating the individual votes of its members, usually via a majority or super-majority rule.
Public firms and charities are run by boards of directors, for example.

To govern is to make decisions: to choose one alternative from among those that are available.
These alternatives may be e.g. policies, candidates for a job, or investments,
or they could be coarser things like \emph{attributes} of policies, candidates or investments.

In practice, most governing committees meet too infrequently to be able to consider and solve each individual decision problem faced by their organisation.
Lacking the ability or time to foresee which decision problems (sets of available alternatives) will arise at any given future date,
the committee cannot specify what alternative ought to be chosen.
Committees therefore tend to set only high-level priorities,
while delegating the details of day-to-day decision-making to executive officers.
Such committee governance may be captured in stylised fashion
by thinking of the committee as determining a \emph{ranking} of all the potentially-available alternatives,
with executives instructed to choose the highest-ranked alternative from among those that turn out to be available.

As observed by Condorcet (\citeyear{Condorcet1785}),
the majority will of a committee can easily fail to be transitive:
there may be a majority for alternative $\alpha$ over $\beta$, for $\beta$ over $\gamma$, and for $\gamma$ over $\alpha$.
Whenever there are such (Condorcet) cycles, the order in which questions are considered by the committee will affect what decision is reached.
The chair of the committee can thus influence collective decisions via her control of the agenda.

We model agenda-setting by the chair
as her determining the order in which pairs of alternatives are voted on.
We assume that the committee is sovereign, so that a majority vote for $\alpha$ over $\beta$ leads to $\alpha$ being ranked above $\beta$.
Suppose that $\beta$ had already been ranked above $\gamma$.
Since rankings are by their nature transitive,
it is then necessary also that $\alpha$ be ranked above $\gamma$.
These are the rules: transitivity is imposed.

We abstract from strategic voting:
we assume that if $\alpha$ garners a majority over $\beta$ under one agenda-setting strategy of the chair,
then the same is true under every agenda-setting strategy.
This assumption buys tractability, but at a price:
it is a strong idealisation in many applied contexts,
albeit empirically reasonable in some cases.
We view this as the main limitation of our analysis.

The chair is uncertain how voters will vote on any given pair of alternatives,
so that the only way to find out whether $\alpha$ would beat $\beta$ or vice-versa
is to offer a (binding) vote on this pair.
In this uncertain world, how much influence can the chair exert via her control of the agenda?
And in doing so, what sorts of agenda-setting strategies should we expect the chair to use?

Our first finding is that agenda-setting is surprisingly powerful.
In particular, there exist \emph{regret-free} agenda-setting strategies:
no matter what the unknown voting behaviour of the committee,
the ranking reached by such a strategy
is not unambiguously worse (from the chair's perspective) than any other ranking that the chair could have reached had she perfectly known how voters would behave.

We prove this by exhibiting a regret-free strategy called \emph{insertion sort,} which works as follows.
Let's call the chair's $k^\text{th}$-favourite alternative simply `$k$'.
No votes involving $k$ are offered
until all of the alternatives that the chair
considers strictly worse than $k$ have been totally ranked.
Once that has happened,
insertion sort pits $k$ against whichever of the worse alternatives is ranked highest; if $k$ wins, then by transitivity, it is ranked above all of the worse alternatives.
If not, then $k$ is next pitted against the \emph{second-}highest-ranked of the worse alternatives; if $k$ wins, then it is ranked above all but the highest-ranked of the worse alternatives.
If not, then $k$ is pitted against the \emph{third-}highest ranked of the worse alternatives; and so on.

Having established that regret-free agenda-setting is possible,
we next ask what its qualitative properties are:
that is, what do regret-free strategies have in common?
We provide two characterisations of regret-free strategies.
The first describes regret-free agenda-setting in terms of its outcomes: regret-free strategies are exactly those that (ex post) reach rankings that satisfy an \emph{efficiency} property.
The second is a qualitative description of behaviour:
it identifies two intuitive types of \emph{error,} and shows that avoiding these is both necessary and sufficient for regret-freeness.

Finally, we ask what it is that distinguishes one regret-free strategy from another.
The answer is \emph{prioritisation:} while one regret-free strategy focusses on getting a certain pair of alternatives ranked `right' (i.e. the way around that the chair prefers), another strategy may prioritise another pair.
We show that insertion sort is characterised by \emph{lexicographic} prioritisation:
it places the chair's favourite alternative as highly as possible,
and subject to that, places her second-favourite alternative as well as possible, and so on.

\subsection{Related literature}
\label{sec:intro:lit}

We contribute to the agenda-manipulation literature initiated by \textcite{Farquharson1969} and \textcite{Black1958}, which asks how a committee's choice can be influenced by varying the order in which binary questions are voted on (the \emph{agenda}).
Two classes of agenda are emphasised: the \emph{amendment procedure} used in Anglo-Saxon and Scandinavian legislatures, and the \emph{successive procedure} widely used in continental Europe.
Under complete information,
for both sincere and strategic voting,
\textcite{Miller1977} and \textcite{Banks1985} characterise which alternatives an agenda-setter can induce a committee to choose using
(i) amendment agendas, (ii) successive agendas, and (iii) arbitrary agendas.%
	\footnote{Part (i) under strategic voting due to \textcite{Banks1985}; the rest are from \textcite{Miller1977}.
	See \textcite[§4.10]{Myerson1991} for a nice textbook treatment of (iii) under strategic voting.}$^,$%
	\footnote{Related work by \textcite{ApesteguiaBallesterMasatlioglu2014} and \textcite{Horan2021} axiomatises the decision rules (which specify a choice for each choice set and preference profile)
	defined by various agendas under strategic voting.}
Extensions include
super- and sub-majority voting rules \parencite{BarberaGerber2017}
and random agendas \parencite{RoesslerShelegiaStrulovici2018}.

Specifically, this paper belongs to the literature on agenda-setting under incomplete information about voters' preferences.
This literature long consisted of a single pioneering paper \parencite{OrdeshookPalfrey1988},
but has recently received interest from \textcite{KleinerMoldovanu2017}, \textcite{GershkovMoldovanuShi2017,GershkovMoldovanuShi2019te,GershkovMoldovanuShi2020} and \textcite{GershkovKleinerMoldovanuShi2020}.
We depart from these papers by considering a committee tasked not with choosing a single alternative, but rather with \emph{ranking} all of the alternatives.
We establish a link with the older literature by relating insertion sort to the amendment procedure (§\ref{sec:charac:ra_is}).

Also related is the social choice literature on \emph{ranking methods,}
meaning maps that assign to each (possibly cyclic) majority will a (transitive) ranking.
`Impossibility' results such as Arrow's (\citeyear{Arrow1950,Arrow1951,Arrow1963}) theorem assert that certain normatively appealing properties are inconsistent.
The literature beginning with \textcite{Zermelo1929,Wei1952,Kendall1955} studies ranking methods with at least \emph{some} attractive normative properties.%
	\footnote{For example, Copeland's (\citeyear{Copeland1951}) method \parencite{Rubinstein1980},
	the Kemeny--Slater method \parencite{Kemeny1959,Slater1961,YoungLevenglick1978,Young1986,Young1988}
	and the fair-bets method \parencite{Daniels1969,MoonPullman1970,SlutzkiVolij2005}. See \textcite{CharonHudry2010} and \textcite{GonzalezdiazHendrickLohmann2014} for an overview.}
Our chair's problem can be formulated as a choice among ranking methods,
but there is a feasibility constraint,
and the objective reflects the chair's self-interest.
This suggests that solutions to our chair's problem will bear little relation to the normative-motivated ranking methods in the literature;
we confirm this in \cref{suppl:ranking_method}.

Both the voting and social choice literatures
presume the existence of Condorcet cycles.
Cycles are certainly a priori plausible.%
	\footnote{For example,
	the probability of a cycle
	with five (seven) voters and five (six) alternatives
	is $39\%$ ($61\%$).
	See \textcite{Gehrlein1989} for a table of such probabilities.}
And they are empirically common,
especially when there are few voters, as on a committee.%
	\footnote{The authoritative survey of \textcite{GehrleinLepelley2011} 
	includes data on 127 elections, of which 32 had small electorates ($\leq 25$ voters).
	A Condorcet cycle existed in $29\%$ of all elections,
	and in $59\%$ of small elections (see their Table 1.1).}

Regret-based criteria for evaluating strategies appear in decision theory, including minimax regret \parencite{Savage1951} and `regret theory' \parencite{Bell1982,LoomesSugden1982,Fishburn1988}.
The online learning literature \parencite{Gordon1999,Zinkevich2003} studies (asymptotic) regret-freeness.

\subsection{Roadmap}
\label{sec:intro:outline}

We begin in §\ref{sec:applications} with two applications.
We describe the environment and basic concepts in §\ref{sec:environment},
and discuss interpretation and extensions in §\ref{sec:discussion}.
We then (§\ref{sec:is}) establish that regret-free strategies exist by exhibiting one: insertion sort.
In §\ref{sec:charac}, we characterise regret-freeness in terms of efficiency (\Cref{theorem:regretfree_efficient}) and error-avoidance (\Cref{theorem:regretfree_errors}), and show that both characterisations are tight (\Cref{proposition:regretfree_efficient_tight,proposition:regretfree_errors_tight}).
In the same section, we show that regret-free strategies differ in their prioritisation,
and characterise insertion sort in terms of its lexicographic prioritisation (\Cref{theorem:is_lexicographic}).

\section{Two applications}
\label{sec:applications}

Before describing the abstract model,
we fix ideas with two applications.

\begin{namedthm}[Hiring.]
	\label{appl:hiring_intr}
	The alternatives $\mathcal{X}$ are candidates for a job.
	Only some unknown subset of candidates would accept the job if offered it.
	The hiring committee decides the order in which offers should be made (a ranking):
	the job will be offered to the first candidate,
	then to the second if the first declined, and so on.

	Relabelling, we may instead think of the alternatives as investment projects of unknown viability.
	A firm's board (or a lower-level committee) ranks the projects,
	whereupon a manager evaluates the first project (e.g. by commissioning market research) and implements it if viable,
	and otherwise evaluates the second project and implements that if viable, and so on.

	Similarly, the committee could be a policy-making body, such as a ministerial cabinet or a parliamentary committee.
	Any given policy may turn out to be infeasible, for example due to a court ruling or political opposition.
	The committee therefore ranks the policies and tasks a bureaucrat with implementing the first policy if feasible, the second if not, and so on.
\end{namedthm}

\begin{namedthm}[Party lists.]
	\label{appl:party_intr}
	A political party's leadership committee must draw up a \emph{party list,} meaning a ranking of the party's parliamentary candidates $\mathcal{X}$.
	The $K$ top-ranked candidates will earn parliamentary seats, where $K$ is the (uncertain) number of seats won by the party in a subsequent election.

	Electoral systems of this type, called \emph{party-list proportional representation,} are used in most of the world's democracies.
	More precisely,
	we described the \emph{closed-list} variant that gives voters no sway over party lists; this system is used in several dozen states.%
		\footnote{For example,
		Argentina,
		Germany,
		Japan,
		South Africa and
		Turkey.}
	Other countries allow the electorate to influence party lists.%
		\footnote{E.g.
		Brazil,
		Indonesia,
		Iraq,
		the Netherlands and
		Ukraine.}
	In many of these, voters exert little influence on party lists in practice,%
		\footnote{In many countries, such as Indonesia and the Netherlands,
		the party list is only altered if a candidate receives a large number of personal votes.
		Furthermore, voting for an individual candidate is typically optional, and many voters do not: in Sweden, only about a quarter do \parencite{Oscarsson2019}.
		Of course there are exceptions, e.g.
		voting for individual candidates is mandatory and important in Finnish parliamentary elections.}
	making the closed-list system a reasonable idealisation.

	We may re-interpret this as a firm planning for downsizing.
	The firm will have to fire an uncertain number $K$ of its workers $\mathcal{X}$.
	The board (or a lower-level committee) plans ahead by drawing up an order in which employees will (if necessary) be let go.
\end{namedthm}

\section{Environment}
\label{sec:environment}

There is a finite set $\mathcal{X}$ of alternatives.
A \emph{ranking} of the alternatives is a binary relation $\rank$,
where $x \rank y$ reads `$x$ is ranked above $y$ according to $\rank$'.
Formally, a ranking is a binary relation 
that is transitive,
irreflexive (no alternative is ranked above itself)
and total (each distinct pair is ranked).%
	\footnote{Definitions of some standard order-theoretic terms are collected in \cref{app:st_definitions}.}
To capture the notion of a `possibly-incomplete ranking',
we use the term \emph{proto-ranking}
for relations which are transitive and irreflexive, but not necessarily total.

There is a committee and its \emph{chair.}
The committee comprises an odd number $I$ of voters, $i \in \{1,\dots,I\}$.
The committee meets to determine a ranking, and the chair sets the agenda.

\subsection{Interaction}
\label{sec:environment:interaction}

Initially, no alternatives are ranked.
In each period, the chair offers a vote on an as-yet unranked pair of alternatives.
The committee votes on this pair, and whichever alternative garners more votes wins. (The chair does not have a vote, though that can be allowed for: see §\ref{sec:discussion}.)
The winning alternative is ranked above the losing one.
In addition, transitivity is imposed: if alternative $x$ is ranked above $y$ and $y$ above $z$, then $x$ is considered ranked above $z$.
The chair continues to offer votes until all alternatives are ranked.

More explicitly,
the ranking decisions made by the end of a period $t$
are captured by a proto-ranking $\rankt$,
where $x \rankt y$ if $x$ has been ranked above $y$, and $x \nrankt y$ otherwise.
In period $t+1$, the chair offers a vote on an $\rankt$-unranked pair $x,y$ of distinct alternatives.
If $x$ is the winner, then the new proto-ranking $\ranktplusone$
is the transitive closure of $\mathord{\rankt} \union \{(x,y)\}$.%
	\footnote{Equivalently, $z \ranktplusone w$ iff either (1) $z \rankt w$, or (2) both (a) $z=x$ or $z \rankt x$ and (b) $y=w$ or $y \rankt w$.
	See \cref{app:pf_charac:noerror_efficient} (\Cref{observation:tr_cl}, \cpageref{observation:tr_cl}) for a proof of equivalence.}

A \emph{history} is a sequence of pairs offered for votes and a winner of each vote.
A \emph{strategy} of the chair specifies what pair to offer after each non-terminal history.
We give formal versions of these definitions in \cref{app:extra:defns}.

\subsection{The majority will}
\label{sec:environment:general_will}

The chair need not keep track of individual votes: all that matters for each pair of alternatives is which one wins.
This essential information is captured by the binary relation $\perm$ such that $x \perm y$ iff a majority of voters vote for $x$ over $y$ when the pair $x,y$ is offered.
We call $\perm$ the \emph{majority will.}

We consider all possible majority wills $\perm$,
meaning all \emph{total and asymmetric} relations.
Clearly only such relations need be contemplated
(for each pair of alternatives, one must win and the other lose).
Conversely, all such relations must be considered
because each of them is the majority will of \emph{some} committee,
as we show in \cref{app:extra:arise_charac}.

The \emph{outcome} of a strategy under a majority will $\perm$ is the ranking that results.
If a history is visited by a strategy $\strat$ under some majority will, then we say that it belongs to the \emph{path} of $\strat$.
We give formal definitions of outcomes and paths in \cref{app:extra:defns}.

\subsection{Reachable rankings}
\label{sec:environment:W-reachability}

For a majority will $\perm$,
we call a ranking \emph{$\perm$-reachable} iff it is the outcome under $\perm$ of some strategy of the chair.
This captures ex-post feasibility:
were the chair to have perfect knowledge of $\perm$,
she could achieve all and only $\perm$-reachable rankings
by offering some sequence of pairwise votes.

\setcounter{example}{0}
\begin{example}
	\label{example:123_Wreachable}
	There are three alternatives $\mathcal{X} = \{\alpha,\beta,\gamma\}$, and the majority will satisfies $\alpha \perm \gamma \perm \beta \perm \alpha$. Graphically:
	\begin{equation*}
		\vcenter{\hbox{%
		\begin{tikzpicture}[every text node part/.style={align=center}]

			\def \radius {0.8cm}
			\def \margin {30}

			\foreach \s in {1,2,3}
			{
			\draw[->,>=latex,thick] ({90 - 120 * (\s + 2) + \margin}:\radius)
				arc ({90 - 120 * (\s + 2) + \margin}:{90 - 120 * (\s+1)-\margin}:\radius);
			}

			\node at ({90 - 120 * (1 + 2)}:\radius) {$\alpha$};
			\node at ({90 - 120 * (2 + 2)}:\radius) {$\beta$};
			\node at ({90 - 120 * (3 + 2)}:\radius) {$\gamma$};

		\end{tikzpicture}%
		}}
	\end{equation*}
	The ranking $\beta \rank \alpha \rank \gamma$ is $\perm$-reachable, achieved by offering $\beta,\alpha$ and $\alpha,\gamma$.
	Similarly, the rankings $\alpha \rankp \gamma \rankp \beta$ and $\gamma \rankpp \beta \rankpp \alpha$ are $\perm$-reachable.
\end{example}


\subsection{The chair's preferences}
\label{sec:environment:preferences}

The chair has a strict preference over the alternatives $\mathcal{X}$, denoted $\pref$.
We do not fully specify the chair's preferences over rankings.
We assume only that she prefers a ranking over another whenever the former is \emph{more aligned} with her own preference over alternatives.

\begin{definition}
	\label{definition:more_aligned}
	For rankings $\pref$, $\rank$ and $\rankp$,
	we say that $\rank$ is \emph{more aligned with $\pref$ than} $\rankp$ iff for any pair $x,y \in \mathcal{X}$ of alternatives with $x \pref y$, if $x \rankp y$ then also $x \rank y$.
\end{definition}

In words, whenever a pair $x,y \in \mathcal{X}$ is ranked `right' by $\rankp$ (viz. $x \pref y$ and $x \rankp y$), it is also ranked `right' by $\rank$ (i.e. $x \rank y$).

\setcounter{example}{0}
\begin{example}[continued]
	\label{example:123_maw}
	Let the chair's preference be $\alpha \pref \beta \pref \gamma$.
	\Cref{fig:maw} depicts how all rankings are ordered by `more aligned with $\pref$ than'.
	\begin{SCfigure}
		\centering
		\begin{tikzpicture}[every text node part/.style={align=center}]

			\def \radius {2cm}
			\def \buff {8}
			\def \margin {5}

			\node at ({30 + 60 * 1}:{\radius}) {$\alpha \rankarg{\natural} \beta \rankarg{\natural} \gamma$};
			\node at ({30 + 60 * 2 + \buff }:\radius) {$\beta \rank \alpha \rank \gamma$};
			\node at ({30 + 60 * 3 - \buff }:\radius) {$\beta \rankarg{\sharp} \gamma \rankarg{\sharp} \alpha$};
			\node at ({30 + 60 * 4}:{\radius}) {$\gamma \rankpp \beta \rankpp \alpha$};
			\node at ({30 + 60 * 5 + \buff }:\radius) {$\gamma \rankarg{\flat} \alpha \rankarg{\flat} \beta$};
			\node at ({30 + 60 * 6 - \buff }:\radius) {$\alpha \rankp \gamma \rankp \beta$};

			\draw[->,>=latex,thick] ({30 + 60 * (1) + \margin*6}:{\radius})
				arc ({30 + 60 * (1) + \margin*6}:{30 + 60 * (1+1) - \margin}:{\radius});
			\draw[->,>=latex,thick] ({30 + 60 * (2) + \margin*3.5}:{\radius})
				arc ({30 + 60 * (2) + \margin*3.5}:{30 + 60 * (2+1) - \margin*3.5}:{\radius});
			\draw[->,>=latex,thick] ({30 + 60 * (3) + \margin}:{\radius})
				arc ({30 + 60 * (3) + \margin}:{30 + 60 * (3+1) - \margin*6}:{\radius});
			\draw[<-,>=latex,thick] ({30 + 60 * (4) + \margin*6}:{\radius})
				arc ({30 + 60 * (4) + \margin*6}:{30 + 60 * (4+1) - \margin}:{\radius});
			\draw[<-,>=latex,thick] ({30 + 60 * (5) + \margin*3.5}:{\radius})
				arc ({30 + 60 * (5) + \margin*3.5}:{30 + 60 * (5+1) - \margin*3.5}:{\radius});
			\draw[<-,>=latex,thick] ({30 + 60 * (6) + \margin}:{\radius})
				arc ({30 + 60 * (6) + \margin}:{30 + 60 * (6+1) - \margin*6}:{\radius});

		\end{tikzpicture}
		\caption{`More aligned with $\pref$ than' for three alternatives $\mathcal{X} = \{\alpha,\beta,\gamma\}$, where $\alpha \pref \beta \pref \gamma$.
		In this (`Hasse') diagram, there is a directed path from one ranking to another iff the former is more aligned with $\pref$.}
		\label{fig:maw}
	\end{SCfigure}
	Since $\rankpp$ ranks \emph{no} pairs `right',
	every ranking is more aligned with $\pref$ than $\rankpp$.
	The rankings $\rank$ and $\rankp$ are incomparable to each other,
	because $\rank$ ranks $\beta,\gamma$ `right' (while $\rankp$ does not),
	whereas $\rankp$ ranks $\alpha,\beta$ `right' ($\rank$ doesn't).
\end{example}

\subsection{Regret-free strategies}
\label{sec:environment:chair_problem}

Given a majority will $\perm$,
we call a ranking \emph{$\perm$-unimprovable} iff there is no other ranking that is both $\perm$-reachable and more aligned with $\pref$.
Explicitly,
$\rank$ is $\perm$-unimprovable exactly if
for any $\perm$-reachable ranking $\mathord{\rankp} \neq \mathord{\rank}$,
some pair $x,y \in \mathcal{X}$
is ranked `right' by $\rank$ and `wrong' by $\rankp$.

\setcounter{example}{0}
\begin{example}[continued]
	\label{example:123_unimp}
	The ranking $\rank$ is $\perm$-unimprovable,
	since neither of the other two $\perm$-reachable rankings is more aligned with $\pref$
	($\rankpp$ is \emph{less} aligned, and $\rankp$ is incomparable).
	$\rankp$ is also $\perm$-unimprovable;
	$\rankpp$ is not.
\end{example}

The chair does not know the majority will.
One would therefore expect her to face trade-offs: a strategy that does well against $\perm$ may have a regrettable outcome under some $\permp$.
A \emph{regret-free} strategy is one that has no such downside: its outcome under any majority will is unimprovable ex post.

\begin{definition}
	\label{definition:regretfree}
	A strategy is \emph{regret-free} iff for any majority will $\perm$, its outcome under $\perm$ is $\perm$-unimprovable.
\end{definition}

Regret-freeness is a highly demanding optimality property.
Our first result (\Cref{theorem:is_efficient} in §\ref{sec:is}) will be that, surprisingly, regret-free strategies exist.

\section{Discussion}
\label{sec:discussion}

Various aspects of the model deserve comment.

\subsection{\texorpdfstring{`}{'}More aligned than'}
\label{sec:environment:maw}

`More aligned than' captures \emph{unambiguous} superiority of one ranking over another with respect to the chair's preference $\pref$.%
	\footnote{It is an instance of \emph{single-crossing dominance,} a general way of comparing rankings (or preferences);
	see \textcite{CurelloSinander2022}.}
We illustrate by example:

\begin{namedthm}[Hiring {\normalfont(continued from §\ref{sec:applications})}.]
	\label{appl:hiring_maw}
	\hyperref[appl:hiring_intr]{Recall} that the top-ranked candidate in $X$ will be hired, where $X \subseteq \mathcal{X}$ is unknown.
	A ranking is more aligned with $\pref$ exactly if it hires a weakly $\pref$-better candidate at every realisation of the uncertainty $X$,
	as we show in \cref{app:extra:maw_charac}.
	Equivalently, a more aligned ranking
	is one that is preferred by every expected-utility preference consistent with $\pref$.
\end{namedthm}

This application demonstrates two general points.
First, more aligned rankings are unambiguously better for the chair, given her preference $\pref$ over the alternatives.
Secondly, these comparisons are the \emph{only} unambiguous ones:
any further comparisons would have to be based on extraneous cardinal information,
rather than on $\pref$ alone.

\begin{namedthm}[Party lists {\normalfont(continued from §\ref{sec:applications})}.]
	\label{appl:party_maw}
	\hyperref[appl:party_intr]{Recall} that the $K$ highest-ranked candidates win parliamentary seats, where $K$ is uncertain.
	If a ranking $\rank$ is more aligned with $\pref$ than $\rankp$,
	then the $k^\text{th}$ $\pref$-best candidate who wins a seat under $\rank$ is weakly $\pref$-better than the $k^\text{th}$ $\pref$-best under $\rankp$.
	The converse is true provided there is some uncertainty about which candidates can take up seats.%
		\footnote{A party cannot alter its list after submitting it,
		but circumstances may render some of its candidates ineligible for parliamentary seats.
		For example, many countries disqualify candidates convicted of a serious crime,
		and all disqualify the deceased.}
\end{namedthm}

The `more aligned than' criterion has its limitations.
No extra credit is given for ranking \emph{many} pairs of alternatives `right', for example,
nor for having pairs \emph{near the top} ranked `right'.

\subsection{Unimprovability as optimality}
\label{sec:environment:rf}

Unimprovability is the strongest ex-post optimality concept available without further assumptions about the chair's preference over rankings.%
	\footnote{It is a non-trivial optimality concept:
	we show in \cref{suppl:frac} that
	not all $\perm$-reachable rankings are $\perm$-unimprovable
	exactly if $\perm$ contains a (Condorcet) cycle,
	and that for a typical $\perm$, \emph{most} $\perm$-reachable rankings are \emph{not} $\perm$-unimprovable.}
It can therefore be thought of as optimality for a chair who is unable to make fine distinctions between rankings (which are complicated objects),
or who is content to `satisfice'.

Were we fully to specify the chair's preference over rankings, we could still break her problem under full information about $\perm$ into two parts: first, reach the frontier ($\perm$-unimprovability), then choose among the frontier rankings.

\begin{namedthm}[Hiring {\normalfont(continued)}.]
	\label{appl:hiring_unimp}
	A ranking $\rank$ is $\perm$-unimprovable
	exactly if any $\perm$-reachable ranking $\mathord{\rankp} \neq \mathord{\rank}$ hires a strictly $\pref$-worse candidate at some realisation $X \subseteq \mathcal{X}$ of uncertainty.
\end{namedthm}

\subsection{The imposition of transitivity}
\label{sec:environment:protocol}

To understand why we assume that transitivity is imposed automatically,
observe that the purpose of the interaction is to turn the will of the majority, which may contain (Condorcet) cycles, into a (by definition transitive) ranking.
Some pairs will thus necessarily be ranked by fiat.
We require that transitivity be imposed immediately after each vote because
this is necessary and sufficient for \emph{committee sovereignty,}
the requirement that $x$ be ranked above $y$ if $x$ beat $y$ in a vote.%
	\footnote{Sufficiency is obvious.
	For necessity, suppose the chair were allowed to offer $x,z$ even though $x \rankt y \rankt z$;
	then committee sovereignty is violated whenever $z$ beats $x$.}
Indeed, we contend that the protocol described in §\ref{sec:environment:interaction} is the only natural one, given that the interaction must end with a ranking:
as shown in \cref{suppl:trans_charac}, any other protocol must violate either committee sovereignty or \emph{democratic legitimacy,}
the requirement that the chair offer enough votes to give the committee a fair say.

\subsection{Fixed voting}
\label{sec:environment:strategic}

By using the majority will, we implicitly assume fixed voting: $x \perm y$ means that $x$ garners a majority over $y$ whenever $x,y$ is offered,
whatever the strategy of the chair.
This is reasonable if voters are non-strategic: if they are unsophisticated, say, or vote `expressively' (to please their constituents, for example).
Empirically, non-strategic voting appears to be the norm in many important institutions, such as the US Congress;%
	\footnote{See \textcite{Ladha1994,PooleRosenthal1997,Wilkerson1999}, as well as the survey by \textcite{GrosecloseMilyo2010}.}
whether the same is true of small groups of voters (such as committees) is unclear, however.

If voters were strategic,
then they could potentially benefit from voting in a non-fixed way,
by tailoring their behaviour to (their conjecture about) the chair's strategy.
There are limits to their gains from strategising, however:
we show in \cref{suppl:ic} that any insincerity by a voter risks producing a final ranking that is less aligned with her preference over the alternatives
than the ranking that sincere voting would have delivered.
By contrast, sincere voting carries no such risk: no strategising can ever improve on its outcome in the (strong) `more aligned' sense.

\section{Existence of regret-free strategies}
\label{sec:is}

In this section, we exhibit a regret-free strategy: \emph{insertion sort.}
We first (§\ref{sec:is:efficiency}) introduce a notion of \emph{efficiency} that implies regret-freeness,
then (in §\ref{sec:is:is}) define insertion sort and show that it is efficient (\Cref{theorem:is_efficient}).

\subsection{Efficiency}
\label{sec:is:efficiency}

A ranking $\rank$ is called \emph{$\perm$-efficient} iff
every pair on which the chair and committee agree
is ranked accordingly:
if $x \pref y$ and $x \perm y$,
then $x \rank y$.

\setcounter{example}{0}
\begin{example}[continued]
	\label{example:123_Wefficient}
	Recall the details from \cpageref{example:123_Wreachable,example:123_maw}.
	Since $\alpha \pref \beta \pref \gamma$,
	$\perm$-efficiency requires precisely that $\alpha$ be ranked above $\gamma$.
	There are three such rankings: $\pref$ itself, $\beta \rank \alpha \rank \gamma$ and $\alpha \rankp \gamma \rankp \beta$.
\end{example}

An \emph{efficient strategy}
is a strategy whose outcome under any majority will $\perm$ is $\perm$-efficient.
Efficiency matters
because it is linked with regret-freeness:

\begin{lemma}
	\label{lemma:Wefficient_Wunimprovable}
	For any majority will $\perm$, every $\perm$-efficient ranking is $\perm$-unimprovable.
\end{lemma}

\begin{corollary}
	\label{corollary:efficient_regret-free}
	Any efficient strategy is regret-free.
\end{corollary}

\begin{proof}[Proof of \Cref{lemma:Wefficient_Wunimprovable}]
	Fix a majority will $\perm$, and let $\rank$ be a $\perm$-efficient ranking.
	To establish that $\rank$ is $\perm$-unimprovable,
	consider any $\perm$-reachable ranking $\mathord{\rankp} \neq \mathord{\rank}$; we must show that $\rankp$ is not more aligned with $\pref$ than $\rank$.

	Since $\mathord{\rankp} \neq \mathord{\rank}$,
	there are alternatives $x,y \in \mathcal{X}$ such that $x \rankp y$ and $y \rank x$;	
	furthermore, we may choose these to be $\rankp$-adjacent
	(i.e. $x \rankp z \rankp y$ holds for no $z \in \mathcal{X}$).
	Since $\rankp$ is $\perm$-reachable, we must have $x \perm y$. (This follows from \Cref{observation:reachability_adjacency} in \cref{app:extra:reachability_adjacency} (\cpageref{observation:reachability_adjacency}).)
	It must be that $y \pref x$,
	because otherwise the $\perm$-efficiency of $\rank$ would require that $x \rank y$.
	Thus the pair $x,y$ is ranked `right' by $\rank$ ($y \pref x$ and $y \rank x$) and `wrong' by $\rankp$ ($x \rankp y$),
	so that $\rankp$ fails to be more aligned with $\pref$ than $\rank$.
\end{proof}


\subsection{Insertion sort, a regret-free strategy}
\label{sec:is:is}

\begin{namedthm}[Insertion sort.]
	\label{algorithm:is}
	Label the alternatives $\mathcal{X} \equiv \{1,\dots,n\}$ so that $1 \pref \cdots \pref n$.
	We proceed in rounds indexed by $k \in \{n-1,\dots,2,1\}$,
	starting with $k=n-1$:
	\begin{itemize}

		\item
		By the $k^\text{th}$ round,
		those alternatives which the chair considers strictly worse than alternative $k$
		(i.e. the alternatives $\{k+1,\dots,n\}$)
		have already been totally ranked.
		Now `insert' alternative $k$ into that ranking:

		\item First, pit alternative $k$ against
		the highest-ranked worse alternative.

		If $k$ won, then $\{k,\dots,n\}$ are now totally ranked (with $k$ on top).

		\item If $k$ lost,
		pit $k$ against the second-highest-ranked worse alternative.

		If $k$ won, then $\{k,\dots,n\}$ are now totally ranked (with $k$ second).

		\item If $k$ lost again,
		pit $k$ against the third-highest-ranked worse alternative.

		\dots

		\item
		Now $\{k,\dots,n\}$ are totally ranked.
		Decrease $k$ by $1$ and repeat.%
			\footnote{\label{footnote:is_class}%
			A detail: our definition of strategies requires them to specify behaviour after all histories, even those that cannot arise.
			(E.g. a strategy that first offers $x,y$ must still specify behaviour after histories where only $y,z$ are ranked.)
			Thus formally,
			insertion sort defines not a strategy,
			but rather an equivalence class of strategies having the same path.}

	\end{itemize}
\end{namedthm}

Insertion sort is illustrated in \Cref{fig:is} for the case of three alternatives.
\begin{figure}
	\centering
	\begin{tikzpicture}[yscale=2,xscale=1.4,every text node part/.style={align=center}]

		\draw[-] ( 0,-0) -- (-2,-1);
		\draw[-] ( 0,-0) -- ( 2,-1);

		\draw[-] (-2,-1) -- (-1,-2);
		\draw[-] ( 2,-1) -- ( 3,-2);
		\draw[-] (-2,-1) -- (-2.7,-1.7);
		\draw[-] ( 2,-1) -- ( 1.3,-1.7);

		\draw[-] (-1,-2) -- (-1.7,-2.7);
		\draw[-] (-1,-2) -- (-0.3,-2.7);
		\draw[-] ( 3,-2) -- ( 2.3,-2.7);
		\draw[-] ( 3,-2) -- ( 3.7,-2.7);

		\node[fill=white,draw,rounded corners,dotted] at (0,-0) {$\mathord{\rank} = \varnothing$ \\ $\hookrightarrow$ offer $2,3$};

		\node[fill=white,draw,rounded corners,dotted] at (-2,-1)
			{$2 \rank 3$ \\ $\hookrightarrow$ offer $1,2$};
		\node[fill=white,draw,rounded corners,dotted] at ( 2,-1)
			{$3 \rank 2$ \\ $\hookrightarrow$ offer $1,3$};

		\node[fill=white,draw,rounded corners,dotted] at (-1,-2)
			{$2 \rank 1$ \& $2 \rank 3$ \\ $\hookrightarrow$ offer $1,3$};
		\node[fill=white,draw,rounded corners,dotted] at ( 3,-2)
			{$3 \rank 1$ \& $3 \rank 2$ \\ $\hookrightarrow$ offer $1,2$};

		\node[gray,anchor=east] at (-0.9,-0.45) {$2 \perm 3$\phantom{--}};
		\node[gray,anchor=west] at ( 0.9,-0.45) {\phantom{--}$3 \perm 2$};

		\node[gray,anchor=east] at (-2.55,-1.55) {$1 \perm 2$\phantom{-}};
		\node[gray,anchor=west] at (-1.50,-1.50) {\phantom{-}$2 \perm 1$};
		\node[gray,anchor=east] at ( 1.45,-1.55) {$1 \perm 3$\phantom{-}};
		\node[gray,anchor=west] at ( 2.50,-1.50) {\phantom{-}$3 \perm 1$};

		\node[gray,anchor=east] at (-1.55,-2.55) {$1 \perm 3$\phantom{-}};
		\node[gray,anchor=west] at (-0.45,-2.55) {\phantom{-}$3 \perm 1$};
		\node[gray,anchor=east] at ( 2.45,-2.55) {$1 \perm 2$\phantom{-}};
		\node[gray,anchor=west] at ( 3.55,-2.55) {\phantom{-}$2 \perm 1$};

	\end{tikzpicture}
	\caption{Insertion sort with three alternatives $\mathcal{X} = \{1,2,3\}$, where the chair's preference is $1 \pref 2 \pref 3$.
	$\rank$ denotes the (evolving) proto-ranking.
	The path taken depends on the majority will $\perm$.}
	\label{fig:is}
\end{figure}%

\begin{theorem}
	\label{theorem:is_efficient}
	Insertion sort is efficient, hence regret-free.
\end{theorem}

\begin{proof}
	Fix a majority will $\perm$, and let $\rank$ be the outcome of insertion sort under $\perm$; we must show that $\rank$ is $\perm$-efficient.
	To that end, fix alternatives $x,y \in \mathcal{X}$ with $x \pref y$ and $x \perm y$; we shall prove that $x \rank y$.

	Enumerate as $z_1 \rank \cdots \rank z_K$ all of the alternatives that are strictly $\pref$-worse than $x$, and note that $z_k = y$ for some $k \leq K$.
	By definition of insertion sort, $x$ will be pitted against $z_1,z_2,\dots$ in turn until it wins a vote.
	If $x$ loses against $z_1,\dots,z_{k-1}$, then it is next pitted against $z_k=y$, and wins since $x \perm y$ by hypothesis; thus $x \rank y$.
	If instead $x$ wins against some $z_\ell$ with $\ell < k$, then $x \rank z_\ell \rank \cdots \rank z_k = y$, so $x \rank y$ by transitivity of $\rank$.
\end{proof}

Insertion sort is a standard sorting algorithm---see e.g. \textcite[§5.2.1]{Knuth1998}.
In the sorting problem,
there is a \emph{transitive} and asymmetric relation $\perm$,
and one seeks to reach a `sorted' ranking $\rank$,
meaning that $x \rank y$ whenever $x \perm y$,
by making (as few as possible) binary comparisons using $\perm$.

\section{Properties of regret-free strategies}
\label{sec:charac}

Insertion sort is regret-free, but it is not the only such strategy.
Which of its characteristics are essential for regret-freeness, and which are extraneous?

We give two answers in this section.
Our first characterisation, \Cref{theorem:regretfree_efficient} (§\ref{sec:charac:regretfree_efficient}), is in terms of outcomes: the regret-free strategies are exactly the efficient ones, i.e. those reaching a $\perm$-efficient ranking under every majority will $\perm$.
We show this characterisation to be tight (\Cref{proposition:regretfree_efficient_tight}).

The second characterisation, \Cref{theorem:regretfree_errors} (§\ref{sec:charac:regretfree_errors}), is in terms of behaviour: regret-freeness requires precisely that two intuitive errors be avoided.
We show in addition that the advice offered by \Cref{theorem:regretfree_errors} can be operationalised myopically: avoiding errors ensures that a non-error pair will always be available to be offered next (\Cref{proposition:regretfree_errors_tight}).

What distinguishes one regret-free strategy from another is prioritisation: \emph{which} pairs of alternatives are ranked `right' when not all of them can be.
We prove that \emph{lexicographic} prioritisation characterises insertion sort (\Cref{theorem:is_lexicographic}, §\ref{sec:charac:is2}):
the chair's favourite alternative is ranked as highly as possible,
her second-favourite alternative is ranked as highly as possible subject to that, and so on.
One consequence (§\ref{sec:charac:ra_is}) is that insertion sort is outcome-equivalent to recursively applying the much-studied \emph{amendment procedure.}

\subsection{Characterisation of regret-free strategies: outcomes}
\label{sec:charac:regretfree_efficient}

Efficiency is not only sufficient for regret-freeness,
but also necessary:

\begin{theorem}
	\label{theorem:regretfree_efficient}
	A strategy is regret-free iff it is efficient.
\end{theorem}

The proof, in \cref{app:pf_charac}, establishes \Cref{theorem:regretfree_efficient} jointly with \Cref{theorem:regretfree_errors} below (§\ref{sec:charac:regretfree_errors}).
The `if' direction was already established in \Cref{corollary:efficient_regret-free} (\cpageref{corollary:efficient_regret-free}).
The following gives a feel for the `only if' direction:

\setcounter{example}{1}
\begin{example}
	\label{example:1234}
	Consider alternatives $\mathcal{X} = \{\alpha,\beta,\gamma,\delta\}$ with $\alpha \pref \beta \pref \gamma \pref \delta$ and the following majority will $\perm$:%
		\footnote{Symbolically, $\alpha \perm \delta \perm \gamma \perm \beta \perm \alpha$, $\gamma \perm \alpha$ and $\beta \perm \delta$.}
	\begin{equation*}
		\vcenter{\hbox{%
		\begin{tikzpicture}[every text node part/.style={align=center}]

			\def \radius {1cm}
			\def \margin {22}
			\def \marginstr {\radius*0.35}

			\foreach \s in {1,2,3,4}
			{
				\draw[->,>=latex,thick] ({90 * (5-\s) + \margin}:\radius)
					arc ({90 * (5-\s) + \margin}:{90 * (6-\s)-\margin}:\radius);
			}

			\node at ({90 * (6-1)}:\radius) {$\alpha$};
			\node at ({90 * (6-2)}:\radius) {$\beta$};
			\node at ({90 * (6-3)}:\radius) {$\gamma$};
			\node at ({90 * (6-4)}:\radius) {$\delta$};

			\draw[<-,>=latex,thick] (0,\radius-\marginstr) -- (0,-\radius+\marginstr);
			\draw[->,>=latex,thick] (\radius-\marginstr,0) -- (-\radius+\marginstr,0);

		\end{tikzpicture}%
		}}
	\end{equation*}
	There is exactly one $\perm$-reachable ranking that ranks $\alpha$ above $\beta$, namely $\alpha \rank \delta \rank \gamma \rank \beta$.%
		\footnote{This can be verified directly. Alternatively, since there is exactly one $\perm$-path from $\alpha$ to $\beta$ (namely $(\alpha,\delta,\gamma,\beta)$), it follows by \Cref{observation:reachability_adjacency} in \cref{app:extra:reachability_adjacency} (\cpageref{observation:reachability_adjacency}).}
	Since no other $\perm$-reachable ranking ranks $\alpha$ above $\beta$, $\rank$ is $\perm$-unimprovable.
	But $\rank$ fails to be $\perm$-efficient, for it ranks $\delta$ above $\beta$ despite the fact that $\beta \perm \delta$.
	(Thus the converse of \Cref{lemma:Wefficient_Wunimprovable} (\cpageref{lemma:Wefficient_Wunimprovable}) is false.)

	Let $\strat$ be a strategy that
	(i) first offers $\beta,\gamma$,
	(ii) then, in case $\gamma$ won, offers $\gamma,\delta$, and
	(iii) next, in case $\gamma$ and then $\delta$ won, offers $\alpha,\delta$.
	This strategy has outcome $\rank$ under $\perm$, so fails to be efficient.
	To see that it is not regret-free, consider a majority will $\permp$ that differs from $\perm$ only in that $\delta \permp \alpha$.
	The outcome of $\strat$ under $\permp$ is $\delta \rankp \gamma \rankp \beta \rankp \alpha$.%
		\footnote{First $\delta \rankp \gamma \rankp \beta$ and $\delta \rankp \alpha$ are determined. Then $\alpha,\beta$ and $\alpha,\gamma$ are offered (in some order that doesn't matter), and $\alpha$ loses in both votes.}
	This ranking fails to be $\permp$-unimprovable since $\gamma \rankpp \beta \rankpp \delta \rankpp \alpha$ is $\permp$-reachable and is more aligned with $\pref$ by inspection.
	Thus $\strat$ fails to be regret-free.
\end{example}

The broader insight underlying the `only if' part of \Cref{theorem:regretfree_efficient} is that reaching a non-$\perm$-efficient outcome necessarily involves a sacrifice. (In the example, $\strat$ forgoes the opportunity to rank $\beta$ above $\delta$.)
For some majority wills, such as $\perm$, the sacrifice pays off. (It allows $\alpha$ to be ranked above $\beta$, something that no $\perm$-reachable $\perm$-efficient ranking achieves.)
But for other majority wills, such as $\permp$, no reward materialises, yielding a non-$\permp$-unimprovable outcome.

The characterisation in \Cref{theorem:regretfree_efficient} is tight, in the following sense:

\begin{proposition}
	\label{proposition:regretfree_efficient_tight}
	For any majority will $\perm$ and any $\perm$-reachable $\perm$-efficient ranking $\rank$, some regret-free strategy has outcome $\rank$ under $\perm$.
\end{proposition}

Thus no statement sharper than \Cref{theorem:regretfree_efficient} can be made about the outcomes under $\perm$ of regret-free strategies: every $\perm$-reachable $\perm$-efficient ranking is admissible.
We give the proof of \Cref{proposition:regretfree_efficient_tight} in \cref{suppl:pf_tightness}.
A non-trivial argument is required because a given $\perm$-reachable $\perm$-efficient ranking can be reached in many ways, not all of which form part of a regret-free strategy.%
	\footnote{To see why, return to \Cref{example:123_Wreachable} (\cpageref{example:123_Wreachable,example:123_Wefficient}).
	The $\perm$-efficient ranking $\alpha \rankp \gamma \rankp \beta$ may be reached by offering $\alpha,\gamma$ and then $\gamma,\beta$.
	But a strategy that does this cannot be regret-free, because it has outcome $\rankp$ also under the majority will $\alpha \permp \beta \permp \gamma \permp \alpha$,
	and $\rankp$ is not a $\permp$-unimprovable ranking (since the more aligned ranking $\pref$ is $\permp$-reachable).}

\subsection{Characterisation of regret-free strategies: behaviour}
\label{sec:charac:regretfree_errors}

To understand the behavioural content of regret-freeness,
we begin with a simple intuition for why insertion sort
is regret-free
in the case of three alternatives (refer to \Cref{fig:is} on \cpageref{fig:is}).
Suppose first that the initial vote on $2,3$ went well ($2 \perm 3$).
Offering $1,2$ next is a good idea because it affords an opportunity of a favourable imposition of transitivity: if the chair gets lucky ($1 \perm 2$), then she obtains $1 \rank 3$ `for free' from $1 \rank 2$ and $2 \rank 3$.
(Even though it may be that $3 \perm 1$.)
Offering $1,3$ would miss this opportunity.

Suppose instead that the initial vote on $2,3$ went badly ($3 \perm 2$).
Offering $1,2$ next would risk an unfavourable imposition of transitivity:
were the vote to go against her ($2 \perm 1$), then $3$ would be ranked above $1$ since $3 \rank 2$ and $2 \rank 1$.
(Even though it may be that $1 \perm 3$.)

To summarise, an intuition for why insertion sort is regret-free is that it never (1) misses an opportunity for a favourable imposition of transitivity nor (2) risks an unfavourable imposition of transitivity.
Our second characterisation formalises this intuition.

\begin{definition}
	\label{definition:errors}
	Let $\rank$ be a non-total proto-ranking, and let $x \pref y$ be unranked alternatives. (I.e. $x,y \in \mathcal{X}$ such that $x \nrank y \nrank x$.)
	\begin{enumerate}

		\item Offering $x,y$ for a vote \emph{misses an opportunity (at $\rank$)} iff there is an alternative $z \in \mathcal{X}$ such that $x \pref z \pref y$ and $y \nrank z \nrank x$.

		\item Offering $x,y$ for a vote \emph{takes a risk (at $\rank$)} iff there is an alternative $z \in \mathcal{X}$ such that either
		\begin{enumerate}

			\item $z \pref y$, $x \rank z$ and $y \nrank z$, or

			\item $x \pref z$, $z \rank y$ and $z \nrank x$.

		\end{enumerate}

	\end{enumerate}
\end{definition}

The `missed opportunity' in (1) is that $x \rank y$ (the hoped-for outcome when offering $x,y$ for a vote) could potentially have been obtained `for free' by offering votes on $x,z$ and $z,y$, via a `favourable imposition of transitivity'.
The `risk' in (2)(a) is that if the vote on $x,y$ were to go badly (so that $y \rank x$), then $y \rank z$ would follow---an `unfavourable imposition of transitivity'.
Similarly for (2)(b).

We say that a strategy \emph{never misses an opportunity} (\emph{never takes a risk}) iff it does not miss an opportunity (take a risk) on the path.

\begin{theorem}
	\label{theorem:regretfree_errors}
	A strategy is regret-free iff it never misses an opportunity or takes a risk.
\end{theorem}

We give a joint proof of \Cref{theorem:regretfree_efficient,theorem:regretfree_errors} in \cref{app:pf_charac}.
The argument is illustrated in \Cref{fig:charac}.
\begin{figure}
	\centering
	\begin{tikzpicture}[every text node part/.style={align=center}]

		\def \radius {2cm}
		\def \margin {25}

		\node at ({90 + 120 * 0}:\radius) {efficient};
		\node at ({90 + 120 * 1}:\radius) {avoids \\ errors};
		\node at ({90 + 120 * 2}:\radius) {regret- \\ free};

		\foreach \s in {1,2,3}
		{
			\draw[<-,>=latex,thick] ({90 + 120 * (\s - 1) + \margin}:\radius)
				arc ({90 + 120 * (\s - 1) + \margin}:{90 + 120 * (\s)-\margin}:\radius);
		}

		\node[gray,anchor=south east] at ({90 + 120 * 0.5}:\radius)
			{\phantom{(p. 15) }\cref{app:pf_charac:noerror_efficient}};
		\node[gray,anchor=north     ] at ({90 + 120 * 1.5}:\radius)
			{\cref{app:pf_charac:unimp_noerror}};
		\node[gray,anchor=south west] at ({90 - 120 * 0.5}:\radius)
			{\Cref{corollary:efficient_regret-free} (\cpageref{corollary:efficient_regret-free})};

	\end{tikzpicture}
	\caption{Proof of \Cref{theorem:regretfree_efficient,theorem:regretfree_errors}.}
	\label{fig:charac}
\end{figure}%

The `only if' part of \Cref{theorem:regretfree_errors} asserts that opportunity-missing and risk-taking are errors, in the sense that committing one of them will lead to a non-$\perm$-unimprovable ranking under some $\perm$.
This is intuitive, but requires some work to show.
To appreciate why, suppose that $\strat$ offers a pair $x \pref y$ under $\perm$, and that doing so either misses an opportunity or takes a risk.
It is then easy to find another majority will $\permp$ such that the outcome $\rank$ of $\strat$ under $\permp$ ranks $y$ above $x$,
whereas some other $\permp$-reachable ranking $\rankp$ ranks $x \rankp y$.
But that is not enough:
if $\rankp$ is to be more aligned with $\pref$ than $\rank$,
then \emph{every} pair $z,w$ ranked `right' by $\rank$ ($z \pref w$ and $z \rank w$) must also be ranked `right' by $\rankp$ ($z \rankp w$).
We construct $\permp$ and $\rankp$ with these properties in \cref{app:pf_charac:unimp_noerror}.

The `if' part of \Cref{theorem:regretfree_errors} asserts that these are the \emph{only} errors.
Thus to achieve regret-freeness, it suffices to avoid missing opportunities and taking risks, separately after each history.
The proof in \cref{app:pf_charac:noerror_efficient} begins with an arbitrary non-efficient strategy $\strat$, meaning one whose outcome $\rank$ under some majority will $\perm$ ranks some pair $x,y \in \mathcal{X}$ as $y \rank x$ even though $x \pref y$ and $x \perm y$.
The pair $x,y$ cannot have been voted on (else $x$ would have prevailed), so must have been ranked by an imposition of transitivity.
We show that avoiding the two errors suffices to preclude such `unfavourable' impositions of transitivity; thus $\strat$ must have committed the one or the other.

\Cref{theorem:regretfree_errors} tells us that non-error pairs are attractive, but it does not promise that they exist.
In particular, there could conceivably be histories along which the chair has committed no errors, but is now forced to do so because no unranked pairs remain that can be offered without missing an opportunity or taking a risk.
The following rules out this scenario:

\begin{proposition}
	\label{proposition:regretfree_errors_tight}
	After any history along which the chair has not missed an opportunity or taken a risk,
	if some pairs remain unranked,
	then there is a pair which can be offered without missing an opportunity or taking a risk.
\end{proposition}

The proof is in \cref{suppl:pf_tightness}.
\Cref{proposition:regretfree_errors_tight} shows that the characterisation in \Cref{theorem:regretfree_errors} of regret-free behaviour is tight:
	for any $\perm$ and sequence of pairs that is error-free under $\perm$,
	some regret-free strategy offers this sequence under $\perm$.%
		\footnote{In particular, a strategy can be constructed which offers this sequence under $\perm$ and commits no errors under other majority wills.
		Such a strategy is regret-free by \Cref{theorem:regretfree_errors}.}

\Cref{theorem:regretfree_errors}, augmented by \Cref{proposition:regretfree_errors_tight}, allows us to give the following simple `myopic' advice to the chair.
After each history, inspect the current proto-ranking to identify an unranked pair of alternatives that would not miss an opportunity or take a risk; that such a pair exists is guaranteed by \Cref{proposition:regretfree_errors_tight}.
Offer any such pair for a vote.
By \Cref{theorem:regretfree_errors}, the outcome will be $\perm$-unimprovable whatever the majority will $\perm$.

\subsection{Priorities, and a characterisation of insertion sort}
\label{sec:charac:is2}

We now show that regret-free strategies differ (only) in how they \emph{prioritise,} and that insertion sort is characterised by its \emph{lexicographic} prioritisation.

Given a majority will $\perm$,
call a pair $x \pref y$ an \emph{agreement pair} iff $x \perm y$,
and a \emph{disagreement pair} otherwise.
All regret-free strategies rank the agreement pairs `right' (i.e. $x \rank y$)
by \Cref{theorem:regretfree_efficient}.
A disagreement pair
is ranked `wrong' (as $y \rank x$) if voted on,
and ranked `right' if instead ranked by transitivity,
since regret-free strategies permit only \emph{favourable} impositions of transitivity by \Cref{theorem:regretfree_errors}.

What distinguishes regret-free strategies is thus precisely
which disagreement pairs are voted on and which are held back.
Held-back pairs have more opportunities to be ranked `right' by transitivity.

\setcounter{example}{0}
\begin{example}[continued]
	\label{example:123_priorities}
	Recall the details from \cpageref{example:123_Wreachable,example:123_maw}.
	The disagreement pairs are $\alpha,\beta$ and $\beta,\gamma$.
	Insertion sort prioritises the former:
	its outcome is $\alpha \rankp \gamma \rankp \beta$.
	Other regret-free strategies prioritise $\beta,\gamma$ instead,
	having outcome $\beta \rank \alpha \rank \gamma$.%
		\footnote{Some regret-free strategy has outcome $\rank$ under $\perm$
		by \Cref{proposition:regretfree_efficient_tight} (p. \pageref{proposition:regretfree_efficient_tight}).
		An example is `reverse insertion sort':
		in its $k^\text{th}$ round, the chair's top $k-1$ alternatives have already been ranked,
		and now her $k^\text{th}$-favourite is pitted in turn against the lowest-ranked of these $\pref$-better alternatives,
		against the second-lowest ranked,
		and so on until it loses a vote.}
\end{example}

To understand the priorities of insertion sort,
label the alternatives $\mathcal{X} \equiv \{1,\dots,n\}$ so that $1 \pref \cdots \pref n$.
Insertion sort leaves alternative $1$ for last,
by totally ranking the alternatives $\{2,\dots,n\}$ before offering any votes involving $1$.
This intuitively maximises opportunities for favourable impositions of transitivity involving $1$,
and so optimises the placement of this alternative.

Apart from alternative $1$, alternative $2$ is left for last by insertion sort,
suggesting that this alternative is placed optimally, subject to optimal placement of alternative $1$.
This suggests that insertion sort prioritises lexicographically,
by placing alternative $1$ as well as possible,
and subject to that placing alternative $2$ as well as possible,
and so on.
We shall formalise this intuition.

A simple cardinal measure of how well a strategy $\strat$
places a given alternative $x \in \mathcal{X}$
is \emph{how many} $\pref$-worse alternatives $y$ are ranked below $x$, i.e.
\begin{equation*}
	N^\strat_x(\mathord{\perm})
	\coloneqq \abs*{ \left\{ y \in \mathcal{X} :
	\text{$x \pref y$ and $x \rankfn{\strat}{\perm} y$}
	\right\} } ,
\end{equation*}
where $\rankfn{\strat}{\perm}$ denotes the outcome of $\strat$ under majority will $\perm$.
We call a strategy $\strat$ \emph{better for $x$} than another strategy $\stratp$ iff
$N^\strat_x$ first-order stochastically dominates $N^{\stratp}_x$, i.e.
\begin{equation*}
	\abs*{ \left\{ \mathord{\perm} :
	N^\strat_x(\mathord{\perm}) \geq k
	\right\} }
	\geq \abs{ \{ \mathord{\perm} :
	N^{\stratp}_x(\mathord{\perm}) \geq k
	\} } 
	\quad \text{for each $k \in \{1,\dots,n-1\}$.}
\end{equation*}
If a strategy $\strat$ belongs to a set $\Sigma$ of strategies
and is better for $x$ than any other $\stratp \in \Sigma$,
then we say that $\strat$ is \emph{best for alternative $x$ among $\Sigma$.}

Call a strategy \emph{lexicographic} iff
among all strategies, it is best for the chair's favourite alternative;
among such strategies, it is best for her second-favourite;
among \emph{such} strategies, it is best for her third-favourite;
etc.
Say that two strategies are \emph{outcome-equivalent} iff
they have the same outcome under $\perm$,
for every majority will $\perm$.

\begin{theorem}
	\label{theorem:is_lexicographic}
	A strategy is outcome-equivalent to insertion sort iff it is lexicographic.
\end{theorem}

The proof is in \cref{app:pf_proposition:is2}.

\subsection{Insertion sort and the amendment procedure}
\label{sec:charac:ra_is}

The following procedure figures prominently in the literature (see §\ref{sec:intro:lit}):

\begin{namedthm}[Amendment procedure.]
	\label{algorithm:amendment}
	Label the alternatives $\mathcal{X} \equiv \{1,\dots,n\}$ so that $1 \pref \cdots \pref n$.
	First, pit $n-1$ against $n$.
	Then pit $n-2$ against the winner.
	Then pit $n-3$ against the previous round's winner.
	And so on.
	Call the winner of the final round the \emph{final winner.}
\end{namedthm}

The amendment procedure is designed to choose a single alternative.
In particular, the final winner is ranked top, since any other alternative either lost a vote to it, or lost a vote to an alternative that lost to it, etc.
The natural way to obtain a ranking is to apply amendment recursively:

\begin{namedthm}[Recursive amendment.]
	\label{algorithm:recursive_amendment}
	Label the alternatives $\mathcal{X} \equiv \{1,\dots,n\}$ so that $1 \pref \cdots \pref n$.
	First, run the amendment procedure on $\mathcal{X}$, and write $y_1$ for the final winner.
	Next, run the amendment procedure on $\mathcal{X} \setminus \{y_1\}$, writing $y_2$ for the final winner.%
		\footnote{If the amendment procedure demands a vote on a pair $x,y$ that is already ranked as (wlog) $x \rank y$,
		treat this as offering $x,y$ and obtaining outcome $x \perm y$.}
	Then run the amendment procedure on $\mathcal{X} \setminus \{y_1,y_2\}$ to obtain a final winner $y_3$.
	And so on.
	The resulting ranking is $y_1 \rank \cdots \rank y_{n-1} \rank y_n$, where $y_n$ denotes the unique element of $\mathcal{X} \setminus \{y_1,\dots,y_{n-1}\}$.%
		\footnote{Like insertion sort,
		recursive amendment describes only a path,
		so formally defines an equivalence class of strategies;
		refer to \cref{footnote:is_class} on \cpageref{footnote:is_class}.}
\end{namedthm}

Recursive amendment is distinct from insertion sort.%
	\footnote{E.g. if $\mathcal{X} = \{1,2,3,4\}$ with $1 \pref 2 \pref 3 \pref 4$,
	$1 \perm 4 \perm 3$ and $4 \perm 2 \perm 3$,
	then insertion sort offers $3,4$, then $2,4$, then $2,3$, then $1,4$,
	whereas recursive amendment offers $3,4$, then $2,4$, then $1,4$, then $2,3$.
	The two strategies \emph{do} coincide when there are three alternatives.}
Nonetheless:

\begin{proposition}
	\label{proposition:is_amendment}
	Recursive amendment is outcome-equivalent to insertion sort.
\end{proposition}

This follows from one of the lemmata in our proof of \Cref{theorem:is_lexicographic}:

\begin{proof}
	Recall from \cref{app:pf_proposition:is2:overview} (\cpageref{definition:Sigmak}) the set $\Sigma_{n-2}$ of strategies.
	Clearly recursive amendment belongs to $\Sigma_{n-2}$.
	Thus by \Cref{lemma:is_oe_charac} (\cpageref{lemma:is_oe_charac}),
	recursive amendment is outcome-equivalent to insertion sort.
\end{proof}

In the sorting context,
recursive amendment is called \emph{selection sort---}see e.g. \textcite[§5.2.3]{Knuth1998}.

\section{Concluding remarks}
\label{sec:concl}

We close by considering some alternative interpretations of our model and of our results,
as well as a few extensions.

\subsection{Broader interpretations}
\label{sec:concl:interp}

The model admits interpretations in which $\perm$ arises from something other than majority voting by a committee.
It could be the (expressed) preference of a single individual, for example:
in this case, cycles reflect inconsistencies in his judgment,
and a crafty advisor (`the chair') exploits these by asking him to make pairwise comparisons in a well-chosen order.

In practice, the chair may be more constrained in her choice of agenda-setting strategy.
In that case, our analysis provides an upper bound on how well she can do,
and gives a qualitative indication of what incremental adjustments to agenda-setting might most benefit her.

Finally, there need not literally be a committee chair:
more generally, our results speak to how institutional or procedural design can influence collective decisions,
in the direction of a(ny) given priority $\pref$.

\subsection{Extensions}
\label{sec:concl:extensions}

Varying the committee's rules merely alters $\perm$,
and so our analysis continues to apply.
For example, we can accommodate a super-majority voting rule,
an even electorate $I \in \N$, and abstentions.
To determine $\perm$, we need only specify which alternative wins in case of an indecisive vote, e.g. by assuming that there is a status quo ranking that prevails in such cases.

A more substantial variation is to permit the chair (sometimes) to decide how two alternatives are to be ranked following an indecisive vote on them.
This happens if the chair has a vote, for example.
We extend all of our results to allow for this in \cref{suppl:indecisive}.

Our definition of a strategy rules out conditioning on who voted how in the past, since a history records only which alternative in each pair wins (rather than the full vote tally).%
	\footnote{Of course, this is only a restriction if the chair can observe individual votes.}
This is merely to avoid uninteresting complications:
we show in \cref{suppl:non-simple} that our results hold for `extended strategies' that can condition on past votes cast.



\begin{appendices}

\renewcommand*{\thesubsection}{\Alph{subsection}}

\crefalias{section}{appsec}
\crefalias{subsection}{appsec}
\crefalias{subsubsection}{appsec}
\section*{Appendices}
\label{app}
\addcontentsline{toc}{section}{Appendices}

\subsection{Standard definitions}
\label{app:st_definitions}

This appendix collects the definitions of order-theoretic concepts used in this paper.
Let $\mathcal{A}$ be a non-empty set, and $\rel$ a binary relation on it.
Recall that $\rel$ is formally a subset of $\mathcal{A} \times \mathcal{A}$, and that `$a \rel b$' is shorthand for `$(a,b) \in \mathord{\rel}$'.
For $a,b \in \mathcal{A}$,
we write $a \nrel b$ iff it is not the case that $a \rel b$.
For $a,b \in \mathcal{A}$ such that $a \rel b$, we denote by $[a,b]_{\mathord{\rel}}$ the \emph{$\rel$-order interval}
\begin{equation*}
	[a,b]_{\mathord{\rel}}
	\coloneqq \{a,b\} \union \{c \in \mathcal{A}: a \rel c \rel b \} .
\end{equation*}
Distinct $a,b \in \mathcal{A}$ with $a \rel b$
are \emph{$\rel$-adjacent} iff $[a,b]_{\mathord{\rel}} = \{a,b\}$.

$\rel$ is
\emph{reflexive} (\emph{irreflexive}) iff $a \rel a$ ($a \nrel a$) for every $a \in \mathcal{A}$,
\emph{total} iff $a \rel b$ or $b \rel a$ for any distinct $a,b \in \mathcal{A}$,
\emph{complete} iff it is reflexive and total,
\emph{asymmetric} iff $a \rel b$ implies $b \nrel a$ for $a,b \in \mathcal{A}$, and
\emph{transitive} iff $a \rel b \rel c$ implies $a \rel c$ for $a,b,c \in \mathcal{A}$.

The \emph{transitive closure} of $\rel$ is the smallest (in the sense of set inclusion) transitive relation that contains $\rel$.%
	\footnote{Every relation possesses a transitive closure because the maximal relation $\mathcal{A} \times \mathcal{A}$ is transitive and the intersection of transitive relations is transitive.}
The \emph{strict part} of $\rel$ is the binary relation
$\str \mathord{\rel}$ such that $a \mathrel{\str \mathord{\rel}} b$ iff $a \rel b$ and $b \nrel a$.
For two binary relations $\rel$ and $\relp$ on $\mathcal{A}$, $\relp$ is an \emph{extension} of $\rel$ iff both $\mathord{\rel} \subseteq \mathord{\relp}$ and $\str \mathord{\rel} \subseteq \str \mathord{\relp}$.

\subsection{Additional material for §\ref{sec:environment}}
\label{app:extra}

This appendix complements the exposition of the environment in §\ref{sec:environment}.
We
show that all and only total and asymmetric relations can be majority wills (§\ref{app:extra:arise_charac}),
formally define histories, strategies, outcomes and paths (§\ref{app:extra:defns}), 
and provide characterisations of $\perm$-reachability (§\ref{app:extra:reachability_adjacency}) and of `more aligned with than' (§\ref{app:extra:maw_charac}).

\subsubsection{Which binary relations are majority wills? (§\ref{sec:environment:general_will})}
\label{app:extra:arise_charac}

In this appendix, we prove that all and only total and asymmetric relations are legitimate majority wills.
A \emph{voting behaviour} $\votei$ specifies for each pair $x,y \in \mathcal{X}$ whether voter $i$ will vote for $x$ ($x \votei y$) or for $y$ ($y \votei x$).
Formally:

\begin{definition}
	\label{definition:voting_profile}
	A \emph{voting behaviour} is a total and asymmetric relation.
	A \emph{voting profile} is a collection $( \mathord{\votei} )_{i=1}^I$ of voting behaviours.
\end{definition}

The majority will of a voting profile $( \mathord{\votei} )_{i=1}^I$ is the relation $\perm$ such that $x \perm y$ iff $x \votei y$ for a majority of voters $i \in \{1,\dots,I\}$.
The following shows that all (and only) total and asymmetric relations are the majority will of some voting profile $( \mathord{\votei} )_{i=1}^I$, even if we insist that each voter's behaviour $\votei$ be transitive (meaning that it can be rationalised as sincere):

\begin{fact}
	\label{fact:arise_charac}
	For a binary relation $\perm$ on $\mathcal{X}$, the following are equivalent:
	\begin{enumerate}

		\item \label{fact:arise_charac:totas}
		$\perm$ is total and asymmetric.

		\item \label{fact:arise_charac:voting}
		For every $I \in \N$ odd, $\perm$ is the majority will of some profile $(\mathord{\votei})_{i=1}^I$ of voting behaviours.

		\item \label{fact:arise_charac:sincere_voting}
		For some $I \in \N$ odd, $\perm$ is the majority will of some profile $(\mathord{\votei})_{i=1}^I$ of \emph{transitive} voting behaviours.

	\end{enumerate}
\end{fact}

\begin{proof}
	It is immediate that \ref{fact:arise_charac:voting} and \ref{fact:arise_charac:sincere_voting} (separately) imply \ref{fact:arise_charac:totas}.
	To see that \ref{fact:arise_charac:totas} implies \ref{fact:arise_charac:voting}, simply observe that a total and asymmetric relation $\perm$ is the majority will of the voting profile $(\mathord{\votei})_{i=1}^I = (\mathord{\perm})_{i=1}^I$ for any $I \in \N$.
	The fact that \ref{fact:arise_charac:totas} implies \ref{fact:arise_charac:sincere_voting} follows from McGarvey's (\citeyear{Mcgarvey1953}) theorem.
\end{proof}

\subsubsection{Formal definitions (§\ref{sec:environment:interaction}--§\ref{sec:environment:general_will})}
\label{app:extra:defns}

In this appendix, we formally define histories, strategies, outcomes and paths.
A history is a sequence $( (x_t,y_t) )_{t=1}^T$, where $\{x_t,y_t\}$ is the pair offered in period $t$, and $x_t$ is the winner (i.e. $x_t \perm y_t$):

\begin{definition}
	\label{definition:histories}
	There is exactly one \emph{history of length $0$} (the `empty history').
	A \emph{history of length $T \in \N$} is a sequence $( (x_t,y_t) )_{t=1}^T$ of ordered pairs of alternatives satisfying $x_t \neq y_t$ and $x_t \nranktone y_t \nranktone x_t$ for each $t \in \{1,\dots,T\}$, where
	$\mathord{\rankzero} \coloneqq \varnothing$
	and $\mathord{\rankt} \coloneqq \tr \left( \mathord{\ranktone} \union \{(x_t,y_t)\} \right)$ for each $t \in \{1,\dots,T\}$.
	(Here `$\tr$' denotes the transitive closure.)
	The history is \emph{terminal} iff $\mathord{\rankT}$ is total.
\end{definition}

Let $\mathcal{H}_T$ be the set of all non-terminal histories of length $T \geq 0$, and write $\mathcal{H} \coloneqq \Union_{T=0}^\infty \mathcal{H}_T$ for all non-terminal histories.
A strategy offers, after each non-terminal history, a pair of alternatives that are unranked at that history:

\begin{definition}
	\label{definition:strategies}
	A \emph{strategy of the chair} is a map $\mathord{\strat} : \mathcal{H} \to 2^{\mathcal{X}}$ such that for each non-terminal history $h \in \mathcal{H}$,
	we have $\mathord{\strat}(h) = \{x,y\}$ for some alternatives $x,y \in \mathcal{X}$
	such that $( h, (x,y) )$ and $( h, (y,x) )$ are histories.
\end{definition}

A strategy $\mathord{\strat}$ and a majority will $\perm$ generate a terminal history $( (x_t,y_t) )_{t=1}^T$ as follows:
for each $t \in \{1,\dots,T\}$,
$(x_t,y_t)$ is given by
\begin{equation*}
	\{x_t,y_t\} \coloneqq \mathord{\strat}\left( ((x_s,y_s))_{s=1}^{t-1} \right)
	\quad \text{and} \quad
	x_t \perm y_t .
\end{equation*}
This history is associated with a sequence of proto-rankings $(\mathord{\rankt})_{t=0}^T$, as outlined in \Cref{definition:histories} above.%
	\footnote{Namely, $\mathord{\rankzero} = \varnothing$ and $\mathord{\rankt} = \tr \left( \mathord{\ranktone} \union \{(x_t,y_t)\} \right) = \tr \left( \Union_{s=1}^t \{(x_s,y_s)\} \right)$ for $t \geq 1$.}

\begin{definition}
	\label{definition:outcome}
	The \emph{outcome} of a strategy $\mathord{\strat}$ under a majority will $\perm$ is the ranking $\rankT$ associated with the terminal history they generate.
\end{definition}

A strategy and a majority will also generate a sequence of non-terminal histories---namely, all truncations of their generated terminal history.
If a non-terminal history is generated by $\strat$ and some majority will, we say that it belongs to the \emph{path} of $\strat$.

\subsubsection{A characterisation of \texorpdfstring{$\boldsymbol{\perm}$-reachability}{W-reachability} (§\ref{sec:environment:W-reachability})}
\label{app:extra:reachability_adjacency}

This appendix contains a characterisation of reachability used in our proofs.

\begin{observation}
	\label{observation:reachability_adjacency}
	Let $\perm$ be a majority will, and $\rank$ a ranking.
	The following are equivalent:
	\begin{enumerate}

		\item \label{item:observation:reachability_adjacency:feas}
		$\rank$ is $\perm$-reachable.

		\item \label{item:observation:reachability_adjacency:adj}
		For any $\rank$-adjacent $x,y \in \mathcal{X}$ with $x \rank y$, we have $x \perm y$.

	\end{enumerate}
\end{observation}

Condition \ref{item:observation:reachability_adjacency:adj} admits a graph-theoretic interpretation. Think of $\perm$ as a directed graph with vertices $\mathcal{X}$ and a directed edge from $x$ to $y$ iff $x \perm y$ (as in \Cref{example:123_Wreachable} on \cpageref{example:123_Wreachable}), and think of a ranking $\rank$ as a sequence of alternatives: the highest-ranked, the second-highest-ranked, and so on.%
	\footnote{Formally: identify $\rank$ with the sequence $(x_k)_{k=1}^{\abs*{\mathcal{X}}}$ such that $x_1 \rank x_2 \rank \cdots \rank x_{\abs*{\mathcal{X}}}$.}
Condition \ref{item:observation:reachability_adjacency:adj} requires precisely that $\rank$ be a directed path in $\perm$.

\begin{proof}
	\emph{\ref{item:observation:reachability_adjacency:adj} implies \ref{item:observation:reachability_adjacency:feas}:}
	Let $\rank$ satisfy \ref{item:observation:reachability_adjacency:adj}.
	Then $\rank$ is the outcome under $\perm$ of any strategy that offers a vote on each $\rank$-adjacent pair of alternatives.

	\emph{\ref{item:observation:reachability_adjacency:feas} implies \ref{item:observation:reachability_adjacency:adj}:}
	Let $\rank$ be $\perm$-reachable, and let $x,y \in \mathcal{X}$ be $\rank$-adjacent with $x \rank y$.
	By $\perm$-reachability, there is a strategy $\strat$ whose outcome under $\perm$ is $\rank$.
	Along its induced history, it is determined that $x \rank y$.
	Since $x,y$ are $\rank$-adjacent, this cannot be via an imposition of transitivity.
	So it must occur in a vote on $\{x,y\}$, in which $x$ wins---thus $x \perm y$.
\end{proof}

\subsubsection{A characterisation of \texorpdfstring{`}{'}more aligned with than' (§\ref{sec:environment:preferences})}
\label{app:extra:maw_charac}

In this appendix, we provide a characterisation of `more aligned with than', and use it to prove the claims made in §\ref{sec:environment:preferences} (\cpageref{appl:hiring_maw}) about the applications.

\begin{observation}
	\label{observation:maw_charac}
	For rankings $\pref$, $\rank$ and $\rankp$, the following are equivalent:%
		\footnote{This is an instance of the Milgrom--Shannon (\citeyear{MilgromShannon1994}) comparative statics theorem:
		viewing $(\mathcal{X},\mathord{\pref})$ as an ordered set of actions and $\mathord{\rank},\mathord{\rankp}$ as (strict) preferences,
		\ref{item:maw_charac:maw} says that $\rank$ single-crossing dominates $\rankp$,
		and \ref{item:maw_charac:subs} says that $\rank$ chooses higher actions than $\rankp$ does.}
	\begin{enumerate}

		\item \label{item:maw_charac:maw}
		$\rank$ is more aligned with $\pref$ than $\rankp$.

		\item \label{item:maw_charac:subs}
		For every non-empty set $X \subseteq \mathcal{X}$, the $\rank$-highest alternative in $X$ is identical to or $\pref$-better than the $\rankp$-highest alternative in $X$.

	\end{enumerate}
\end{observation}

\begin{proof}
	\emph{\ref{item:maw_charac:maw} implies \ref{item:maw_charac:subs}:}
	We prove the contra-positive.
	Suppose that $\mathord{\rank},\mathord{\rankp}$ do not satisfy \ref{item:maw_charac:subs}, so that there is a non-empty $X \subseteq \mathcal{X}$ whose $\rank$-highest element $x$ is (strictly) $\pref$-worse than its $\rankp$-highest element $x'$.
	Then $\rankp$ ranks $x,x'$ `right' ($x' \pref x$ and $x' \rankp x$) whereas $\rank$ ranks them `wrong' ($x \rank x'$), so $\rank$ is not more aligned with $\pref$ than $\rankp$.

	\emph{\ref{item:maw_charac:subs} implies \ref{item:maw_charac:maw}:}
	We prove the contra-positive.
	Suppose that $\rank$ is not more aligned with $\pref$ than $\rankp$, so that there are alternatives $x,x' \in \mathcal{X}$ with $x' \pref x$, $x' \rankp x$ and $x \rank x'$.
	Then the $\rank$-highest alternative in $X \coloneqq \{x,x'\}$ is (strictly) $\pref$-worse than the $\rankp$-highest.
\end{proof}

\begin{namedthm}[Hiring {\normalfont(continued)}.]
	\label{appl:hiring_pf}
	We claimed that a more aligned ranking is precisely one that hires a weakly $\pref$-better candidate at every realisation of uncertainty.
	This follows immediately from \Cref{observation:maw_charac}.
\end{namedthm}

\begin{namedthm}[Party lists {\normalfont(continued)}.]
	\label{appl:party_pf}
	Enrich the model so that only a random subset $X \subseteq \mathcal{X}$ of candidates is available.%
		\footnote{The grand set $X=\mathcal{X}$ can occur with high probability, if desired.}
	We claim that $\rank$ is more aligned with $\pref$ than $\rankp$
	iff for every realisation $(K,X)$ and every $k \leq K$, the $k^\text{th}$ $\pref$-best candidate hired by $\rank$ is weakly $\pref$-better than the $k^\text{th}$ $\pref$-best hired by $\rankp$.

	To prove the `only if' part, let $\rank$ be more aligned with $\pref$ than $\rankp$, and fix an $X \subseteq \mathcal{X}$ and a $K$.
	Assume without loss of generality that $K \leq \abs*{X}$.
	Label the candidates $\{x_1,\dots,x_K\}$ hired by $\rank$ under $X$ so that $x_1 \rank \cdots \rank x_K$, and similarly write $x_1' \rankp \cdots \rankp x_K'$ for those hired by $\rankp$.
	We must show that $x_k \prefeq x_k'$ for every $k \in \{1,\dots,K\}$.
	We proceed by strong induction on $k$.
	The base case $k=1$ is immediate from \Cref{observation:maw_charac}.

	For the induction step, suppose for some $k \in \{2,\dots,K\}$ that $x_\ell \prefeq x_\ell'$ for every $\ell<k$.
	Define $Y \coloneqq X \setminus \{x_1,\dots,x_{k-1},x_1',\dots,x_{k-1}'\}$.
	We have $x_\ell \prefeq x_\ell' \pref x_k'$ for every $\ell<k$ by the induction hypothesis.
	It follows that $x_k' \in Y$, and hence that $x_k'$ is the $\rankp$-highest alternative in $Y$.
	If $x_k \in Y$, then $x_k$ is the $\rank$-highest alternative in $Y$, whence $x_k \prefeq x_k'$ by \Cref{observation:maw_charac}.
	If instead $x_k \notin Y$, then $x_k = x_\ell' \pref x_k'$ for some $\ell<k$.

	For the `if' part, we prove the contra-positive.
	Suppose that $\rank$ is not more aligned with $\pref$ than $\rankp$.
	Then by \Cref{observation:maw_charac}, there is a subset $X' \subseteq \mathcal{X}$ such that
	$\rank$ hires a strictly $\pref$-worse candidate than $\rankp$
	at the realisation $(K,X) = (1,X')$ of uncertainty.
\end{namedthm}

\subsection{Proof of \texorpdfstring{\Cref{theorem:regretfree_efficient,theorem:regretfree_errors}}{Theorems \ref{theorem:regretfree_efficient} and \ref{theorem:regretfree_errors}} (§\ref{sec:charac},
\texorpdfstring{\cpageref{theorem:regretfree_efficient,theorem:regretfree_errors}}{pp. \pageref{theorem:regretfree_efficient} and \pageref{theorem:regretfree_errors}})}
\label{app:pf_charac}

We prove \Cref{theorem:regretfree_efficient,theorem:regretfree_errors} jointly, in the manner depicted in \Cref{fig:charac} (\cpageref{fig:charac}).
We already showed that efficiency implies regret-freeness (\Cref{corollary:efficient_regret-free}, \cpageref{corollary:efficient_regret-free}).
We shall establish the other two parts in §\ref{app:pf_charac:noerror_efficient} and §\ref{app:pf_charac:unimp_noerror}.

\subsubsection{Proof that error-avoiding strategies are efficient}
\label{app:pf_charac:noerror_efficient}

The proof relies on two intermediate results, \Cref{lemma:noerror_nocry} and \Cref{corollary:nomissed_pres} below.
We first require an abstract fact about the transitive closure operation:

\begin{observation}
	\label{observation:tr_cl}
	Consider a proto-ranking $\rank$ and unranked alternatives $x,y \in \mathcal{X}$ (i.e. $x \nrank y \nrank x$).
	Let $\rankp$ be the transitive closure of $\mathord{\rank} \union \{(z,w)\}$, and suppose that $y \rankp x$.
	Then (a) either $y \rank z$ or $y=z$, and (b) either $w \rank x$ or $w=x$.
\end{observation}

\begin{proof}
	Since $\rankp$ is the transitive closure of $\mathord{\rank} \union \{(z,w)\}$ and $y \rankp x$,
	there must be a sequence $(z_k)_{k=1}^K$ of alternatives with
	$z_1 = y$, $z_K = x$ and
	\begin{equation*}
		(z_k,z_{k+1}) \in \mathord{\rank} \union \{(z,w)\}
		\quad \text{for every $k<K$.}
	\end{equation*}
	Since $y \nrank x$ and $\rank$ is transitive, we must have $(z_k,z_{k+1}) = (z,w)$ for some $k<K$.
	The result follows since $\rank$ is transitive.
\end{proof}

\begin{definition}
	\label{definition:missed_op}
	Let $\rank$ be a proto-ranking.
	An ordered pair of alternatives $(x,y) \in \mathcal{X}$ is a \emph{missed opportunity in $\rank$} iff
	$y \rank x$
	and there is an alternative $z \in \mathcal{X}$ such that $x \pref z \pref y$ and $y \nrank z \nrank x$.
\end{definition}

\begin{lemma}
	\label{lemma:noerror_nocry}
	Consider a proto-ranking $\rank$ that contains no missed opportunities.
	Let $x \pref y$ be alternatives with $y \nrank x$.
	Suppose that offering $\{z,w\}$ (where $z \nrank w \nrank z$) does not take a risk at $\rank$, and that doing this leads to a proto-ranking $\rankp$ such that $y \rankp x$.
	Then $\{z,w\} = \{x,y\}$.
\end{lemma}

\begin{proof}
	Let $\rank$, $x$, $y$, $z$, $w$ and $\rankp$ satisfy the hypothesis of the lemma, and assume wlog that $z \pref w$. We must show that $z=x$ and $w=y$.

	\begin{namedthm}[Claim.]
		\label{claim:noerror_nocry:zRw}
		$w \rankp z$.
	\end{namedthm}

	\begin{proof}[Proof of the \protect{\hyperref[claim:noerror_nocry:zRw]{claim}}]%
		\renewcommand{\qedsymbol}{$\square$}
		Suppose toward a contradiction that $w \nrankp z$.
		We will show that $\rank$ contains a missed opportunity.

		Since $\rankp$ is induced from $\rank$ by offering $\{z,w\}$, and $w \nrankp z$, it must be that $\rankp$ is the transitive closure of $\mathord{\rank} \union \{(z,w)\}$.
		Since $y \nrank x$ and $y \rankp x$, it follows by \Cref{observation:tr_cl} that (a) either $y \rank z$ or $y=z$, and (b) either $w \rank x$ or $w=x$.
		Now consider two cases.

		\emph{Case 1: $z \pref x$ or $z=x$.}
		We will show that $z \pref x \pref y$, $y \rank z$, and $y \nrank x \nrank z$, so that $(z,y)$ is a missed opportunity in $\rank$.
		Both $x \pref y$ and $y \nrank x$ hold by hypothesis.
		For $x \nrank z$, suppose to the contrary that $x \rank z$; then since $w \rank x$ or $w=x$ by property (b), we have $w \rank z$ by transitivity of $\rank$, a contradiction.

		To obtain $y \rank z$, observe that $z \pref y$ since by hypothesis $z \pref x$ or $z=x$, and we know that $x \pref y$ and that $\pref$ is transitive. Thus $z \neq y$, whence $y \rank z$ follows by property (a).
		To see that $z \pref x$, simply note that $z=x$ is impossible because $y \rank z$ and $y \nrank x$.

		\emph{Case 2: $x \pref z$.}
		We will show that $x \pref z \pref w$, $w \rank x$, and $w \nrank z \nrank x$, so that $(x,w)$ is a missed opportunity in $\rank$.
		Both $x \pref z \pref w$ and $w \nrank z$ hold by hypothesis.
		For $z \nrank x$, suppose to the contrary that $z \rank x$; then since $y \rank z$ or $y=z$ by property (a), we have $y \rank x$ by transitivity of $\rank$, a contradiction.

		To obtain $w \rank x$, observe that $x \pref w$ since $x \pref z \pref w$ and $\pref$ is transitive.
		Thus $w \neq x$, whence $w \rank x$ follows by property (b).
	\end{proof}%
	\renewcommand{\qedsymbol}{$\blacksquare$}

	In light of the \hyperref[claim:noerror_nocry:zRw]{claim}, $\rankp$ must be the transitive closure of $\mathord{\rank} \union \{(w,z)\}$.
	Since $y \nrank x$ and $y \rankp x$, applying \Cref{observation:tr_cl} yields that (a) either $y \rank w$ or $y=w$, and (b) either $z \rank x$ or $z=x$.

	We claim that
	\begin{equation}
		z \neq x \pref w
		\label{eq:lemma:noerror_nocry1}
	\end{equation}
	is impossible.
	Suppose toward a contradiction that \eqref{eq:lemma:noerror_nocry1} holds; we will show that offering $\{z,w\}$ takes a risk at $\rank$, i.e. that $x \pref w$, $z \rank x$ and $w \nrank x$.
	We have $x \pref w$ by \eqref{eq:lemma:noerror_nocry1}, and $z \rank x$ by \eqref{eq:lemma:noerror_nocry1} and property (a).
	To see that $w \nrank x$, suppose to the contrary that $w \rank x$; then since $y \rank w$ or $y=w$ by property (a), it follows by transitivity of $\rank$ that $y \rank x$, a contradiction.

	Now suppose toward a contradiction that $\{z,w\} \neq \{x,y\}$.
	We claim that
	\begin{equation}
		z \pref y \neq w .
		\label{eq:lemma:noerror_nocry2}
	\end{equation}
	If $z=x$, then $z \pref y$ is immediate, and $y \neq w$ follows since $\{z,w\} \neq \{x,y\}$ by hypothesis.
	Suppose instead that $z \neq x$.
	Since \eqref{eq:lemma:noerror_nocry1} cannot hold, it must be that either $w \pref x$ or $w=x$. Since $x \pref y$, it follows by transitivity of $\pref$ that $w \pref y$, so that $y \neq w$.
	Furthermore, since $z \pref w$, transitivity of $\pref$ yields $z \pref y$.
	So \eqref{eq:lemma:noerror_nocry2} holds.

	It remains to derive a contradiction from $\{z,w\} \neq \{x,y\}$, using the fact that \eqref{eq:lemma:noerror_nocry2} must hold.
	We shall show that $z \pref y$, $y \rank w$ and $y \nrank z$, so that offering $\{z,w\}$ takes a risk at $\rank$.
	We obtain $z \pref y$ from \eqref{eq:lemma:noerror_nocry2}, and $y \rank w$ from \eqref{eq:lemma:noerror_nocry2} and property (a).
	And it must be that $y \nrank z$ because $y \rank z$ together with property (b) and transitivity of $\rank$ imply the contradiction $y \rank x$.
\end{proof}

\begin{corollary}
	\label{corollary:nomissed_pres}
	Suppose that $\rank$ contains no missed opportunities, and that offering $\{z,w\}$ (where $z \nrank w \nrank z$) does not miss an opportunity or take a risk at $\rank$.
	Then the proto-ranking $\rankp$ induced by offering $\{z,w\}$ contains no missed opportunities.
\end{corollary}

\begin{proof}
	Let $\rank$, $z$, $w$ and $\rankp$ be as in the hypothesis of the lemma, and suppose toward a contradiction that there is a missed opportunity $(x,y)$ in $\rankp$.

	We claim that $y \nrank x$ and $y \rankp x$, so that \Cref{lemma:noerror_nocry} is applicable.
	It must be that $y \nrank x$, for otherwise $(x,y)$ would be a missed opportunity in $\rank$.
	That $y \rankp x$ is immediate from the fact that $(x,y)$ is a missed opportunity in $\rankp$.

	It follows by \Cref{lemma:noerror_nocry} that $\{z,w\} = \{x,y\}$.
	But since $(x,y)$ is a missed opportunity in $\rankp$, there is an alternative $z' \in \mathcal{X}$ such that $x \pref z' \pref y$ and $y \nrankp z' \nrankp x$, and thus $y \nrank z' \nrank x$ since $\mathord{\rank} \subseteq \mathord{\rankp}$.
	Thus offering $\{z,w\} = \{x,y\}$ misses an opportunity at $\rank$---a contradiction.
\end{proof}

Armed with \Cref{lemma:noerror_nocry} and \Cref{corollary:nomissed_pres}, we are ready to prove that error-avoiding strategies are efficient.

\begin{namedthm}[Proposition.]
	A strategy that never misses an opportunity or takes a risk is efficient.
\end{namedthm}

\begin{proof}
	Take a strategy $\strat$ that is not efficient, and suppose that it never misses an opportunity or takes a risk; we shall derive a contradiction.
	Since $\strat$ is not efficient, there exists a majority will $\perm$ such that the outcome $\rank$ of $\strat$ under $\perm$ fails to be $\perm$-efficient,
	which is to say that $x \pref y$, $x \perm y$ and $y \rank x$ for some alternatives $x,y \in \mathcal{X}$.

	Write $\varnothing = \mathord{\rankzero} \subseteq \mathord{\rankarg{1}} \subseteq \cdots \subseteq \mathord{\rankarg{T'}} = \mathord{\rank}$ for the sequence of proto-rankings associated with the terminal history generated by $\strat$ and $\perm$.
	Let $T \leq T'$ be the first period in which $x,y$ are ranked ($y \nrankTone x \nrankTone y$ and $y \rankT x$).
	Since $x \perm y$ and $y \rankT x$, it cannot be that $\{x,y\}$ is voted on in period $T$.

	Because $\mathord{\rankzero} = \varnothing$ contains no missed opportunities and $\strat$ never misses an opportunity or takes a risk, \Cref{corollary:nomissed_pres} promises that $\rankTone$ contains no missed opportunities.
	Thus by \Cref{lemma:noerror_nocry}, $\strat$ must offer $\{x,y\}$ in period $T$---a contradiction.
\end{proof}

\subsubsection{Proof that regret-free strategies avoid errors}
\label{app:pf_charac:unimp_noerror}

The proof relies on two lemmata.
For the first, recall from \cref{app:st_definitions} the notation $[x,y]_{\mathord{\rank}}$ for $\rank$-order intervals.

\begin{lemma}
	\label{lemma:unimp_noerror:1}
	Given a pair of alternatives $x,y \in \mathcal{X}$, let $\rankp$ be a ranking such that $x \rankp y$, and let $\perm$ be a majority will that agrees with $\rankp$ on every pair $\{z,w\} \nsubseteq [x,y]_{\mathord{\rankp}}$.
	Then the outcome under $\perm$ of any strategy agrees with $\rankp$ on every pair $\{z,w\} \nsubseteq [x,y]_{\mathord{\rankp}}$.
\end{lemma}

\begin{proof}
	Let $x,y$, $\rankp$ and $\perm$ satisfy the hypothesis, and let $\rank$ be the outcome under $\perm$ of some strategy of the chair.

	\begin{namedthm}[Claim.]
		\label{claim:unimp_noerror}
		If $\{z,w\} \nsubseteq [x,y]_{\mathord{\rankp}}$ satisfy (a) $z \rankp w$ and (b) either $z \nrankp x$ or $y \nrankp w$,
		then $z \rank w$.
	\end{namedthm}

	\begin{proof}[Proof of the \protect{\hyperref[claim:unimp_noerror]{claim}}]%
		\renewcommand{\qedsymbol}{$\square$}
		Assume that $z \nrankp x$; we omit the similar argument for the case $y \nrankp w$.
		Suppose toward a contradiction that $w \rank z$.
		Label the alternatives $[w,z]_{\mathord{\rank}} \equiv \{x_1,\dots,x_K\}$ so that
		\begin{equation*}
			w = x_1 \rank \cdots \rank x_K = z .
		\end{equation*}
		Since $\rank$ is $\perm$-reachable, we have $x_1 \perm \cdots \perm x_K$ by \Cref{observation:reachability_adjacency} (\cref{app:extra:reachability_adjacency}, \cpageref{observation:reachability_adjacency}).
		Suppose that $y \rankp x_k$ for every $k<K$.
		Then $\{x_k,x_{k+1}\} \nsubseteq [x,y]_{\mathord{\rankp}}$ for every $k<K$.
		Since $\rankp$ agrees with $\perm$ on pairs $\{z',w'\} \nsubseteq [x,y]_{\mathord{\rankp}}$, it follows that $x_k \rankp x_{k+1}$ for every $k < K$, whence $w \rankp z$ by transitivity of $\rankp$, contradicting (a).

		Suppose instead that $x_k \rankp y$ for some $k<K$, and let $k'$ be the smallest such $k$.
		It must be that $x \rankp w$, since otherwise $z \rankp w$ and the transitivity of $\rankp$ would produce the contradiction $z \rankp x$.
		Thus $x \rankp w$. Then $w \notin [x,y]_{\mathord{\rankp}}$, as $\{z,w\} \nsubseteq [x,y]_{\rankp}$, $z \rankp w$ and $z \nrankp x$. Hence $y \rankp w = x_1$, so that $k' > 1$.
		By definition of $k'$, we have $x_{k'} \rankp y \rankp x_{k'-1}$.
		On the one hand, the transitivity of $\rankp$ demands that $x_{k'} \rankp x_{k'-1}$.
		On the other hand, since $\{x_{k'-1},x_{k'}\} \nsubseteq [x,y]_{\mathord{\rankp}}$ (because $y \rankp x_{k'-1}$) and $x_{k'-1} \perm x_{k'}$, we must have $x_{k'-1} \rankp x_{k'}$---a contradiction.
	\end{proof}%
	\renewcommand{\qedsymbol}{$\blacksquare$}

	Now fix a pair $\{z,w\} \nsubseteq [x,y]_{\mathord{\rankp}}$ such that $z \rankp w$; we must show that $z \rank w$.
	If either $x \rankp z$ or $w \rankp y$, then $z \rank w$ follows from the \hyperref[claim:unimp_noerror]{claim}.

	Suppose instead that $z \rankp x$ and $y \rankp w$.
	Observe that $\{z,x\} \nsubseteq [x,y]_{\mathord{\rankp}}$, (a) $z \rankp x$ and (b) $y \nrankp x$.
	We may therefore apply the \hyperref[claim:unimp_noerror]{claim} to $\{z,x\}$ to obtain $z \rank x$.
	Similarly applying the \hyperref[claim:unimp_noerror]{claim} to $\{x,w\}$ yields $x \rank w$, whence $z \rank w$ follows by transitivity of $\rank$.
	%
\end{proof}

\begin{lemma}
	\label{lemma:unimp_noerror:3}
	Let $\rank$ be a proto-ranking, and let $x,y,z \in \mathcal{X}$ be such that
	\begin{equation*}
		\{x,y,z\}^2 \intersect \mathord{\rank} \subseteq \{(x,z)\} .
	\end{equation*}
	Then there exists a ranking $\mathord{\rankp} \supseteq \mathord{\rank} \union \{(x,z),(z,y)\}$ such that $[x,y]_{\mathord{\mathord{\rankp}}} = [x,z]_{\mathord{\rank}} \union \{y\}$.
\end{lemma}

To interpret the conclusion, observe that the properties of $\rankp$ are equivalent to the following: (a) $x \rankp z$ and $[x,z]_{\mathord{\mathord{\rankp}}} = [x,z]_{\mathord{\rank}}$, and (b) $z \rankp y$ and $[z,y]_{\mathord{\rankp}} = \{z,y\}$.
In words, \Cref{lemma:unimp_noerror:3} runs as follows.
Suppose that a proto-ranking $\rank$ ranks $x$ above $z$, and has nothing else to say about $\{x,y,z\}$.%
	\footnote{For simplicity, neglect the case $\{x,y,z\}^2 \intersect \mathord{\rank} = \varnothing$.}
Call the (possibly empty) set of alternatives ranked below $x$ and above $z$ (i.e. $[x,z]_{\mathord{\rank}} \setminus \{x,z\}$) the `in-between set'.
The lemma asserts that $\rank$ may be extended to a ranking $\rankp$ that (i) adds nothing to the in-between set ($[x,z]_{\mathord{\rankp}} = [x,z]_{\mathord{\rank}}$) and that (ii) ranks $y$ immediately below $z$ ($z \rankp y$ and $[z,y]_{\mathord{\rankp}} = \{z,y\}$).

The proof of \Cref{lemma:unimp_noerror:3} relies on the following general extension principle.
Recall from \cref{app:st_definitions} the definition of `extension'.

\begin{namedthm}[Extension lemma.]
	\label{lemma:extension}
	Let $\rank$ be a proto-ranking, and let $A \subseteq \mathcal{X}$ be such that $[x,y]_{\mathord{\rank}} \subseteq A$ for all $x,y \in A$ with $x \rank y$.
	Then the binary relation $\mathord{\rank} \union A^2$ admits a complete and transitive extension.
\end{namedthm}

\begin{proof}[Proof of the \protect{\hyperref[lemma:extension]{extension lemma}}]
	Let $\rank$ and $A$ satisfy the hypothesis; we seek a complete and transitive extension of the relation $\mathord{\rel} \coloneqq \mathord{\rank} \union A^2$.
	By Suzumura's extension theorem,%
		\footnote{See e.g. \textcite[p. 45]{BossertSuzumura2010}.} 
	it suffices to show that for any finite sequence of alternatives $(w_k)_{k=1}^K$ such that $w_1 \rel \cdots \rel w_K$, we have either $w_1 \rel w_K$ or $w_1 \nrel w_K \nrel w_1$.
	There are two cases.

	\emph{Case 1: $w_k \rank w_{k+1}$ for every $k<K$.}
	Then $w_1 \rank w_K$ since $\rank$ is transitive (being a proto-ranking), so $w_1 \rel w_K$.

	\emph{Case 2: $\{w_k,w_{k+1}\} \subseteq A$ for some $k<K$.}
	Let $k'$ ($k''$) be the smallest (largest) such $k<K$, so that
	$w_1 \rank \cdots \rank w_{k'}$ if $k'>1$
	and $w_{k''+1} \rank \cdots \rank w_K$ if $k''<K-1$.
	Assume toward a contradiction that $w_K \rel w_1$ and $w_1 \nrel w_K$.
	Since $\{w_1,w_K\} \nsubseteq A$ (as otherwise $w_1 \rel w_K \rel w_1$), it must be that $w_K \rank w_1$
	and either $k'>1$ or $k''<K-1$.
	By transitivity of $\rank$, $w_1 \rank w_{k'}$ if $k' > 1$, and $w_{k''+1} \rank w_K$ if $k'' < K-1$;
	in either case, $\{w_K,w_1\} \subseteq [w_{k''+1},w_{k'}]_{\mathord{\rank}}$.
	Note that $[w_{k''+1},w_{k'}]_{\mathord{\rank}} \subseteq A$ since $w_{k''+1},w_{k'} \in A$ by construction.
	Therefore $\{w_K,w_1\} \subseteq A$, which implies that $w_K \rel w_1 \rel w_K$---a contradiction.
\end{proof}

\begin{proof}[Proof of \Cref{lemma:unimp_noerror:3}]
	Let a proto-ranking $\rank$ and alternatives $x,y,z \in \mathcal{X}$ satisfy the hypothesis. Define $A \coloneqq [x,z]_{\mathord{\rank}} \union \{y\}$.

	\begin{namedthm}[Claim.]
		\label{namedthm:hypothesis_verified}
		For any $x',y' \in A$ such that $x' \rank y'$, we have $[x',y']_{\mathord{\rank}} \subseteq A$.
	\end{namedthm}

	\begin{proof}[Proof of the \protect{\hyperref[namedthm:hypothesis_verified]{claim}}]%
		\renewcommand{\qedsymbol}{$\square$}
		Fix alternatives $x',y' \in A$ with $x' \rank y'$.
		By definition of $A$, it suffices to show that $\{x',y'\} \centernot{\ni} y$.
		So suppose toward a contradiction that $x' = y$; the other case is similar.
		We have $y' \neq y$ since $x' \rank y'$ and $\rank$ is irreflexive (being a proto-ranking).
		Since $y' \in A$, it follows that $y' \in [x,z]_{\mathord{\rank}}$.
		By $x' \rank y'$ and the transitivity of $\rank$, we obtain $y = x' \rank z$.
		But $y \nrank z$ by hypothesis---a contradiction.
	\end{proof}%
	\renewcommand{\qedsymbol}{$\blacksquare$}

	By the \hyperref[namedthm:hypothesis_verified]{claim}, the \hyperref[lemma:extension]{extension lemma} is applicable, so there exists a complete and transitive extension $\rel$ of the binary relation $\mathord{\rank} \union A^2$.
	Since $z' \rel w' \rel z'$ for any $z',w' \in A$, we have in particular that $w \rel y \rel w$ for any $w \in [x,z]_{\mathord{\rank}}$.
	We may therefore obtain the desired ranking $\rankp$ by appropriately breaking indifferences in $\rel$.
\end{proof}

With \Cref{lemma:unimp_noerror:1,lemma:unimp_noerror:3} in hand, we are ready to prove that regret-free strategies avoid errors.

\begin{namedthm}[Proposition.]
	A regret-free strategy never misses an opportunity or takes a risk.
\end{namedthm}

\begin{proof}
	We shall prove the contra-positive.
	To that end, fix a strategy $\strat$ and a majority will $\perm$ such that $\strat$ misses an opportunity or takes a risk under $\perm$.
	We shall construct a majority will $\permp$ such that the outcome $\rank$ of $\strat$ under $\permp$ fails to be $\permp$-unimprovable.
	In particular, we shall find a $\permp$-reachable ranking $\mathord{\rankp} \neq \mathord{\rank}$ that is more aligned with $\pref$ than $\rank$.

	Let $T$ be the first period in which $\strat$ either misses an opportunity or takes a risk under $\perm$.
	Write $\rankTone$ for the associated end-of-period-$(T-1)$ proto-ranking, and let $\{x,y\}$ be the pair offered in period $T$.

	We shall consider three cases, based on the behaviour of $\strat$ under $\perm$ in period $T$.
	By hypothesis, $\{x,y\}$ either misses an opportunity or takes a risk at $\rankTone$.
	If $\{x,y\}$ misses an opportunity, there is an alternative $z \in \mathcal{X}$ such that $x \pref z \pref y$ and $y \nrankTone z \nrankTone x$.
	It cannot be that $x \rankTone z \rankTone y$, as this would imply that $x \rankTone y$, contradicting the fact that $\{x,y\}$ is offered in period $T$.
	Thus one of the following must hold:
	\begin{enumerate}[label=(\alph*)]

		\item \label{item:pf_regret-free_noerror:a}
		$x \nrankTone z \nrankTone y$,

		\item \label{item:pf_regret-free_noerror:b}
		$x \rankTone z \nrankTone y$, or

		\item \label{item:pf_regret-free_noerror:c}
		$x \nrankTone z \rankTone y$.

	\end{enumerate}
	If instead $\{x,y\}$ takes a risk, then there is a $z \in \mathcal{X}$ such that either
	\begin{enumerate}[label=(\alph*)]

		\setcounter{enumi}{3}

		\item \label{item:pf_regret-free_noerror:d}
		$z \pref y$, $x \rankTone z$, and $y \nrankTone z$, or

		\item \label{item:pf_regret-free_noerror:e}
		$x \pref z$, $z \rankTone y$, and $z \nrankTone x$.

	\end{enumerate}

	This yields three cases, as follows.
	Case 1 is \ref{item:pf_regret-free_noerror:a}.
	Case 2 encompasses both \ref{item:pf_regret-free_noerror:b} and \ref{item:pf_regret-free_noerror:d} under the (slightly more general) hypothesis that `there exists a $z \in \mathcal{X}$ such that $z \pref y$, $z \nrankTone y \nrankTone z$, and $x \rankTone z$'.
	Finally, case 3 encompasses \ref{item:pf_regret-free_noerror:c} and \ref{item:pf_regret-free_noerror:e} under the hypothesis that `there exists a $z \in \mathcal{X}$ such that $x \pref z$, $z \nrankTone x \nrankTone z$ and $z \rankTone y$.
	Since cases 2 and 3 are analogous, we omit the proof for the latter.

	\emph{Case 1: $\exists z \in \mathcal{X}$ such that $x \pref z \pref y$ and $\{x,z,y\}^2 \intersect \mathord{\rankTone} = \varnothing$.}
	By \Cref{lemma:unimp_noerror:3}, there is a ranking $\rankp$ such that
	\begin{equation*}
		\mathord{\rankp}
		\supseteq \mathord{\rankTone} \union \{(x,z),(z,y)\}
		\quad \text{and} \quad
		[x,y]_{\mathord{\rankp}} = \{x,y,z\} .
	\end{equation*}
	Define $\permp$ to equal $\rankp$, except that $y \permp x$.
	Clearly $\permp$ is a majority will (total and asymmetric), and $\rankp$ is $\permp$-reachable by \Cref{observation:reachability_adjacency} (\cref{app:extra:reachability_adjacency}, \cpageref{observation:reachability_adjacency}) since $x,y$ are not $\rankp$-adjacent.
	Denote by $\rank$ the outcome of $\strat$ under $\permp$.
	It remains to show that $\mathord{\rank} \neq \mathord{\rankp}$, and that $\rankp$ is more aligned with $\pref$ than $\rank$.

	For the former, since $x \rankp y$, it suffices to show that $y \rank x$.
	To this end, observe that that $\mathord{\rankTone} \subseteq \mathord{\permp}$.
	Thus the history of length $T-1$ generated by $\strat$ and $\permp$ is the same as that generated by $\strat$ and $\perm$, which means in particular that $\{x,y\}$ is offered in period $T$.
	Since $y \permp x$, it follows that $y \rank x$, as desired.

	To show that $\rankp$ is more aligned with $\pref$ than $\rank$,
	observe that $\permp$ agrees with $\rankp$ on every pair $\{w,w'\} \nsubseteq \{x,y,z\} = [x,y]_{\mathord{\rankp}}$.
	It follows by \Cref{lemma:unimp_noerror:1} that $\rank$ and $\rankp$ agree on every pair $\{w,w'\} \nsubseteq \{x,y,z\}$.
	Since $x \pref z \pref y$ and $x \rankp z \rankp y$, it follows that $\rankp$ is more aligned with $\pref$ than $\rank$.

	\emph{Case 2: $\exists z \in \mathcal{X}$ such that $z \pref y$, $z \nrankTone y \nrankTone z$ and $x \rankTone z$.}
	We shall begin with an auxiliary ranking $\rankpp$, then use it to construct our majority will $\permp$ and ranking $\rankp$.
	By \Cref{lemma:unimp_noerror:3}, there is a ranking
	\begin{equation*}
		\mathord{\rankpp}
		\supseteq \mathord{\rankTone} \union \{(x,z),(z,y)\}
	\end{equation*}
	such that
	\begin{equation}
		[x,y]_{\mathord{\rankpp}}
		= [x,z]_{\mathord{\rankTone}} \union \{y\}.
		\label{eq:int_a}
	\end{equation}
	Define
	\begin{equation*}
		X
		\coloneqq \left\{
		w \in [x,z]_{\mathord{\rankTone}} \setminus \{x\} :
		w \pref y \right\} ,
	\end{equation*}
	and let $\permp$ be such that
	\begin{enumerate}[label=(\roman*)]

		\item \label{item:unimp_noerror:permp:i}
		$w \permp y$ for every $w \in X$,

		\item \label{item:unimp_noerror:permp:ii}
		$y \permp w$ for every $w \in [x,z]_{\mathord{\rankTone}} \setminus X$, and

		\item \label{item:unimp_noerror:permp:iii}
		$\permp$ agrees with $\rankpp$ on every other pair.

	\end{enumerate}
	Denote by $\rank$ the outcome of $\strat$ under the majority will $\permp$.

	Observe that
	\ref{item:unimp_noerror:permp:i} $y \nrankTone w$ for every $w \in X$
	(since otherwise $y \rankTone w \rankTone z$, contradicting the case-2 hypothesis),
	\ref{item:unimp_noerror:permp:ii} $w \nrankTone y$ for every $w \in [x,z]_{\mathord{\rankTone}} \setminus X$
	(otherwise either $x = w \rankTone y$ or $x \rankTone w \rankTone y$, whence $x \rankTone y$), and
	\ref{item:unimp_noerror:permp:iii} $\mathord{\rankTone} \subseteq \mathord{\rankpp}$.
	Thus $\mathord{\rankTone} \subseteq \mathord{\permp}$ by definition of the latter.
	It follows that the history of length $T-1$ generated by $\strat$ and $\permp$ is the same as that generated by $\strat$ and $\perm$, which means in particular that $\{x,y\}$ is offered in period $T$.
	Since $y \permp x$, we thus have $y \rank x$.

	Since $X \subseteq [x,z]_{\mathord{\rankTone}} \subseteq [x,y]_{\mathord{\rankpp}}$ (by definition of $X$ and \eqref{eq:int_a}), $\permp$ agrees with $\rankpp$ on every pair $\{w,w'\} \nsubseteq [x,y]_{\mathord{\rankpp}}$.
	It follows by \Cref{lemma:unimp_noerror:1} that $\rank$ agrees with $\rankpp$ on every pair $\{w,w'\} \nsubseteq [x,y]_{\rankpp}$.
	This, together with \eqref{eq:int_a}, $\mathord{\rankTone} \subseteq \mathord{\rank}$ and $y \rank x$,
	implies that $y,x$ are $\rank$-adjacent, whence
	\begin{equation}
		[y,z]_{\mathord{\rank}}
		= \{y\} \union [x,z]_{\mathord{\rank}}
		= \{y\} \union [x,z]_{\mathord{\rankTone}}
		= [x,y]_{\mathord{\rankpp}} .
		\label{eq:int_b}
	\end{equation}
	It follows that $X \union \{x\} \subseteq [y,z]_{\mathord{\rank}}$.

	Define $X' \coloneqq X \union \{x\}$, and label its elements $X' \equiv \{a_1,\dots,a_K\}$ so that $a_1 \rank \cdots \rank a_K$.
	Similarly label $[y,z]_{\mathord{\rank}} \setminus X' \equiv \{b_1,\dots,b_L\}$ so that $b_1 \rank \cdots \rank b_L$.%
		\footnote{$[y,z]_{\mathord{\rank}} \setminus X'$ is non-empty since $y$ belongs to it.}
	Let $\rankp$ be the ranking that
	\begin{enumerate}[label=(\Roman*)]

		\item \label{item:unimp_noerror:rankp:i}
		agrees with $\rank$ on any pair $\{w,w'\} \nsubseteq [y,z]_{\mathord{\rank}}$, and

		\item \label{item:unimp_noerror:rankp:ii}
		ranks the elements of $[y,z]_{\mathord{\rank}}$ as
		$a_1 \rankp \cdots \rankp a_K \rankp b_1 \rankp \cdots \rankp b_L .$

	\end{enumerate}

	We have now constructed a majority will $\permp$ and a ranking $\rankp$.
	Recall that $\rank$ is the outcome of $\strat$ under $\permp$.
	It remains to show that
	\begin{enumerate}[leftmargin=1.2cm]

		\item[\namedlabel{item:pf_regret-free_noerror:dist}{(dist)}]
		$\rankp$ is distinct from $\rank$,

		\item[\namedlabel{item:pf_regret-free_noerror:alig}{(alig)}]
		$\rankp$ is more aligned with $\pref$ than $\rank$, and

		\item[\namedlabel{item:pf_regret-free_noerror:feas}{(feas)}]
		$\rankp$ is $\permp$-reachable.

	\end{enumerate}

	For \ref{item:pf_regret-free_noerror:dist},
	observe that since $x \in X'$ and $y \in [y,z]_{\mathord{\rank}} \setminus X'$, we have $x = a_k \rankp b_\ell = y$ for some $k$ and $\ell$.%
		\footnote{In fact, $k=\ell=1$ since $y$ is $\rank$-highest in $[y,z]_{\mathord{\rank}}$ and (recall) $y,x$ are $\rank$-adjacent.}
	Since $y \rank x$, it follows that $\mathord{\rankp} \neq \mathord{\rank}$.

	For \ref{item:pf_regret-free_noerror:alig}, fix a pair $w,w' \in \mathcal{X}$ with $w \rank w'$ and $w' \rankp w$; we must show that $w' \pref w$.
	By definition of $\rankp$, it must be that $w' \in X' = X \union \{x\}$ and that $w \in [y,z]_{\mathord{\rank}} \setminus X'$.
	Thus $w' \pref y$ (by $x \pref y$ and the definition of $X$) and either $y=w$ or $y \pref w$, whence $w' \pref w$ by transitivity of $\pref$.

	It remains to establish \ref{item:pf_regret-free_noerror:feas}.
	To this end (recalling \Cref{observation:reachability_adjacency} in \cref{app:extra:reachability_adjacency}, \cpageref{observation:reachability_adjacency}), fix an $\rankp$-adjacent pair $w,w' \in \mathcal{X}$ with $w' \rankp w$;
	we must show that $w' \permp w$.
	There are two cases.

	\emph{Sub-case (2)(a): $\{w,w'\} \nsubseteq [y,z]_{\mathord{\rank}}$.}
	Then $w' \rank w$ by part \ref{item:unimp_noerror:rankp:i} of the definition of $\rankp$.
	Since $\rank$ agrees with $\rankpp$ on any pair $\{z',z''\} \nsubseteq [y,z]_{\mathord{\rank}}$, it follows that $w' \rankpp w$.
	It therefore suffices to show that $\permp$ agrees with $\rankpp$ on $\{w,w'\}$.
	Recalling the definition \ref{item:unimp_noerror:permp:i}--\ref{item:unimp_noerror:permp:iii} of $\permp$,
	\begin{itemize}

		\item If $\{w,w'\} \centernot{\ni} y$, then $\permp$ agrees with $\rankpp$ on $\{w,w'\}$ by \ref{item:unimp_noerror:permp:iii}.

		\item If $w' = y \in [y,z]_{\mathord{\rank}}$, then $w \notin [y,z]_{\mathord{\rank}} \supseteq [x,z]_{\mathord{\rankTone}} \supseteq X$ by \eqref{eq:int_b}, so neither \ref{item:unimp_noerror:permp:i} nor \ref{item:unimp_noerror:permp:ii} applies to the pair $\{w,y\}=\{w,w'\}$.
		Thus by \ref{item:unimp_noerror:permp:iii}, $\permp$ agrees with $\rankpp$ on $\{w,w'\}$.

		\item If $w = y \in [y,z]_{\mathord{\rank}}$, then $w' \notin [y,z]_{\mathord{\rank}} \supseteq [x,z]_{\mathord{\rankTone}} \supseteq X$ by \eqref{eq:int_b}, so neither \ref{item:unimp_noerror:permp:i} nor \ref{item:unimp_noerror:permp:ii} applies to the pair $\{y,w'\}=\{w,w'\}$.
		Thus by \ref{item:unimp_noerror:permp:iii}, $\permp$ agrees with $\rankpp$ on $\{w,w'\}$.

	\end{itemize}

	\emph{Sub-case (2)(b): $\{w,w'\} \subseteq [y,z]_{\mathord{\rank}}$.}
	Recall the definition \ref{item:unimp_noerror:permp:i}--\ref{item:unimp_noerror:permp:iii} of $\permp$.
	Recall also part \ref{item:unimp_noerror:rankp:ii} of the definition of $\rankp$.
	Observe that $y = b_1$ since $y$ is $\rank$-highest in $[y,z]_{\mathord{\rank}}$, and $y \notin X' = X \union \{x\}$.
	Furthermore, $x = a_1$ since $y,x$ are $\rank$-adjacent (recall \eqref{eq:int_b}) and $x \in X$.
	Finally, remark that $K \geq 2$ since $z \in X' = X \union \{x\}$ and $z \neq x$.

	Suppose first that $w' = y = b_1$.
	Then since $w',w$ are $\rankp$-adjacent with $w' \rankp w$, we have $w = b_2 \notin X' \supseteq X$.
	Thus part \ref{item:unimp_noerror:permp:i} does not apply to the pair $\{w,y\} = \{w,w'\}$.
	So by \eqref{eq:int_b}, part \ref{item:unimp_noerror:permp:ii} applies, yielding $w' = y \permp w$.

	Suppose instead that $w = y = b_1$.
	Then since $w',w$ are $\rankp$-adjacent with $w' \rankp w$, we have $w' = a_K \in X' = X \union \{x\}$.
	Since $x=a_1$ and $K \geq 2$, we have $w' \neq x$.
	Thus $w' \in X$, so that \ref{item:unimp_noerror:permp:i} applies to the pair $\{y,w'\} = \{w,w'\}$, yielding $w' \permp y = w$.

	Finally, suppose that $\{w,w'\} \centernot{\ni} y$.
	Then $\permp$ and $\rankpp$ agree on the pair $\{w,w'\}$ by \ref{item:unimp_noerror:permp:iii}, so it suffices to show that $w' \rankpp w$.
	Since $b_1 = y \notin \{w,w'\}$, we have either $\{w,w'\} \subseteq X'$ or $\{w,w'\} \subseteq [y,z]_{\mathord{\rank}} \setminus X'$.
	Thus $\rank$ and $\rankp$ agree on $\{w,w'\}$ by part \ref{item:unimp_noerror:rankp:ii} of the definition of $\rankp$, so that $w' \rank w$.

	Now label $[w',w]_{\mathord{\rank}} \equiv \{z_1,\dots,z_J\}$ so that $z_1 \rank \cdots \rank z_J$.
	Since $\rank$ is $\permp$-reachable, we have $z_1 \permp \cdots \permp z_J$ by \Cref{observation:reachability_adjacency} (\cref{app:extra:reachability_adjacency}, \cpageref{observation:reachability_adjacency}).
	By the hypotheses $w' \in [y,z]_{\mathord{\rank}}$ and $w' \neq y$, we must have $y \rank w'$ and thus $y \notin [w',w]_{\mathord{\rank}}$.
	This together with the fact that $z_j \permp z_{j+1}$ for each $j<J$ implies, via part \ref{item:unimp_noerror:permp:iii} of the definition of $\permp$, that $z_j \rankpp z_{j+1}$ for each $j<J$.
	It follows by transitivity of $\rankpp$ that $w' = z_1 \rankpp z_J = w$, as desired.
\end{proof}

\subsection{Proof of \texorpdfstring{\Cref{theorem:is_lexicographic}}{Theorem \ref{theorem:is_lexicographic}} (§\ref{sec:charac:is2}, p. \pageref{theorem:is_lexicographic})}
\label{app:pf_proposition:is2}

In this appendix, we first prove \Cref{theorem:is_lexicographic} using two lemmata (§\ref{app:pf_proposition:is2:overview}), and then prove these lemmata (§\ref{app:pf_proposition:is2:prelim}--§\ref{app:pf_proposition:is2:pf_is_oe_charac}).
Throughout, we label the alternatives $\mathcal{X} \equiv \{1,\dots,n\}$ so that $1 \pref \cdots \pref n$.

\subsubsection{Proof using lemmata}
\label{app:pf_proposition:is2:overview}

\begin{definition}
	\label{definition:Sigmak}
	Let $\Sigma_0$ be the set of all strategies.
	For every integer $k \in \{1,\dots,n-2\}$, let $\Sigma_k$ be the set of all strategies $\strat$ with the following property:
	for any majority will $\perm$ and alternative $j \leq k$, labelling the alternatives $\{j+1,\dots,n\} \equiv \{x_{j+1},\dots,x_n\}$ as
	\begin{equation*}
		x_{j+1} \rankfn{\strat}{\perm}
		\cdots
		\rankfn{\strat}{\perm} x_n ,
	\end{equation*}
	the first vote involving $j$ that $\strat$ offers under $\perm$ is on $\{j,x_{j+1}\}$;
	if $j$ loses, then a second vote involving $j$ is offered, namely on $\{j,x_{j+2}\}$;
	if $j$ loses again, then a third vote involving $j$ is offered, namely on $\{j,x_{j+3}\}$;
	and so on.
\end{definition}

The definition of $\Sigma_{n-2}$ describes a natural generalisation of insertion sort:
for each alternative $j$, given how the $\pref$-worse alternatives $\{x_{j+1},\dots,x_n\}$ are ultimately ranked,
the same votes involving $j$ are offered, in the same order, though not necessarily in adjacent periods.%
	\footnote{Note a subtlety in the definition: although we label $\{x_{j+1},\dots,x_n\}$ according to the outcome $\rankfn{\strat}{\perm}$ of $\strat$ under $\perm$, a strategy in $\Sigma_k$ need not (as insertion sort would) have totally ranked $\{j+1,\dots,n\}$ before offering votes involving $j$.}
Each $\Sigma_k$ for $k<n-2$ is defined by the same property restricted to those alternatives $j$ that are $\pref$-better than or equal to $k$,
so that $\Sigma_0 \supseteq \Sigma_1 \supseteq \cdots \supseteq \Sigma_{n-2}$.

\begin{lemma}
	\label{lemma:is_oe_charac}
	A strategy is outcome-equivalent to insertion sort iff it belongs to $\Sigma_{n-2}$.
\end{lemma}

\begin{lemma}
	\label{lemma:best_Sigmak}
	Given $k \in \{1,\dots,n-2\}$, a strategy $\strat \in \Sigma_{k-1}$ is best for $k$ among $\Sigma_{k-1}$ iff it belongs to $\Sigma_k$.
\end{lemma}

\begin{proof}[Proof of \Cref{theorem:is_lexicographic}]
	Let $B_0 = \Sigma_0$ be the set of all strategies,
	and for each $k \in \{1,\dots,n-2\}$, let $B_k$ be the set of strategies in $B_{k-1}$ that are best for $k$ among $B_{k-1}$.%
		\footnote{Let $B_k \coloneqq \varnothing$ if $B_{k-1}$ is empty. (It is in fact non-empty, but we haven't proved it yet.)}
	A lexicographic strategy is precisely one that lives in $B_{n-2}$.

	By \Cref{lemma:is_oe_charac}, a strategy is outcome-equivalent to insertion sort iff it lives in $\Sigma_{n-2}$; so what must be shown is that $\Sigma_{n-2} = B_{n-2}$.
	We shall prove the stronger claim that $\Sigma_k = B_k$ for each $k \in \{0,\dots,n-2\}$ by (weak) induction on $k$.
	The base case $k=0$ holds by definition of $\Sigma_0$ and $B_0$.
	For the induction step, suppose that $\Sigma_{k-1} = B_{k-1}$; then $\Sigma_k = B_k$ by \Cref{lemma:best_Sigmak}.
\end{proof}

The remainder of this appendix is devoted to proving \Cref{lemma:is_oe_charac,lemma:best_Sigmak}.
We begin in §\ref{app:pf_proposition:is2:prelim} with two preliminary results, then prove \Cref{lemma:best_Sigmak} in §\ref{app:pf_proposition:is2:lemma:best_Sigmak_pf} and \Cref{lemma:is_oe_charac} in §\ref{app:pf_proposition:is2:pf_is_oe_charac}.

\subsubsection{Preliminary results}
\label{app:pf_proposition:is2:prelim}

The following lemma is used in the proof of \Cref{lemma:is_oe_charac}.

\begin{lemma}
	\label{lemma:is2:prelim}
	Given a $k \in \{1,\dots,n-2\}$, consider a strategy $\strat \in \Sigma_k$, a majority will $\perm$, an alternative $j \leq k$ and some $m \in \{1,\dots,n-j\}$.
	Suppose that under $\perm$,
	$\strat$ offers at least $m$ votes involving $j$,
	and that $j$ loses the first $m-1$ of these.
	Let $\ell$ be the $m^{\text{th}}$ alternative pitted against $j$, and let $\rank$ be the proto-ranking associated with the history after which the vote on $\{j,\ell\}$ occurs.
	Then $\ell \rank i$ for any $i \neq \ell$ such that $i > j$ and $j$ did not lose against $i$ prior to the vote on $\{j,\ell\}$.
\end{lemma}

\begin{proof}
	Suppose toward a contradiction that $\ell \nrank i$ for some $i \neq \ell$ with $i > j$ such that $j$ did not lose to $i$ prior to the vote on $\{j,\ell\}$.
	Then there exists a ranking $\permp$ such that $\mathord{\rank} \subseteq \mathord{\permp}$ and $i \permp \ell$.
	(Note that a ranking is precisely a transitive majority will.)
	Being a ranking, $\permp$ is the only $\permp$-reachable ranking (by \Cref{observation:reachability_adjacency} in \cref{app:extra:reachability_adjacency}, \cpageref{observation:reachability_adjacency}), so that $\mathord{\rankfn{\strat}{\permp}} = \mathord{\permp}$, and in particular $i \rankfn{\strat}{\permp} \ell$.

	Let $T$ be the period in which $\strat$ offers $\{j,\ell\}$ (i.e. the $m^{\text{th}}$ vote involving $j$) under $\perm$.
	Since $\mathord{\rank} \subseteq \mathord{\permp}$,
	the history of length $T-1$ generated by $\strat$ and $\permp$ is the same as that generated by $\strat$ and $\perm$.
	So in particular, under $\permp$,
	$\strat$ offers $\{j,\ell\}$ in period $T$,
	$\strat$ does not offer $\{j,i\}$ in an earlier period,
	and $j$ does not win a vote in any earlier period.
	Since $\strat \in \Sigma_k$, it follows that $\ell \rankfn{\strat}{\permp} i$, a contradiction.
\end{proof}

The proof of \Cref{lemma:best_Sigmak} relies on the following.

\begin{namedthm}[Lemma $\boldsymbol{\natural}$.]
	\label{lemma:worst-case}
	Given $k \in \{2,\dots,n-2\}$,
	if $N_k^\strat(\mathord{\perm}) > 0$ for some $\strat \in \Sigma_{k-1}$ and majority will $\perm$,
	then prior to winning its first vote, $k$ is pitted only against $\pref$-worse alternatives.
\end{namedthm}

\begin{proof}
	We prove the contra-positive.
	Let $k \in \{2,\dots,n-2\}$, $\strat \in \Sigma_{k-1}$ and a majority will $\perm$ be such that $\strat$ pits $k$ against some $j < k$ under $\perm$ and $k$ wins no vote before the one against $j$; we must show that $N_k^\strat(\mathord{\perm}) = 0$.
	Let $\rank$ be the proto-ranking associated with the history after which the vote on $\{j,k\}$ occurs.
	It suffices to show that $\ell \rank k$ for all $\ell > k$.

	\begin{namedthm}[Claim.]
		\label{claim:a}
		$j$ wins no vote prior to the one against $k$.
	\end{namedthm}

	\begin{proof}[Proof of the \protect{\hyperref[claim:a]{claim}}]%
		\renewcommand{\qedsymbol}{$\square$}
		Suppose toward a contradiction that the first alternative against which $j$ wins is $\ell \neq k$.
		Then $j \rank \ell$.
		Since $\strat \in \Sigma_{k-1}$, the vote on $\{j,\ell\}$ is the $m^{\text{th}}$ involving $j$, for some $m < n-j$.
		It follows by \Cref{lemma:is2:prelim} (above) that $\ell \rank k$, which together with $j \rank \ell$ and the transitivity of $\rank$ yields $j \rank k$.
		On the other hand, since $\{j,k\}$ is offered after a history with proto-ranking $\rank$, we must have $j \nrank k \nrank j$.
		Contradiction!
	\end{proof}%
	\renewcommand{\qedsymbol}{$\blacksquare$}

	Fix an $\ell > k$; we shall show that $\ell \rank k$.
	Since $\strat \in \Sigma_{k-1}$, the vote on $\{j,k\}$ is the $m^{\text{th}}$ involving $j$, for some $m < n-j$.
	It must be that $\strat$ offers $\{j,\ell\}$ prior to $\{j,k\}$, since otherwise \Cref{lemma:is2:prelim} would yield $k \rank \ell$, contradicting the hypothesis that $k$ wins no votes prior to the one against $j$.
	Since $j$ wins no vote prior to the one against $k$ and $m < n-j$, \Cref{lemma:is2:prelim} yields $\ell \rank k$, as desired.
\end{proof}

\subsubsection{Proof of \texorpdfstring{\Cref{lemma:best_Sigmak}}{Lemma \ref{lemma:best_Sigmak}} (§\ref{app:pf_proposition:is2:overview}, p. \pageref{lemma:best_Sigmak})}
\label{app:pf_proposition:is2:lemma:best_Sigmak_pf}

Let $\mathcal{W}$ denote the set of all majority wills (total and asymmetric relations) on $\mathcal{X}$.
We shall use the probabilistic notation $\Pr(E|F) \coloneqq \abs*{E \intersect F} / \abs*{F}$ for the fraction of majority wills in $F \subseteq \mathcal{W}$ that belong to $E \intersect F \subseteq F$, and similarly $\Pr(E) \coloneqq \Pr(E|\mathcal{W})$.
This corresponds the thought experiment in which the majority will $\perm$ is drawn uniformly at random from $\mathcal{W}$.

We must establish that membership of $\Sigma_k$ is necessary and sufficient for being best for $k$ among $\Sigma_{k-1}$.
We prove sufficiency and necessity in turn, making use of \hyperref[lemma:worst-case]{Lemma $\natural$} from the previous section.

\begin{proof}[Proof of sufficiency]
	Fix $k$, a strategy $\strat \in \Sigma_{k-1}$ and an $m \in \{1,\dots,n-k\}$.
	We shall derive an upper bound for $\Pr( N_k^\strat \geq m )$, then show that it is attained if $\strat \in \Sigma_k$.

	For each $\ell \in \{1,\dots,n-k\}$, let $F_\ell \subseteq \mathcal{W}$ be the set of majority wills under which $\strat$ offers at least $\ell$ votes involving alternative $k$,
	with alternative $k$ losing the first $\ell-1$ of these and winning the $\ell^{\text{th}}$.

	\begin{namedthm}[Claim.]
		\label{claim:pf_lemma:best_Sigmak}
		$\Pr( F_\ell ) \leq 1 / 2^\ell$ for each $\ell \in \{1,\dots,n-k\}$,
		with equality if $\strat \in \Sigma_k$.
	\end{namedthm}

	\begin{proof}[Proof of the \protect{\hyperref[claim:pf_lemma:best_Sigmak]{claim}}]%
		\renewcommand{\qedsymbol}{$\square$}
		A majority will $\perm$ lies in $F_\ell$ iff under $\strat$ and $\perm$,
		\begin{itemize}

			\item alternative $k$ loses its first vote (probability $1/2$),

			\item a second vote involving alternative $k$ occurs (probability $\leq 1$, with equality if $\strat \in \Sigma_k$) and $k$ loses (probability $1/2$),

			\dots

			\item an $(\ell-1)^{\text{th}}$ vote involving alternative $k$ occurs (probability $\leq 1$, with equality if $\strat \in \Sigma_k$) and $k$ again loses (probability $1/2$), and

			\item an $\ell^{\text{th}}$ vote involving alternative $k$ occurs (probability $\leq 1$, with equality if $\strat \in \Sigma_k$) and $k$ wins (probability $1/2$).

		\end{itemize}
		Thus
		\begin{equation*}
			\Pr(F_\ell)
			\leq \frac{1}{2}
			\times \left( 1 \times \frac{1}{2} \right)^{\ell-1}
			= \frac{1}{2^\ell} ,
			\quad \text{with equality if $\strat \in \Sigma_k$.}
			\qedhere
		\end{equation*}
	\end{proof}%
	\renewcommand{\qedsymbol}{$\blacksquare$}

	By \hyperref[lemma:worst-case]{Lemma $\natural$} (§\ref{app:pf_proposition:is2:prelim}, \cpageref{lemma:worst-case}),
	if $N_k^\strat(\mathord{\perm}) > 0$, then prior to winning its first vote, $k$ was only pitted against $\pref$-worse alternatives.
	There are $n-k$ of these: $\{k+1,\dots,n\}$.
	Thus if $N_k^\strat(\mathord{\perm}) \geq m$ holds,
	then alternative $k$ cannot have lost strictly more than $n-m-k$ votes and must have won at least one---so in particular, $\mathord{\perm} \in F_1 \union \cdots \union F_{n-k-m+1}$.
	Thus
	\begin{align}
		\Pr( N_1^\strat \geq m )
		&= \sum_{\ell=1}^{n-k-m+1} \Pr(F_\ell) \Pr(N_1^\strat \geq m|F_\ell)
		\nonumber
		\\
		&\leq \sum_{\ell=1}^{n-k-m+1} \Pr(F_\ell)
		\tag{$\sharp$}
		\label{eq:best_Sigma1:i}
		\\
		&\leq \sum_{\ell=1}^{n-k-m+1} \frac{1}{ 2^\ell } ,
		\tag{$\flat$}
		\label{eq:best_Sigma1:ii}
	\end{align}
	where \eqref{eq:best_Sigma1:ii} holds by the \hyperref[claim:pf_lemma:best_Sigmak]{claim}.

	Now suppose that $\strat \in \Sigma_k$;
	we shall show that \eqref{eq:best_Sigma1:i} and \eqref{eq:best_Sigma1:ii} hold with equality, so that $\strat$ attains the bound.
	For \eqref{eq:best_Sigma1:ii}, this follows from the \hyperref[claim:pf_lemma:best_Sigmak]{claim}.
	For \eqref{eq:best_Sigma1:i}, fix an $\ell \in \{1,\dots,n-k-m+1\}$ and a $\mathord{\perm} \in F_\ell$; we must show that $N_k^\strat(\mathord{\perm}) \geq m$.
	Label $\{k+1,\dots,n\} \equiv \{x_{k+1},\dots,x_n\}$ so that
	\begin{equation*}
		x_{k+1} \rankfn{\strat}{\perm} \dots \rankfn{\strat}{\perm} x_n .
	\end{equation*}
	Since $\mathord{\perm} \in F_\ell$, we have by definition of $\Sigma_k$ that $k \rankfn{\strat}{\perm} x_{\ell+1}$.
	Thus $k \rankfn{\strat}{\perm} x_{\ell'}$ for each $\ell' \in \{\ell+1,\dots,n\}$,
	so that $N_k^\strat(\mathord{\perm}) \geq n-\ell \geq m$.
\end{proof}

\begin{proof}[Proof of necessity]
	Take a strategy $\strat$ in $\Sigma_{k-1} \setminus \Sigma_k$.
	Since $\strat$ belongs to $\Sigma_{k-1}$, it satisfies the inequalities \eqref{eq:best_Sigma1:i} and \eqref{eq:best_Sigma1:ii} in the sufficiency argument.
	Suppose that one or the other holds strictly for some $m \in \{1,\dots,n-k\}$, so that $\strat$ fails to attain the bound in the sufficiency proof.
	Since any $\stratp \in \Sigma_k$ attains the bound for every $m$ by the (just-proved) sufficiency part, it follows that $\strat$ is not best for $k$ among $\Sigma_{k-1}$.
	It therefore suffices to find an $m \in \{1,\dots,n-k\}$ such that either \eqref{eq:best_Sigma1:i} or \eqref{eq:best_Sigma1:ii} holds strictly.

	Since $\strat \in \Sigma_{k-1} \setminus \Sigma_k$,
	there is a majority will $\perm$ such that, labelling the alternatives $\{k+1,\dots,n\} \equiv \{x_{k+1},\dots,x_n\}$ so that
	\begin{equation*}
		x_{k+1} \rankfn{\strat}{\perm} \dots \rankfn{\strat}{\perm} x_n,
	\end{equation*}
	one of the following holds under $\perm$:
	\begin{enumerate}[label=(\alph*)]

		\item \label{item:pf_best_Sigmak_necessity:a}
		$\strat$ pits $k$ against some $j \neq x_{k+1}$ prior to pitting it against $x_{k+1}$.

		\item \label{item:pf_best_Sigmak_necessity:b}
		For some $\ell \in \{k+1,\dots,n-k\}$, $\strat$ offers at least $\ell$ votes involving $k$,
		the first $\ell-1$ of which are against $x_{k+1},\dots,x_\ell$ and are all lost by $k$,
		and the $\ell^{\text{th}}$ of which is against some $j \neq x_{\ell+1}$.

		\item \label{item:pf_best_Sigmak_necessity:c}
		For some $\ell \in \{k+1,\dots,n-k\}$, $\strat$ offers exactly $\ell-1$ votes involving $k$, against $x_{k+1},\dots,x_\ell$, each of which is lost by $k$.

	\end{enumerate}

	\emph{Case \ref{item:pf_best_Sigmak_necessity:a}.}
	We shall exhibit a majority will $\mathord{\permp} \in F_1$ such that $N_k^\strat(\mathord{\permp}) < n-k$, so that \eqref{eq:best_Sigma1:i} holds strictly for $m = n-k$.
	This is trivial if $N_k^\strat(\mathord{\perm})=0$ (let $\mathord{\permp} \coloneqq \mathord{\perm}$), so suppose that $N_k^\strat(\mathord{\perm})>0$.
	Let $j \neq x_{k+1}$ be the first alternative pitted against $k$ by $\strat$ under $\perm$, and let $\rank$ be the proto-ranking associated with the history after which this occurs.
	Clearly $j \nrank k \nrank x_{k+1}$.
	Since $\strat \in \Sigma_{k-1}$ and $N_k^\strat(\mathord{\perm})>0$, \hyperref[lemma:worst-case]{Lemma $\natural$} (§\ref{app:pf_proposition:is2:prelim}, \cpageref{lemma:worst-case}) implies that $j > k$.
	So $x_{k+1} \rankfn{\strat}{\perm} j$ (since $j \in \{k+1,\dots,n\} = \{x_{k+1},\dots,n\}$ and $j \neq x_{k+1}$), and thus $j \nrank x_{k+1}$.

	It follows that there is a ranking $\permp$ such that $\mathord{\rank} \subseteq \mathord{\permp}$ and $x_{k+1} \permp k \permp j$.%
		\footnote{To see why, observe that the transitive closure $\rankp$ of $\mathord{\rank} \union \{(x_{k+1},k)\}$ is a proto-ranking since $k \nrank x_{k+1}$.
		Since $j \neq x_{k+1}$ and $j \nrank x_{k+1}$,
		we have $j \nrankp k$ by \Cref{observation:tr_cl} (\cref{app:pf_charac:noerror_efficient}, \cpageref{observation:tr_cl}),
		and thus the transitive closure $\rankpp$ of $\mathord{\rankp} \union \{(k,j)\}$ is also a proto-ranking.
		Let $\permp$ be any ranking that contains $\rankpp$.}
	(Note that a ranking is precisely a transitive majority will.)
	Let $T$ be the period in which $\strat$ offers $\{j,k\}$ under $\perm$.
	Since $\mathord{\rank} \subseteq \mathord{\permp}$, the history of length $T-1$ generated by $\strat$ and $\permp$ is the same as that generated by $\strat$ and $\perm$, and thus $\{j,k\}$ is the first pair involving $k$ that $\strat$ offers under $\permp$.
	Thus $\mathord{\permp} \in F_1$ since $k \permp j$.
	Being a ranking, $\permp$ is the only $\permp$-reachable ranking (by \Cref{observation:reachability_adjacency} in \cref{app:extra:reachability_adjacency}, \cpageref{observation:reachability_adjacency}), so $\mathord{\rankfn{\strat}{\permp}} = \mathord{\permp}$.
	Thus $N_k^\strat(\mathord{\permp}) < n-k$ since $x_{k+1} \permp k$.

	\emph{Case \ref{item:pf_best_Sigmak_necessity:b}.}
	If $j > k$, then the case-\ref{item:pf_best_Sigmak_necessity:a} argument yields a $\mathord{\permp} \in F_\ell$ such that $N_k^\strat(\mathord{\permp}) < n-k-\ell+1$, so that \eqref{eq:best_Sigma1:i} holds strictly for $m = n-k-\ell+1$.
	Suppose instead that $j < k$.
	Let $\permp$ be the majority will that agrees with $\perm$ on every pair, except that $k \permp j$.
	Clearly $\mathord{\permp} \in F_\ell$.
	By (the contra-positive of) \hyperref[lemma:worst-case]{Lemma $\natural$}, we have $N_k^\strat(\mathord{\permp}) = 0$.
	Thus $\Pr(N_k^\strat \geq 1|F_\ell) < 1$, so that \eqref{eq:best_Sigma1:i} holds strictly for $m = 1$.

	\emph{Case \ref{item:pf_best_Sigmak_necessity:c}.}
	Let $E$ be the set of majority wills under which $\strat$ offers at least $\ell-1$ votes involving $k$, the first $\ell-1$ of which $k$ loses.
	Then $\Pr(E) > 0$, and the probability that $\strat$ offers at least $\ell$ votes conditional on $E$ is strictly less than 1.
	The argument for the \hyperref[claim:pf_lemma:best_Sigmak]{claim} in the sufficiency proof therefore yields $\Pr(F_\ell) < 1/2^\ell$,
	so that \eqref{eq:best_Sigma1:ii} holds strictly for e.g. $m = 1$.
\end{proof}

\subsubsection{Proof of \texorpdfstring{\Cref{lemma:is_oe_charac}}{Lemma \ref{lemma:is_oe_charac}} (§\ref{app:pf_proposition:is2:overview}, p. \pageref{lemma:is_oe_charac})}
\label{app:pf_proposition:is2:pf_is_oe_charac}

We must show that membership of $\Sigma_{n-2}$ is necessary and sufficient for outcome-equivalence to insertion sort.
For the sufficiency part, we shall make use of \Cref{lemma:is2:prelim} in §\ref{app:pf_proposition:is2:lemma:best_Sigmak_pf}.

\begin{proof}[Proof of sufficiency]
	Let $\strat \in \Sigma_{n-2}$, fix a majority will $\perm$,
	and let $\rank$ ($\rankp$) be the outcome of $\strat$ (of insertion sort) under $\perm$.
	We will show for each $k \in \{n-2,\dots,1\}$ that $\rank$ agrees with $\rankp$ on $\{k,\dots,n\}$, using induction on $k$.
	For the base case $k=n-2$, let $\{j,\ell\}$ be the first pair offered by $\strat$, where $j<\ell$; it suffices to show that $j = n-1$.
	Suppose to the contrary; then by \Cref{lemma:is2:prelim} (§\ref{app:pf_proposition:is2:lemma:best_Sigmak_pf}, \cpageref{lemma:is2:prelim}), prior to the first vote, $\ell$ is already ranked above every $i > j$ such that $i \neq \ell$, which is absurd.
	The induction step is immediate from the fact that $\strat \in \Sigma_{n-2}$.
\end{proof}

The proof of necessity refers to the argument for \Cref{lemma:best_Sigmak} in §\ref{app:pf_proposition:is2:lemma:best_Sigmak_pf} above.

\begin{proof}[Proof of necessity]
	Let $\strat \notin \Sigma_{n-2}$; we must show that it is not outcome-equivalent to insertion sort.
	Note that there is a (unique) $k \in \{1,\dots,n-2\}$ such that $\strat \in \Sigma_{k-1} \setminus \Sigma_k$.
	Recall the bound in the proof of the sufficiency part of \Cref{lemma:best_Sigmak} (§\ref{app:pf_proposition:is2:lemma:best_Sigmak_pf} above).
	Since insertion sort belongs to $\Sigma_k$, it attains the bound for every $m \in \{1,\dots,n-k\}$ by the \Cref{lemma:best_Sigmak} sufficiency argument, and thus all of its outcome-equivalents do.
	By the \Cref{lemma:best_Sigmak} necessity argument, $\strat$ fails to attain the bound for some $m \in \{1,\dots,n-k\}$, so is not among the outcome-equivalents of insertion sort.
\end{proof}

\crefalias{section}{supplsec}
\crefalias{subsection}{supplsec}
\crefalias{subsubsection}{supplsec}
\section*{Supplemental appendices}
\label{suppl}
\addcontentsline{toc}{section}{Supplemental appendices}

\subsection{Proofs of \texorpdfstring{\Cref{proposition:regretfree_efficient_tight,proposition:regretfree_errors_tight}}{Propositions \ref{proposition:regretfree_efficient_tight} and \ref{proposition:regretfree_errors_tight}} (§\ref{sec:charac}, \texorpdfstring{\cpageref{proposition:regretfree_efficient_tight,proposition:regretfree_errors_tight}}{pp. \pageref{proposition:regretfree_efficient_tight} and \pageref{proposition:regretfree_errors_tight}})}
\label{suppl:pf_tightness}

In this appendix, we establish tightness for the characterisations of regret-freeness in \Cref{theorem:regretfree_efficient,theorem:regretfree_errors}.
We begin in §\ref{suppl:pf_tightness:lemma} with a lemma, then use it to deduce \Cref{proposition:regretfree_errors_tight} (§\ref{suppl:pf_tightness:regretfree_errors}) and \Cref{proposition:regretfree_efficient_tight} (§\ref{suppl:pf_tightness:regretfree_efficient}).

\subsubsection{A lemma}
\label{suppl:pf_tightness:lemma}

\begin{definition}
	\label{definition:make_error}
	Given a proto-ranking $\rank$ and alternatives $x \pref y$ and $z \neq w$, say that \emph{$\{z,w\}$ makes $\{x,y\}$ an error at $\rank$} iff both $x \nrank y \nrank x$ and $z \nrank w \nrank z$, and one of the following holds:
	\begin{itemize}

		\item $x \pref z \pref y$, $y \nrank z \nrank x$, and $w \in \{x,y\}$.

		\item $z \pref y$, $x \rank z$ and $w = y$.

		\item $x \pref z$, $z \rank y$ and $w = x$.

	\end{itemize}
\end{definition}

If $\{z,w\}$ makes $\{x,y\}$ an error, then offering $\{x,y\}$ either misses an opportunity or takes a risk at $\rank$, and the chair `should' offer $\{z,w\}$ instead.%
	\footnote{This is heuristic, as offering $\{z,w\}$ might itself miss an opportunity or take a risk at $\rank$.}

Recall from \cref{app:pf_charac:noerror_efficient} (\cpageref{definition:missed_op}) the definition of a missed opportunity.

\begin{lemma}
	\label{lemma:tightness_lemma}
	Let $\rank$ be a proto-ranking, and let $A \subseteq \mathcal{X}^2$ be a non-empty set of pairs of distinct alternatives.
	Suppose that for any pair $\{x,y\} \in A$, there is a pair $\{z,w\} \in A$ that makes $\{x,y\}$ an error at $\rank$.
	Then $\rank$ contains a missed opportunity.
\end{lemma}

\begin{proof}
	Let $\rank$ and $A$ satisfy the hypothesis.
	Then there is a pair $\{z,w\} \in A$ and another pair $\{z',w'\} \in A$ that makes $\{z,w\}$ an error at $\rank$.
	Assume (wlog) that $z \pref w$ and $z' \pref w'$.
	Since $\{z,w\} \neq \{z',w'\}$, we must have either $z \neq z'$ or $w \neq w'$.
	Assume that $z \neq z'$; the case $w \neq w'$ is similar.

	\begin{namedthm}[First claim.]
		\label{claim:tightness_lemma:1}
		There exists a sequence $(x_t)_{t=1}^T$ in $\mathcal{X}$ with $T \geq 2$ and $x_1 \neq x_2$ such that for every $t \leq T$,
		writing $x_{T+1} \coloneqq x_1$,
		\begin{enumerate}[label=(\roman*)]

			\item \label{item:claim:tightness_lemma:1:i}
			if $x_t \pref x_{t+1}$ then $x_{t+1} \nrank x_t$, and

			\item \label{item:claim:tightness_lemma:1:ii}
			if $x_{t+1} \pref x_t$ then $x_t \rank x_{t+1}$.

		\end{enumerate}
	\end{namedthm}

	\begin{proof}[Proof of the \protect{\hyperref[claim:tightness_lemma:1]{first claim}}]%
		\renewcommand{\qedsymbol}{$\square$}
		Define $\{x_1,y_1\} \coloneqq \{z,w\}$ and $\{x_2,y_2\} \coloneqq \{z',w'\}$.
		By the hypothesis of the lemma,
		there is a pair $\{x_3,y_3\} \in A$ with (wlog) $x_3 \pref y_3$ that makes $\{x_2,y_2\}$ an error at $\rank$,
		a $\{x_4,y_4\} \in A$ with $x_4 \pref y_4$ that makes $\{x_3,y_3\}$ an error at $\rank$,
		and so on.
		Since $A$ is finite, $\{x_1,y_1\}$ makes $\{x_T,y_T\}$ an error for some $T \in \N$.
		We have $T \geq 2$ and $x_1 \neq x_2$ by construction, and \ref{item:claim:tightness_lemma:1:i}--\ref{item:claim:tightness_lemma:1:ii} must hold because $\{x_{t+1},y_{t+1}\}$ makes $\{x_t,y_t\}$ an error at $\rank$.
	\end{proof}%
	\renewcommand{\qedsymbol}{$\blacksquare$}

	Let $(x_t)_{t=1}^T$ be a minimal sequence satisfying the conditions of the \hyperref[claim:tightness_lemma:1]{first claim} (one with no strict subsequence that satisfies the conditions).

	\begin{namedthm}[Second claim.]
		\label{claim:tightness_lemma:2}
		$x_t \neq x_s$ for all distinct $t,s \in \{1,\dots,T\}$.
	\end{namedthm}

	\begin{proof}[Proof of the \protect{\hyperref[claim:tightness_lemma:2]{second claim}}]%
		\renewcommand{\qedsymbol}{$\square$}
		Suppose toward a contradiction that $x_t = x_{t+1}$;
		then the sequence $x_1,\dots,x_{t-1},x_{t+1},\dots,x_T$ satisfies the conditions of the \hyperref[claim:tightness_lemma:1]{first claim}, contradicting the minimality of $(x_t)_{t=1}^T$.
		Assume for the remainder that $x_t \neq x_{t+1}$ for every $t \in \{1,\dots,T\}$.

		Suppose toward a contradiction that $x_t = x_s$, where $t+1<s$.
		Then the sequence $x_{t+1},\dots,x_s$ satisfies the conditions of the \hyperref[claim:tightness_lemma:1]{first claim}, which is absurd since $(x_t)_{t=1}^T$ is minimal.
	\end{proof}
	\renewcommand{\qedsymbol}{$\blacksquare$}

	In light of the \hyperref[claim:tightness_lemma:2]{second claim}, we may re-label the sequence $(x_t)_{t=1}^T$ so that $x_1 \pref x_t$ for every $t \in \{2,\dots,T\}$.
	Let $t' \leq T$ be the least $t \geq 2$
	such that $x_T \pref x_{T-1} \pref \cdots \pref x_t$.
	(So $t'=T$ exactly if $x_{T-1} \pref x_T$.)
	We shall show that $\{x_1,x_{t'}\}$ is a missed opportunity in $\rank$; in particular, that $t' \geq 3$ and
	\begin{enumerate}[label=(\alph*)]

		\item \label{item:claim:tightness_lemma:1:a}
		$x_1 \pref x_{t'-1} \pref x_{t'}$,

		\item \label{item:claim:tightness_lemma:1:b}
		$x_{t'} \rank x_1$, and

		\item \label{item:claim:tightness_lemma:1:c}
		$x_{t'} \nrank x_{t'-1} \nrank x_1$.

	\end{enumerate}

	For \ref{item:claim:tightness_lemma:1:b},
	if $t'=T$ then $x_{t'} = x_T \rank x_1$ by property \ref{item:claim:tightness_lemma:1:ii},
	and if not then $x_{t'} \rank \cdots \rank x_T \rank x_1$ by property \ref{item:claim:tightness_lemma:1:ii}, whence $x_{t'} \rank x_1$ by transitivity of $\rank$.
	The second half of \ref{item:claim:tightness_lemma:1:a} (i.e. $x_{t'-1} \pref x_{t'}$) holds by definition of $t'$.
	The first half of \ref{item:claim:tightness_lemma:1:c} (i.e. $x_{t'} \nrank x_{t'-1}$) then follows by property \ref{item:claim:tightness_lemma:1:i}.
	Since $x_{t'} \rank x_1$, it follows that $t'-1 \neq 1$, which is to say that $t' \geq 3$.
	The first half of \ref{item:claim:tightness_lemma:1:a} (i.e. $x_1 \pref x_{t'-1}$) then holds by construction.
	Finally, the second half of \ref{item:claim:tightness_lemma:1:c} (i.e. $x_{t'-1} \nrank x_1$) must hold since otherwise the sequence $(x_t)_{t=1}^{t'-1}$ would satisfy the conditions of the \hyperref[claim:tightness_lemma:1]{first claim}, contradicting the minimality of $(x_t)_{t=1}^T$.
\end{proof}

\subsubsection{Proof of \texorpdfstring{\Cref{proposition:regretfree_errors_tight}}{Proposition \ref{proposition:regretfree_errors_tight}} (p. \pageref{proposition:regretfree_errors_tight})}
\label{suppl:pf_tightness:regretfree_errors}

At a history at which the chair has committed no errors, the proto-ranking clearly contains no missed opportunities.
The following therefore implies \Cref{proposition:regretfree_errors_tight}.

\begin{namedthm}[\Cref{proposition:regretfree_errors_tight}$\boldsymbol{^\star}$.]
	\label{proposition:prod_gen}
	Let $\rank$ be a non-total proto-ranking containing no missed opportunities.
	Then there exist distinct $x,y \in \mathcal{X}$ such that $x \nrank y \nrank x$ and offering a vote on $\{x,y\}$ neither misses an opportunity nor takes a risk at $\rank$.
\end{namedthm}

\begin{proof}
	Let $\rank$ be a non-total proto-ranking, and suppose that for any distinct $x,y \in \mathcal{X}$ with $x \nrank y \nrank x$, offering a vote on $\{x,y\}$ either misses an opportunity or takes a risk at $\rank$. We shall show that $\rank$ contains a missed opportunity.

	Let $A$ be the set of all pairs $\{x,y\} \subseteq \mathcal{X}$ with $x \neq y$ and $x \nrank y \nrank x$.
	The set $A$ is non-empty since $\rank$ is not total.
	By hypothesis, for any $\{x,y\} \in A$, offering $\{x,y\}$ either misses an opportunity or takes a risk at $\rank$, implying that some $\{z,w\} \in A$ makes $\{x,y\}$ an error at $\rank$.
	It follows by \Cref{lemma:tightness_lemma} (§\ref{suppl:pf_tightness:lemma}, \cpageref{lemma:tightness_lemma}) that $\rank$ contains a missed opportunity.
\end{proof}

\subsubsection{Proof of \texorpdfstring{\Cref{proposition:regretfree_efficient_tight}}{Proposition \ref{proposition:regretfree_efficient_tight}} (p. \pageref{proposition:regretfree_efficient_tight})}
\label{suppl:pf_tightness:regretfree_efficient}

\begin{lemma}
	\label{lemma:pf_tightness:regretfree_efficient}
	Fix a majority will $\perm$, let $\rank$ be a $\perm$-reachable $\perm$-efficient ranking,
	and let $\mathord{\rankp} \subseteq \mathord{\rank}$ be a non-total proto-ranking containing no missed opportunities.
	Then there exist distinct $x,y \in \mathcal{X}$ such that $x \nrankp y \nrankp x$,
	$\perm$ and $\rank$ agree on $\{x,y\}$,
	and offering $\{x,y\}$ does not miss an opportunity or take a risk at $\rankp$.
\end{lemma}

\begin{proof}[Proof of \Cref{proposition:regretfree_efficient_tight}]
	Fix a majority will $\perm$ and a $\perm$-reachable $\perm$-efficient ranking $\rank$.
	By \Cref{proposition:regretfree_errors_tight} (\cpageref{proposition:regretfree_errors_tight}; already proved), it suffices to find a terminal history $((x_t,y_t))_{t=1}^T$, with associated proto-rankings $(\mathord{\rankt})_{t=0}^T$,%
		\footnote{Recall that $\mathord{\rankzero}=\varnothing$ and that $\mathord{\rankt}$ is the transitive closure of $\Union_{s=1}^t \{(x_s,y_s)\}$, for each $t$.}
	such that
	\begin{itemize}

		\item for every $t \in \{1,\dots,T\}$, $x_t \perm y_t$ and $x_t \rank y_t$, and

		\item for every $t \in \{2,\dots,T\}$, offering $\{x_t,y_t\}$ does not miss an opportunity or take a risk at $\ranktone$.

	\end{itemize}
	Such a terminal history is obtained by repeatedly applying \Cref{lemma:pf_tightness:regretfree_efficient}.
\end{proof}

\begin{proof}[Proof of \Cref{lemma:pf_tightness:regretfree_efficient}]
	We shall prove the contra-positive.
	Fix a majority will $\perm$ and a $\perm$-reachable $\perm$-efficient ranking $\rank$, and let $\mathord{\rankp} \subseteq \mathord{\rank}$ be a non-total proto-ranking.
	Suppose that for any distinct $x,y \in \mathcal{X}$ with $x \nrankp y \nrankp x$ such that $\perm$ and $\rank$ agree on $\{x,y\}$, offering $\{x,y\}$ misses an opportunity or takes a risk at $\rankp$.
	We will show that $\rankp$ contains a missed opportunity.

	Let $\mathcal{A}$ be the set of all pairs $\{x,y\} \subseteq \mathcal{X}$ such that $x \pref y$, $x \nrankp y \nrankp x$, and there is no $z \in \mathcal{X}$ such that $x \rank z \rank y$.
	(So $\mathcal{A}$ is a set of two-element subsets of $\mathcal{X}$.)
	The set $\mathcal{A}$ is non-empty since it includes any $\rank$-adjacent pair $\{x,y\}$ with $x \nrankp y \nrankp x$, and there must be such a pair since $\rankp$ is non-total and $\mathord{\rankp} \subseteq \mathord{\rank}$.
	By \Cref{lemma:tightness_lemma} (§\ref{suppl:pf_tightness:lemma}, \cpageref{lemma:tightness_lemma}), it suffices to show that for any pair $\{x,y\} \in \mathcal{A}$, there is a pair $\{z,w\} \in \mathcal{A}$ that makes $\{x,y\}$ an error at $\rankp$.

	So fix a pair $\{x,y\} \in \mathcal{A}$.
	We claim that $\perm$ and $\rank$ must agree on $\{x,y\}$.
	If $x,y$ are $\rank$-adjacent, then this holds by \Cref{observation:reachability_adjacency} (\cref{app:extra:reachability_adjacency}, \cpageref{observation:reachability_adjacency}) since $\rank$ is $\perm$-reachable.
	If $x,y$ are not $\rank$-adjacent, then since no $z \in \mathcal{X}$ satisfies $x \rank z \rank y$, it must be that $y \rank x$.
	Since $x \pref y$ and $\rank$ is $\perm$-efficient, it follows that $y \perm x$, so that $\perm$ and $\rank$ agree on $\{x,y\}$.

	It follows from the (contra-positive) hypothesis that offering $\{x,y\}$ either misses an opportunity or takes a risk at $\rankp$.
	Consider each in turn.

	\emph{Case 1: $\{x,y\}$ misses an opportunity.}
 	In this case there is a $z \in \mathcal{X}$ satisfying $x \pref z \pref y$ and $y \nrankp z \nrankp x$.
	Since $\{x,y\} \in \mathcal{A}$, we must have either $z \rank x$ or $y \rank z$.
	Assume that $z \rank x$; the case $y \rank z$ is analogous.
	Since $\mathord{\rankp} \subseteq \mathord{\rank}$, we have $x \nrankp z \nrankp x$.
	Thus the pair $\{x,z\}$ lives in $\mathcal{A}$ and makes $\{x,y\}$ an error at $\rankp$.

	\emph{Case 2: $\{x,y\}$ takes a risk.}
	Assume that there is a $z \in \mathcal{X}$ such that $z \pref y$, $x \rankp z$ and $y \nrankp z$; the case in which $x \pref z$, $z \rankp y$ and $z \nrankp x$ is similar.
	Then $\{y,z\}$ makes $\{x,y\}$ an error at $\rankp$.
	To see that $\{y,z\}$ belongs to $\mathcal{A}$, observe that (i) $z \pref y$ and $y \nrankp z$, 
	that (ii) $z \nrankp y$ since otherwise $x \rankp z$ and the transitivity of $\rankp$ would imply the falsehood $x \rankp y$, and that (iii) $y \rank z$ since $\{x,y\} \in \mathcal{A}$ and $x \rank z$ (as $x \rankp z$ and $\mathord{\rankp} \subseteq \mathord{\rank}$), so that there is no $z' \in \mathcal{X}$ such that $z \rank z' \rank y$.
\end{proof}

\subsection{Relation to ranking methods}
\label{suppl:ranking_method}

In this appendix, we investigate the link with the social choice literature mentioned in §\ref{sec:intro:lit} (\cpageref{sec:intro:lit}).
We recast the chair's problem as a choice among ranking methods,
characterise the constraint set of this problem,
and compare its solutions to ranking methods in the literature.

A \emph{ranking method} is a map that assigns to each majority will a ranking.
Each strategy $\strat$ induces a ranking method,
namely the one whose value at a majority will $\perm$ is the outcome of $\strat$ under $\perm$.
Call a ranking method $\rho$ \emph{feasible} iff it is induced by some strategy,
and \emph{regret-free} iff $\rho(\mathord{\perm})$ is $\perm$-unimprovable for every $\perm$.
Clearly the chair's problem in §\ref{sec:environment} can be re-formulated
as a choice between ranking methods,
where the constraint set consists of the feasible ranking methods
and the objective is to choose a regret-free one.

For a majority will $\perm$ and rankings $\mathord{\rank},\mathord{\rankp}$,
say that $\rank$ is \emph{more aligned with $\perm$ than} $\rankp$
iff for any pair $x,y \in \mathcal{X}$ of alternatives with $x \perm y$, if $x \rankp y$ then also $x \rank y$.
This is exactly the definition in the text (§\ref{sec:environment:preferences}, \cpageref{definition:more_aligned}), except that we allow $\perm$ to be any majority will (not necessarily a ranking).

\begin{definition}
	\label{definition:faithful}
	A ranking method $\rho$ is \emph{faithful} iff
	for every majority will $\perm$,
	no ranking $\mathord{\rank} \neq \rho(\mathord{\perm})$ is more aligned with $\perm$ than $\rho(\mathord{\perm})$.
\end{definition}

Faithfulness clearly admits a normative interpretation.
It is a natural strengthening of Condorcet consistency,
the requirement that $\rho(\mathord{\perm})$ rank $x$ highest if $x \perm y$ for every alternative $y \neq x$.
The following shows that it also has a positive interpretation:

\begin{observation}
	\label{observation:Wfeas_faithful}
	A ranking method $\rho$ is faithful
	iff $\rho(\mathord{\perm})$ is $\perm$-reachable for every majority will $\perm$.
\end{observation}

\begin{proof}
	Fix a ranking method $\rho$ and a majority will $\perm$, and write $\mathord{\rank} \coloneqq \rho(\mathord{\perm})$.
	If $\rank$ is $\perm$-reachable, then any $\mathord{\rankp} \neq \mathord{\rank}$ fails to be more aligned with $\perm$ since it must rank some $\rank$-adjacent pair $x \rank y$ as $y \rankp x$,
	where $x \perm y$ by \Cref{observation:reachability_adjacency} (\cref{app:extra:reachability_adjacency}, \cpageref{observation:reachability_adjacency}).
	If $\rank$ is not $\perm$-reachable, then by \Cref{observation:reachability_adjacency} there is an $\rank$-adjacent pair $x \rank y$ such that $y \perm x$,
	so the ranking $\mathord{\rankp} \neq \mathord{\rank}$ that agrees with $\rank$ on every pair but $x,y$ is more aligned with $\perm$.
\end{proof}

By \Cref{observation:Wfeas_faithful}, any feasible ranking method must be faithful.
The converse does not hold, because feasibility also imposes restrictions across majority wills.
To describe these constraints, we introduce a second property:

\begin{definition}
	\label{definition:consistent}
	A ranking method $\rho$ is \emph{consistent} iff
	whenever $\rho(\mathord{\perm}) \neq \rho(\mathord{\permp})$ for two majority wills $\perm$ and $\permp$,
	there are alternatives $x,y \in \mathcal{X}$ such that $x \perm y \permp x$ and
	\begin{equation*}
		x \mathrel{\rho(\mathord{\permpp})} y
		\quad \text{iff} \quad
		x \permpp y
		\quad \text{for every majority will $\mathord{\permpp} \supseteq \mathord{\perm} \intersect \mathord{\permp}$.}
	\end{equation*}
\end{definition}

This property is mathematically natural, but we do not think that it has any normative appeal.
Instead, it captures constraints that the rules of the interaction impose on the chair:

\begin{proposition}
	\label{proposition:meth_feas}
	A ranking method is feasible
	iff it is faithful and consistent.
\end{proposition}

Call a ranking method $\rho$ \emph{efficient}
iff $\rho(\mathord{\perm})$ is $\perm$-efficient for every majority will $\perm$.
($\perm$-efficiency is defined in §\ref{sec:is}, \cpageref{sec:is:efficiency}.)
By \Cref{theorem:regretfree_efficient} (\cpageref{theorem:regretfree_efficient}),
a feasible ranking method is regret-free iff it is efficient.
Thus:

\begin{corollary}
	\label{corollary:meth_rf}
	A ranking method is feasible and regret-free
	iff it is faithful, consistent and efficient.
\end{corollary}

While faithfulness has a normative interpretation,
consistency and efficiency are `positive' in nature:
the former is a constraint imposed by the rules of the game,
while the latter is defined in terms of the chair's self-interested preference $\pref$.
This makes feasible and regret-free ranking methods quite different from those studied in the literature, which are characterised by purely normative axioms (e.g. \textcite{Rubinstein1980} for the Copeland method).
Indeed, standard ranking methods such as those of Copeland and Kemeny--Slater are neither consistent nor efficient, though the latter is faithful.

\begin{proof}[Proof of \Cref{proposition:meth_feas}]
	For necessity, let $\rho$ be feasible.
	Then $\rho$ is faithful by \Cref{observation:Wfeas_faithful}.
	To show that it is consistent, let $\strat$ be a strategy inducing $\rho$, and fix majority wills $\perm$ and $\permp$ such that $\rho(\mathord{\perm}) \neq \rho(\mathord{\permp})$.
	Let $t$ be the first period in which the history generated by $\strat$ and $\perm$ differs from that generated by $\strat$ and $\permp$,
	and let $\{x,y\}$ be the pair offered in this period.
	Then $\perm$ and $\permp$ disagree on $\{x,y\}$.
	Furthermore, the pair $\{x,y\}$ is clearly offered in period $t$ of the history generated by $\strat$ and any $\mathord{\permpp} \supseteq \mathord{\perm} \intersect \mathord{\permp}$, so that $x \mathrel{\rho(\mathord{\permpp})} y$ iff $x \permpp y$.

	For sufficiency, let $\rho$ be faithful and consistent; we shall construct a strategy that induces $\rho$.
	For each history $h$, let $\permarg{h}$ and $\permparg{h}$ be majority wills such that
	\begin{enumerate}[label=(\alph*)]
	
		\item \label{bullet:revealed_W:a}
		if $h = ((x_t,y_t))_{t=1}^T$,
		then $x_t \permarg{h} y_t$ and $x_t \permparg{h} y_t$ for each $t \in \{1,\dots,T\}$, and
	
		\item \label{bullet:revealed_W:b}
		$\permarg{h}$ and $\permparg{h}$ disagree on any pair that is not voted on in $h$.
	
	\end{enumerate}
	Since $\rho$ is faithful, $\rho(\mathord{\permarg{h}})$ is $\permarg{h}$-reachable and $\rho(\mathord{\permparg{h}})$ is $\permparg{h}$-reachable
	by \Cref{observation:Wfeas_faithful}.
	Thus by \Cref{observation:reachability_adjacency} (\cref{app:extra:reachability_adjacency}, \cpageref{observation:reachability_adjacency}),
	we have $\rho(\mathord{\permarg{h}}) = \rho(\mathord{\permparg{h}})$ iff $h$ is terminal.
	Since $\rho$ is consistent,
	we may for each non-terminal history $h$ choose a pair $\strat(h) \coloneqq \{x,y\} \subseteq \mathcal{X}$ that satisfies
	\begin{enumerate}[label=(\alph*)]

		\setcounter{enumi}{2}
	
		\item \label{bullet:revealed_W:c}
		$x \permarg{h} y \permparg{h} x$ and
	
		\item \label{bullet:revealed_W:d}
		$x \permpp y$ iff $x \mathrel{\rho(\mathord{\permpp})} y$ for any majority will $\mathord{\permpp} \supseteq \mathord{\permarg{h}} \intersect \mathord{\permparg{h}}$.
	
	\end{enumerate}

	\begin{namedthm}[Claim.]
		\label{claim:meth_feas}
		Let $h = ((x_t,y_t))_{t=1}^T$ be a history of length $T \geq 1$
		such that $\{x_1,y_1\} = \strat(\varnothing)$ and
		$\{x_t,y_t\} = \strat\left( ((x_s,y_s))_{s=1}^{t-1} \right)$
		for each $t \in \{2,\dots,T\}$.
		Then (i) for any majority will $\permpp$ with $x_t \permpp y_t$ for each $t \in \{1,\dots,T\}$,
		we have $x_t \mathrel{\rho(\mathord{\permpp})} y_t$ for each $t \in \{1,\dots,T\}$, and (ii) the pair $\strat(h)$ is unranked by the transitive closure of $\Union_{t=1}^T \{(x_t,y_t)\}$.
	\end{namedthm}

	\begin{proof}[Proof of the \protect{\hyperref[claim:meth_feas]{claim}}]%
		\renewcommand{\qedsymbol}{$\square$}
		For the first part, fix a $t \in \{1,\dots,T\}$ and a majority will $\permpp$ such that $x_s \permpp y_s$ for each $s \in \{1,\dots,T\}$.
		Define $h' \coloneqq ((x_s,y_s))_{s=1}^{t-1}$ (meaning $h' = \varnothing$ if $t=1$), noting that $\strat(h') = \{x_t,y_t\}$.
		We have $\mathord{\permpp} \supseteq \mathord{\permarg{h'}} \intersect \mathord{\permparg{h'}}$ since $\permarg{h'}$ and $\permparg{h'}$ satisfy \ref{bullet:revealed_W:b},
		whence $x_t \mathrel{\rho(\mathord{\permpp})} y_t$ by \ref{bullet:revealed_W:d}.

		For the second part, we have $x_t \mathrel{\rho(\mathord{\permarg{h}})} y_t$ and $x_t \mathrel{\rho(\mathord{\permparg{h}})} y_t$ for every $t \in \{1,\dots,T\}$ by \ref{bullet:revealed_W:a} and the first part of the \hyperref[claim:meth_feas]{claim},
		implying that $\rho(\mathord{\permarg{h}})$ and $\rho(\mathord{\permparg{h}})$ (being transitive) agree on every pair ranked by the transitive closure of $\Union_{t=1}^T \{(x_t,y_t)\}$.
		Since $\rho(\mathord{\permarg{h}})$ and $\rho(\mathord{\permparg{h}})$ disagree on the pair $\strat(h)$ by \ref{bullet:revealed_W:c} and \ref{bullet:revealed_W:d}, it follows that $\strat(h)$ is unranked by the transitive closure.
	\end{proof}%
	\renewcommand{\qedsymbol}{$\blacksquare$}

	By the second part of the \hyperref[claim:meth_feas]{claim}, $\strat$ is a well-defined strategy.%
		\footnote{We actually defined $\strat$ only on the path. Off the path, any behaviour will do.}
	To show that it induces $\rho$, fix a majority will $\perm$, and let $h = ((x_t,y_t))_{t=1}^T$ be the terminal history generated by $\strat$ and $\perm$;
	we must demonstrate that $\rho(\mathord{\perm})$ is the transitive closure of $\Union_{t=1}^T \{(x_t,y_t)\}$.
	Since both are rankings, it suffices to show that $x_t \mathrel{\rho(\mathord{\perm})} y_t$ for every $t \in \{1,\dots,T\}$.
	And this follows from the \hyperref[claim:meth_feas]{claim} since $x_t \perm y_t$ for every $t \in \{1,\dots,T\}$.
\end{proof}

\subsection{\texorpdfstring{How many $\boldsymbol{\perm}$-reachable rankings are $\boldsymbol{\perm}$-unimprovable?}{How many W-reachable rankings are W-unimprovable?}}
\label{suppl:frac}

This appendix contains two results.
In §\ref{suppl:frac:benefit},
we show that for a given majority will $\perm$,
every $\perm$-reachable ranking is $\perm$-unimprovable iff $\perm$ is transitive.
In §\ref{suppl:frac:frac},
we show that on average across majority wills $\perm$,
only a small fraction of $\perm$-reachable rankings are $\perm$-unimprovable
if there are many alternatives.

\subsubsection{When is agenda-setting valuable?}
\label{suppl:frac:benefit}

Given the majority will $\perm$, the value of agenda-setting lies in being able to reach a $\perm$-unimprovable ranking rather than some (necessarily $\perm$-reachable) ranking that is not $\perm$-unimprovable.
This motivates the following definition:

\begin{definition}
	\label{definition:chair_benefit}
	Given her preference $\pref$ (a ranking),
	the chair \emph{benefits from agenda-setting} under a majority will $\perm$
	iff there exists a $\perm$-reachable ranking that is not $\perm$-unimprovable.
\end{definition}

\begin{proposition}
	\label{proposition:chair_benefit}
	For a majority will $\perm$, the following are equivalent:
	\begin{enumerate}

		\item \label{proposition:chair_benefit:i}
		$\perm$ is not a ranking (i.e. is not transitive).

		\item \label{proposition:chair_benefit:ii}
		For some $\pref$, the chair benefits from agenda-setting under $\perm$.

		\item \label{proposition:chair_benefit:iii}
		For every $\pref$, the chair benefits from agenda-setting under $\perm$.

	\end{enumerate}
\end{proposition}

In words, agenda-setting is valuable precisely because it allows the chair to exploit Condorcet cycles: the chair benefits whenever there is a cycle in $\perm$, and otherwise does not benefit.

\begin{proof}
	\ref{proposition:chair_benefit:iii} immediately implies \ref{proposition:chair_benefit:ii}.
	To see that \ref{proposition:chair_benefit:ii} implies \ref{proposition:chair_benefit:i}, consider the contra-positive: if $\perm$ is a ranking, then it is clearly the only $\perm$-reachable ranking, so the chair does not benefit from agenda-setting for any $\pref$.

	To prove that \ref{proposition:chair_benefit:i} implies \ref{proposition:chair_benefit:iii}, fix any ranking $\pref$ and any majority will $\perm$ that is not a ranking;
	it suffices to exhibit distinct $\perm$-reachable rankings $\rank$ and $\rankp$ such that $\rank$ is more aligned with $\pref$ than $\rankp$.
	To that end let $\rank$ be a $\perm$-efficient ranking (these are easily seen to exist).
	Similarly let $\rankp$ be a \emph{$\perm$-anti-efficient} ranking,
	i.e. one such that $x \apref y$ and $x \perm y$ implies $x \rank y$.

	To show that $\rank$ is more aligned with $\pref$ than $\rankp$,
	take $x,y \in \mathcal{X}$ with $x \pref y$.
	If $x \perm y$, then $x \rank y$ since $\rank$ is $\perm$-efficient.
	If instead $y \perm x$, then $y \rankp x$ since $\rankp$ is $\perm$-anti-efficient.
	Thus $x \rankp y$ implies $x \rank y$.

	It remains only to show that $\rank$ and $\rankp$ are distinct.
	Since $\perm$ is not a ranking, there must be $x,y,z \in \mathcal{X}$ such that $x \perm y \perm z \perm x$.
	Suppose wlog that $x \pref z$.
	There are three cases.
	If $x \pref y \pref z$, then $x \rank y \rank z$ and $z \rankp x$.
	If $y \pref x \pref z$, then $y \rank z$ and $z \rankp x \rankp y$.
	If $x \pref z \pref y$, then $x \rank y$ and $y \rankp z \rankp x$.
	In each case,
	$\mathord{\rank} \neq \mathord{\rankp}$
	by transitivity.
\end{proof}

\subsubsection{\texorpdfstring{Most $\boldsymbol{\perm}$-reachable rankings are not $\boldsymbol{\perm}$-unimprovable}{Most W-reachable rankings are not W-unimprovable}}
\label{suppl:frac:frac}

The following asserts that when there are enough alternatives,
only a small fraction of a typical $\perm$'s $\perm$-reachable rankings are $\perm$-unimprovable.

\begin{proposition}
	\label{proposition:frac}
	For each $n \in \N$,
	let $\ranksup{n}$ ($\permsup{n}$) denote a uniform random draw from the set of all rankings (majority wills) on $\mathcal{X}_n \coloneqq \{1,\dots,n\}$,
	with $\ranksup{n}$ and $\permsup{n}$ independent.
	Then
	\begin{equation*}
		\Pr\left( \text{$\ranksup{n}$ is $\permsup{n}$-unimprovable} \middle|
		\text{$\ranksup{n}$ is $\permsup{n}$-reachable} \right)
		\to 0
		\quad \text{as $n \to \infty$.}
	\end{equation*}
\end{proposition}

\begin{proof}
	Fix any $n \geq 5$, and define $K_n \coloneqq \floor{ (n-1)/4 }$.
	Further fix a ranking $\rank$ and a majority will $\perm$ on $\mathcal{X}_n$,
	and label the alternatives $\mathcal{X}_n = \{x_1,\dots,x_n\}$ so that $x_1 \rank \dots \rank x_n$.
	Given $k \in \{1,\dots,K_n\}$, say that $\rank$ \emph{admits a local $\perm$-improvement at $(x_{4k-2},x_{4k-1},x_{4k})$} iff both
	\begin{itemize}

		\item $x_{4k-3} \perm x_{4k} \perm x_{4k-2}$ and
		$x_{4k-1} \perm x_{4k+1}$, and

		\item $x_{4k} \pref x_{4k-1}$ and $x_{4k} \pref x_{4k-2}$.

	\end{itemize}
	If $\rank$ admits a local $\perm$-improvement at $(x_{4k-2},x_{4k-1},x_{4k})$,
	then it fails to be $\perm$-unimprovable since the ranking
	\begin{equation*}
		x_1 \rankp \cdots
		\rankp x_{4k-3} \rankp x_{4k} \rankp x_{4k-2}
		\rankp x_{4k-1} \rankp x_{4k+1}
		\rankp \cdots \rankp x_n
	\end{equation*}
	is then $\perm$-reachable (by \Cref{observation:reachability_adjacency} in \cref{app:extra:reachability_adjacency} (\cpageref{observation:reachability_adjacency}))
	and more aligned with $\pref$.

	For each $n \geq 5$, let $\left( X^n_k \right)_{k=1}^n$ be random variables such that
	\begin{equation*}
		\left\{ X^n_1, \dots, X^n_n \right\} = \mathcal{X}_n
		\quad \text{and} \quad
		X^n_1 \ranksup{n} \dots \ranksup{n} X^n_n
		\quad \text{a.s.}
	\end{equation*}
	The events `$X^n_{4k} \pref X^n_{4k-1}$ and $X^n_{4k} \pref X^n_{4k-2}$' are independent across $k \in \{1,\dots,K_n\}$
	and each have probability $1/4$.
	It follows by \Cref{observation:reachability_adjacency}
	that conditional on $\ranksup{n}$ being $\permsup{n}$-reachable, the events
	\begin{equation*}
		\text{`$\ranksup{n}$ admits a local $\permsup{n}$-improvement at $\left( X^n_{4k-2}, X^n_{4k-1}, X^n_{4k} \right)$'}
	\end{equation*}
	are independent across $k \in \{1,\dots,K_n\}$
	and have probability $(1/2)^5$.
	Thus
	\begin{equation*}
		\Pr\left( \text{$\ranksup{n}$ is $\permsup{n}$-unimprovable} \middle|
		\text{$\ranksup{n}$ is $\permsup{n}$-reachable} \right)
		\leq \left( 1 - (1/2)^5 \right)^{K_n} ,
	\end{equation*}
	which vanishes as $n \to \infty$ since $K_n = \floor{ (n-1)/4 }$ diverges.
\end{proof}

\subsection{A characterisation of our \texorpdfstring{`}{'}transitive' protocol}
\label{suppl:trans_charac}

In this appendix, we show that among all possible rules of interaction between the chair and committee that lead to a ranking,
the `transitive' protocol studied in this paper (described in §\ref{sec:environment:interaction}) is the only one (up to restriction) that denies the chair arbitrary power and that allows votes only on pairs.
This protocol is thus the natural one, given the restriction to pairwise votes.
Non-binary votes raise issues that are beyond the scope of this paper.%
	\footnote{Unlike in the binary case, there is no `most natural' non-binary protocol.
	In particular, reasonable protocols can differ in what they deem the committee to have `decided' in a vote on three or more alternatives in which none won an outright majority.}

A \emph{ballot} is a set of two or more alternatives.
An \emph{election} is $(B,V)$, where $B$ is a ballot and $V$ is a map $\{1,\dots,I\} \to B$ specifying what alternative each voter votes for.
An \emph{electoral history} is a finite sequence of elections with distinct ballots.
For two (distinct) electoral histories $h,h'$, we write $h \sqsubseteq \mathrel{(\sqsubset)} h'$ iff $h$ is a truncation of $h'$.

A \emph{protocol} specifies for each (permitted) electoral history either
(1) a set of ballots that the chair is permitted to offer or
(2) a ranking.
Formally:

\begin{definition}
	\label{definition:protocol}
	A \emph{protocol} is $(\mathcal{H},\rho)$, where
	\begin{enumerate}

		\item $\mathcal{H}$ is a non-empty set of electoral histories such that
		\begin{itemize}

			\item if $h'$ belongs to $\mathcal{H}$, then so does any $h \sqsubseteq h'$, and

			\item if $h = ( (B_1,V_1),\dots,(B_t,V_t) )$
			belongs to $\mathcal{H}$,
			then so does $h' = ( (B_1,V_1),\dots,(B_t,V_t') )$
			for any $V_t' : \{1,\dots,n\} \to B_t$.

		\end{itemize}

		Call $h \in \mathcal{H}$ \emph{terminal} (in $\mathcal{H}$) iff there is no $h' \sqsupset h$ in $\mathcal{H}$.

		\item $\rho$ is a map that assigns a ranking to each terminal $h \in \mathcal{H}$.

	\end{enumerate}
\end{definition}

Call an electoral history \emph{binary} iff each ballot has exactly two elements.
A \emph{binary protocol} $(\mathcal{H},\rho)$ is one whose $\mathcal{H}$ consists of binary electoral histories.
For any binary electoral history $h = \left( \left( \{x_s,y_s\}, V_s \right) \right)_{s=1}^t$,
where wlog $\abs*{ \{ i : V_s(i)=x_s \} } > I/2$ for each $s \in \{1,\dots,t\}$,
let $\rankh$ denote the transitive closure of $\Union_{s=1}^t \{(x_s,y_s)\}$.%
	\footnote{If $h$ is the empty electoral history, then $\mathord{\rankh} = \varnothing$.}
The \emph{transitive protocol} is the binary protocol that permits the chair to offer a ballot $\{x,y\}$ after binary electoral history $h$ exactly if the pair $x,y$ is unranked by $\rankh$, and assigns the ranking $\rankh$ to each terminal $h$.%
	\footnote{Explicitly it is $(\mathcal{H}^\star,\rho^\star)$, where $\mathcal{H}^\star$ consists of all binary electoral histories $h'$ such that
	\begin{equation*}
		h 
		\sqsubset ( (\{x_1,y_1\},V_1), \dots, (\{x_t,y_t\},V_t) )
		\sqsubseteq h'
		\quad \text{implies} \quad
		x_t \nrankh y_t \nrankh x_t ,
	\end{equation*}
	(so that $h \in \mathcal{H}$ is terminal iff $\rankh$ is a ranking,)
	and $\rho^\star(h) \coloneqq \mathord{\rankh}$ for each terminal $h \in \mathcal{H}^\star$.}

To deny the chair arbitrary power, we focus on protocols that rank $x$ above $y$ whenever $x$ won an outright majority and $y$ was also on the ballot:

\begin{definition}
	\label{definition:cs}
	A protocol $(\mathcal{H},\rho)$ satisfies \emph{committee sovereignty}
	iff for any terminal $h = ( (B_t,V_t) )_{t=1}^T \in \mathcal{H}$
	such that $\abs*{ \{ i : V_t(i) = x \} } > I/2$ and $y \in B_t \setminus \{x\}$ for some $t \in \{1,\dots,T\}$,
	we have $x \mathrel{\rho(h)} y$.
\end{definition}

For binary protocols, committee sovereignty is equivalent to imposing transitivity after every vote:

\begin{observation}
	\label{observation:transitive_charac}
	A binary protocol $(\mathcal{H},\rho)$ satisfies committee sovereignty
	iff $\rho(h) \supseteq \mathord{\rankh}$ for every terminal $h \in \mathcal{H}$.
\end{observation}

That is, any pair linked by a chain of majorities ($x \rankh y$)
must be ranked accordingly ($x \mathrel{\rho(h)} y$),
and so cannot be offered for a vote.%
	\footnote{Formally: if $x \rankh y$, then no terminal $h' \sqsupseteq h$ can feature the ballot $\{x,y\}$ (except in $h$).
	For otherwise there would be a terminal $h'$ in which $y$ beats $x$ in a vote,
	so that $x \rankhp y \rankhp x$,
	which is impossible since $\rho(h') \supseteq \mathord{\rankhp}$ and $\rho(h')$ is a ranking.}

\begin{proof}
	Let $(\mathcal{H},\rho)$ be binary and satisfy committee sovereignty, and
	take a terminal $h = ( ( \{x_t,y_t\}, V_t ) )_{t=1}^T \in \mathcal{H}$,
	where wlog $x_t \rankh y_t$ for each $t \in \{1,\dots,T\}$.
	Then $\rho(h) \supseteq \Union_{t=1}^T \{(x_t,y_t)\}$ by committee sovereignty,
	whence $\rho(h) \supseteq \mathord{\rankh}$
	because $\rho(h)$ is transitive
	and $\rankh$ is by definition the smallest transitive relation containing $\Union_{t=1}^T \{(x_t,y_t)\}$.

	For the converse, let $(\mathcal{H},\rho)$ be binary with $\rho(h) \supseteq \mathord{\rankh}$ for every terminal $h \in \mathcal{H}$.
	Take any terminal $h = ( (\{x_t,y_t\},V_t) )_{t=1}^T \in \mathcal{H}$ and suppose that $\abs*{ \{ i : V_t(i) = x_t \} } > I/2$; we must show that $x_t \mathrel{\rho(h)} y_t$.
	Since $x_t \rankh y_t$, this follows immediately from $\rho(h) \supseteq \mathord{\rankh}$.
\end{proof}

More is needed to deny the chair excessive power: she must also be required to offer enough ballots to give the committee a fair say.
To formalise this, write $x \sayh y$ for an electoral history $h = ( (B_t,V_t) )_{t=1}^T$ iff
\begin{equation*}
	x,y \in B_t
	\quad \text{and} \quad
	\abs*{ \{ i : V_t(i) = x \} }
	\geq \abs*{ \{ i : V_t(i) = y \} }
\end{equation*}
for some $t \in \{1,\dots,T\}$,
and say that $h$ \emph{gives the committee a say on $x,y$} iff
$\{z_1,z_L\} = \{x,y\}$ for some sequence $z_1 \sayh z_2 \sayh \cdots \sayh z_L$ of alternatives.

\begin{definition}
	\label{definition:dl}
	A protocol $(\mathcal{H},\rho)$ satisfies \emph{democratic legitimacy}
	iff every terminal $h \in \mathcal{H}$
	gives the committee a say on each pair of alternatives.
\end{definition}

Write $\tau(\mathcal{H})$ for the terminal elements of $\mathcal{H}$.
A protocol $(\mathcal{H},\rho)$ is a \emph{restriction} of $(\mathcal{H}',\rho')$
iff $\tau(\mathcal{H}) \subseteq \tau(\mathcal{H}')$ and $\rho = \rho'|_{\tau(\mathcal{H})}$.%
	\footnote{$\tau(\mathcal{H}) \subseteq \tau(\mathcal{H}')$ holds exactly if $\mathcal{H} \subseteq \mathcal{H}'$ \emph{and} any $h \in \tau(\mathcal{H})$ is terminal in $\mathcal{H}'$.}
To wit, anything the chair can do under $(\mathcal{H},\rho)$, she can also do under $(\mathcal{H}',\rho')$.

\begin{proposition}
	\label{proposition:transitive_charac}
	A protocol is binary and satisfies committee sovereignty and democratic legitimacy
	iff it is a restriction of the transitive protocol.
\end{proposition}

Thus any binary protocol that does not give the chair arbitrary power must be the transitive protocol, possibly with limitations on what unranked pairs the chair may offer at some histories.
Neglecting such limitations as ad hoc, we arrive at the transitive protocol.

\begin{proof}
	Any restriction of the transitive protocol $(\mathcal{H}^\star,\rho^\star)$
	satisfies the three properties since $(\mathcal{H}^\star,\rho^\star)$ does
	and the properties are preserved under restriction.
	For the converse, let $(\mathcal{H},\rho)$ satisfy the three properties;
	we must show that $\tau(\mathcal{H}) \subseteq \tau(\mathcal{H}^\star)$ and $\rho = \rho^\star|_{\tau(\mathcal{H})}$.

	To establish $\tau(\mathcal{H}) \subseteq \tau(\mathcal{H}^\star)$,
	we show separately that $\mathcal{H} \subseteq \mathcal{H}^\star$
	and that any $h \in \tau(\mathcal{H}) \subseteq \mathcal{H}^\star$ is terminal in $\mathcal{H}^\star$.
	For the former, fix a pair of electoral histories
	$h \sqsubset h' = ( ( \{x_s,y_s\}, V_s ) )_{s=1}^t \in \mathcal{H}$.
	We must show that the pair $x_t,y_t$ is unranked by $\rankh$,
	so suppose toward a contradiction that $x_t \rankh y_t$.
	Then we must have $x_t \mathrel{\rho(h'')} y_t$ for any terminal $h'' \in \mathcal{H}$ such that $h'' \sqsupseteq h$ since $\rho(h'') \supseteq \mathord{\rankhpp} \supseteq \mathord{\rankh}$ by \Cref{observation:transitive_charac}.
	In particular, this must hold for any terminal $h'' \in \mathcal{H}$
	with first $t-1$ elections $(\{x_1,y_1\},V_1), \dots, (\{x_{t-1},y_{t-1}\},V_{t-1})$ and $t^\text{th}$ election $(\{x_t,y_t\},V_t')$, where $V_t'$ satisfies $\abs*{ \{ i : V_t'(i) = y_t \} } > I/2$.
	But for such an $h''$, committee sovereignty of $(\mathcal{H},\rho)$ clearly demands that $y_t \mathrel{\rho(h'')} x_t$---a contradiction.

	For the latter, let $h \in \tau(\mathcal{H}) \subseteq \mathcal{H}^\star$;
	we must show that $h$ is terminal in $\mathcal{H}^\star$, meaning precisely that $\rankh$ is total.
	Since $h$ is binary, a pair $x,y$ is ranked by $\rankh$ iff $h$ gives the committee a say on $x,y$.
	And $h$ gives the committee a say on every pair since $(\mathcal{H},\rho)$ satisfies democratic legitimacy.

	To show that $\rho = \rho^\star|_{\tau(\mathcal{H})}$, fix an $h \in \tau(\mathcal{H})$.
	Then $\rho(h) \supseteq \mathord{\rankh} = \rho^\star(h)$ by \Cref{observation:transitive_charac} and the definition of $\rho^\star$,
	and the containment must be an equality since $\rho(h)$ and $\rho^\star(h)$ are both rankings.
\end{proof}

\subsection{The limits of strategic voting}
\label{suppl:ic}

In this appendix, we show that sincere voting is the unique regret-free strategy of each voter.
In fact, we show something stronger:
deviating from sincere voting results in a no better (a worse) outcome against any (some) strategies of the chair and the other voters,
in the `more aligned' sense.

Let each voter $i \in \{1,\dots,I\}$ have a strict preference $\prefi$ over the alternatives $\mathcal{X}$.
A \emph{strategy} $\strati$ of a voter specifies, after each history and for every offered pair $x,y$, whether $x$ or $y$ should be voted for.
(Since a history records only which pairs were offered and which alternative won in each pair, not who voted how,
this definition of a strategy rules out complex path-dependence.
We shall relax that stricture below.)

A voter's \emph{sincere strategy} is the one that always instructs her to vote for whichever alternative she likes better.
For a strategy $\strat$ of the chair and strategies $\stratv{1},\dots,\stratv{I}$ of the voters, write $\mathord{\rank}(\strat,\stratv{1},\dots,\stratv{I})$ for the outcome (the ranking that results).

\begin{definition}
	\label{definition:obv_b}
	Let $\strati,\stratvp{i}$ be strategies of voter $i$, and $\strat,\stratv{-i}$ strategies of the chair and the other voters.
	$\stratvp{i}$ is \emph{obviously better than} $\stratv{i}$ against $\strat,\stratv{-i}$ iff $\mathord{\rank}(\strat,\stratvp{i},\stratv{-i})$
	is distinct from, and more aligned with $\prefi$ than,
	$\mathord{\rank}(\strat,\stratv{i},\stratv{-i})$.
\end{definition}

When one strategy is obviously better than another, it yields a better outcome no matter what voter $i$'s exact preference over rankings, given only the weak assumption that voter $i$ weakly prefers rankings more aligned with her preference $\prefi$ over alternatives.
By contrast, comparing strategies that are not related by `obviously better than' involves trade-offs.

\begin{definition}
	\label{definition:ds}
	A strategy $\stratv{i}$ of a voter is \emph{dominant} iff for any alternative strategy $\stratvp{i}$,
	\begin{enumerate}[label=(\alph*)]

		\item[\namedlabel{item:ds_all}{$(\mathord{\nexists})$}]
		there exist no strategies $\strat,\stratv{-i}$ of the chair and other voters against which $\stratvp{i}$ is obviously better than $\stratv{i}$, and

		\item[\namedlabel{item:ds_some}{$(\mathord{\exists})$}]
		there exist strategies $\strat,\stratv{-i}$ of the chair and other voters against which $\stratv{i}$ is obviously better than $\stratvp{i}$.

	\end{enumerate}
\end{definition}

Dominance is strong. (Albeit not as strong as conventional dominance, since `obviously better' is only a partial ordering.)
Observe that there can be at most one dominant strategy.
In fact, there is exactly one:

\begin{proposition}
	\label{proposition:ic}
	For each voter, the sincere strategy is (uniquely) dominant.
\end{proposition}

\Cref{proposition:ic} remains true, with the same proof, if the definition of dominance is strengthened to allow the alternative strategy $\stratip$ to be an `extended strategy' that can condition on who voted how in previous periods.

\begin{proof}
	Fix a voter $i$, and let $\stratv{i}^\star$ be her sincere strategy.
	We must establish properties \ref{item:ds_all} and \ref{item:ds_some} in the definition of dominance.

	\emph{Property \ref{item:ds_all}:}
	Fix strategies $\strat,\stratv{-i}$ of the chair and other voters and a non-sincere strategy $\stratvp{i}$ of voter $i$, and suppose that $\mathord{\rankp} \coloneqq \mathord{\rank}(\strat,\stratvp{i},\stratv{-i})$ is distinct from $\mathord{\rank} \coloneqq \mathord{\rank}(\strat,\stratv{i}^\star,\stratv{-i})$;
	we must show that $\rankp$ is not more aligned with $\prefi$ than $\rank$.
	Let $T$ be the first period in which the proto-rankings $\rankT$ and $\rankpT$ differ, and let $\{x,y\}$ be the pair voted on in this period, where (wlog) $x \rankT y$ and $y \rankpT x$.
	The two strategy profiles generate the same length-$(T-1)$ history $h$ (by definition of $T$), and thus the same period-$T$ votes $\stratv{j}(h)$ by the other voters $j \neq i$.
	So voter $i$ is pivotal after history $h$, and since $\stratv{i}^\star$ is sincere it must be that $x \prefi y$.
	Thus $\mathord{\rankp} \supseteq \mathord{\rankpT}$ is not more aligned with $\prefi$ than $\mathord{\rank} \supseteq \mathord{\rankT}$.

	\emph{Property \ref{item:ds_some}:}
	Take any non-sincere strategy $\stratvp{i}$.
	Choose strategies $\strat',\stratv{-i}'$ such that $\stratvp{i}$ votes non-sincerely along the terminal history induced by the strategy profile $(\strat',\stratvp{i},\stratv{-i}')$, and let $T$ be the first period in which this occurs.
	Then the proto-ranking in period $T-1$ is the same under the strategy profiles $(\strat',\stratvp{i},\stratv{-i}')$ and $(\strat',\stratv{i}^\star,\stratv{-i}')$; call it $\rankTone$.
	Write $\{x,y\}$ for the pair of alternatives that are voted on in period $T$, where (wlog) $x \prefi y$.

	\begin{namedthm}[Extension claim.]
		\label{claim:extension}
		Let $\rankpp$ be a proto-ranking, and let $x,y \in \mathcal{X}$ be distinct alternatives such that $x \nrankpp y \nrankpp x$.
		Then there is a ranking $\mathord{\rankp} \supseteq \mathord{\rankpp}$ such that $x,y$ are $\rankp$-adjacent with $x \rankp y$.
	\end{namedthm}

	\begin{proof}[Proof of the \protect{\hyperref[claim:extension]{extension claim}}]%
		\renewcommand{\qedsymbol}{$\square$}
		Let $\rankpp$ and $x,y \in \mathcal{X}$ satisfy the hypothesis.
		Then $\mathord{\rankpp} \union \{x,y\}^2$ admits a complete and transitive extension $\rel$ by the \hyperref[lemma:extension]{extension lemma} in \cref{app:pf_charac:unimp_noerror}.
		Note that $x \rel y \rel x$.
		It follows that by appropriately breaking indifferences in $\rel$, we may obtain a ranking $\mathord{\rankp} \supseteq \mathord{\rankpp} \union \{(x,y)\}$ such that $x,y$ are $\rankp$-adjacent with $x \rankp y$.
	\end{proof}%
	\renewcommand{\qedsymbol}{$\blacksquare$}

	By the \hyperref[claim:extension]{extension claim}, there exists a ranking $\mathord{\rank} \supseteq \mathord{\rankTone}$ with $x \rank y$ and $x,y$ $\rank$-adjacent.
	Let $\rankp$ be exactly $\rank$, except with the positions of $x$ and $y$ reversed.
	Clearly $\rank$ is more aligned with $\prefi$ than $\rankp$, and the two are distinct.

	It thus suffices to find strategies $\strat$ and $\stratv{-i}$ such that $\mathord{\rank}(\strat,\stratv{i}^\star,\stratv{-i}) = \mathord{\rank}$ and $\mathord{\rank}(\strat,\stratvp{i},\stratv{-i}) = \mathord{\rankp}$.
	For the chair, let $\strat \coloneqq \strat'$.
	As for $\stratv{-i}$, let half of the other voters $j \in I \setminus \{i\}$ vote according $\rank$ (i.e. vote for $z$ over $w$ iff $z \rank w$), and the rest vote according to $\rankp$.
\end{proof}

\subsection{Extension: indecisive votes}
\label{suppl:indecisive}

In this appendix, we allow the vote on a pair of alternatives to be indecisive, in which case the chair may choose how they are ranked.
(This occurs e.g. when the chair is a voting member of the committee.)
To that end, we re-interpret $x \perm y$ to mean that the chair is \emph{permitted} to rank $x$ above $y$, and allow for the possibility that both $x \perm y$ and $y \perm x$.
A vote on $\{x,y\}$ with $x \perm y$ is \emph{indecisive} if also $y \perm x$, and \emph{decisive} otherwise.

The `majority will' $\perm$ must still be total and irreflexive, but not necessarily asymmetric.
By appeal to an argument similar to that for \Cref{fact:arise_charac} (\cref{app:extra:arise_charac}, \cpageref{fact:arise_charac}), \emph{any} total and irreflexive relation $\perm$ should be considered.

A history still records what pairs were offered and how each pair was ranked, and a \emph{strategy} now specifies not only what pair to offer after each history, but also how to rank them if the vote is indecisive.
Note that a history does not record whether a vote was decisive or not, and thus that we rule out strategies that condition on this information.
We show in \cref{suppl:non-simple} how this restriction may be dropped.

Regret-free and efficient strategies are defined as before, with `for any majority will $\perm$' replaced by `for any total and irreflexive $\perm$'.
By \Cref{lemma:Wefficient_Wunimprovable} (\cpageref{lemma:Wefficient_Wunimprovable}), efficiency still implies regret-freeness.

When the chair offers $\{x,y\}$ with $x \pref y$ and the vote is indecisive, we say that she \emph{ranks in her interest} iff she ranks $x$ above $y$, and \emph{against her interest} otherwise.
Augment the definition of insertion sort in §\ref{sec:is} so that the chair ranks in her interest whenever a vote is indecisive.
\Cref{theorem:is_efficient} (§\ref{sec:is}, \cpageref{theorem:is_efficient}) remains true, with the same proof: insertion sort is efficient, and thus regret-free.

The characterisations of regret-free strategies (\Cref{theorem:regretfree_efficient,theorem:regretfree_errors} in §\ref{sec:charac}, \cpageref{theorem:regretfree_efficient,theorem:regretfree_errors}) extend as follows:

\begin{namedthm}[Theorem (\ref{theorem:regretfree_efficient}+\ref{theorem:regretfree_errors})$\boldsymbol{'}$.]
	For a strategy $\strat$, the following are equivalent:
	\hspace{1pt}
	\begin{enumerate}[label=(\alph*)]

		\item \label{item:charac_indecisive:regret-free}
		$\strat$ is regret-free.

		\item \label{item:charac_indecisive:efficient}
		$\strat$ is efficient.

		\item \label{item:charac_indecisive:noerror}
		$\strat$ never misses an opportunity, takes a risk, or ranks against the chair's interest.

	\end{enumerate}
\end{namedthm}

\begin{proof}
	We establish the implications depicted in \Cref{fig:charac} (\cpageref{fig:charac}).
	The proof that \ref{item:charac_indecisive:noerror} implies \ref{item:charac_indecisive:efficient} given in \cref{app:pf_charac:noerror_efficient} applies essentially unchanged.
	That \ref{item:charac_indecisive:efficient} implies \ref{item:charac_indecisive:regret-free} follows from \Cref{lemma:Wefficient_Wunimprovable} (\cpageref{lemma:Wefficient_Wunimprovable}).

	It remains to show that \ref{item:charac_indecisive:regret-free} implies \ref{item:charac_indecisive:noerror}.
	To establish that regret-free strategies never miss an opportunity or take a risk, it suffices to replicate the argument in \cref{app:pf_charac:unimp_noerror}.
	To show that a regret-free strategy must not rank against the chair's interest, we prove the contra-positive.
	Let $\strat$ be a strategy that ranks against the chair's interest under some majority will $\perm$; we shall find a majority will $\permp$ such that the outcome $\rank$ of $\strat$ under $\permp$ fails to be $\permp$-unimprovable.
	In particular, we shall exhibit a $\permp$-reachable ranking $\mathord{\rankp} \neq \mathord{\rank}$ that is more aligned with $\pref$ than $\rank$.

	Let $T$ be the first period in which $\strat$ ranks against the chair's interest under $\perm$.
	Write $\rankTone$ for the associated end-of-period-$(T-1)$ proto-ranking, and let $\{x,y\}$ be the pair offered in period $T$.
	By hypothesis, $x \perm y \perm x$, and the chair chooses to rank $y$ above $x$.

	By the \hyperref[claim:extension]{extension claim} in \cref{suppl:ic} (\cpageref{claim:extension}), there exists a ranking $\mathord{\rankp} \supseteq \mathord{\rankTone} \union \{(x,y)\}$ such that $x,y$ are $\rankp$-adjacent.
	Define a majority will $\permp$ by $\mathord{\permp} \coloneqq \mathord{\rankp} \union \{(y,x)\}$, and denote by $\rank$ the outcome of $\strat$ under $\permp$.
	Clearly $\rankp$ is $\permp$-reachable.
	It remains to show that $\mathord{\rank} \neq \mathord{\rankp}$ and that $\rankp$ is more aligned with $\pref$ than $\rank$.

	For the former, since $x \rankp y$, it suffices to show that $y \rank x$.
	To this end, observe that that $\mathord{\rankTone} \subseteq \mathord{\rankp} \subseteq \mathord{\permp}$.
	Thus the history of length $T-1$ generated by $\strat$ and $\permp$ is the same as that generated by $\strat$ and $\perm$, which means in particular that $\{x,y\}$ is offered in period $T$, and that $y$ is ranked above $x$ if the vote is indecisive.
	Under $\permp$, the vote is indeed indecisive ($x \permp y \permp x$), and thus $y \rank x$ as desired.

	To show that $\rankp$ is more aligned with $\pref$ than $\rank$,
	observe that $\permp$ agrees with $\rankp$ on every pair $\{z,w\} \nsubseteq \{x,y\} = [x,y]_{\mathord{\rankp}}$.
	Thus by \Cref{lemma:unimp_noerror:1} in \cref{app:pf_charac:unimp_noerror} (\cpageref{lemma:unimp_noerror:1}), $\rank$ and $\rankp$ agree on every pair $\{z,w\} \neq \{x,y\}$.
	Since $x \pref y$ and $x \rankp y$, it follows that $\rankp$ is more aligned with $\pref$ than $\rank$.
\end{proof}

All of the remaining results also extend:
the characterisations of regret-freeness are tight (\Cref{proposition:regretfree_efficient_tight,proposition:regretfree_errors_tight} in §\ref{sec:charac}, \cpageref{proposition:regretfree_efficient_tight,proposition:regretfree_errors_tight}),
the outcome-equivalents of insertion sort are (include) the lexicographic (amendment) strategies (\Cref{theorem:is_lexicographic} and \Cref{proposition:is_amendment} in §\ref{sec:charac}, \cpageref{theorem:is_lexicographic,proposition:is_amendment}),
and sincere voting is dominant (\Cref{proposition:ic} in \cref{suppl:ic}).

\subsection{Extension: strategies with extended history-dependence}
\label{suppl:non-simple}

By definition, a strategy does not condition on who voted how in the past.
To relax this restriction, let an \emph{extended history} be a sequence of pairs offered and votes cast by each member on each pair, and let an \emph{extended strategy} assign to each extended history an unranked pair of alternatives.

Recall from \cref{app:extra:arise_charac} (\cpageref{definition:voting_profile}) the definition of voting profiles.
The \emph{outcome} of an extended strategy $\strat$ under a voting profile $(\mathord{\votei})_{i=1}^I$ is the final ranking that results.
A strategy is \emph{regret-free} (\emph{efficient}) iff its outcome under every voting profile $(\mathord{\votei})_{i=1}^I$ is $\perm$-unimprovable ($\perm$-efficient), where $\perm$ denotes the majority will of $(\mathord{\votei})_{i=1}^I$.

Insertion sort is clearly an extended strategy, so is efficient by \Cref{theorem:is_efficient} (§\ref{sec:is}, \cpageref{theorem:is_efficient}).
Our characterisation of regret-freeness (\Cref{theorem:regretfree_efficient,theorem:regretfree_errors} in §\ref{sec:charac}, \cpageref{theorem:regretfree_efficient,theorem:regretfree_errors}) remains valid:

\begin{namedthm}[Theorem (\ref{theorem:regretfree_efficient}+\ref{theorem:regretfree_errors})$\boldsymbol{''}$.]
	For an extended strategy $\strat$, the following are equivalent:
	\hspace{1pt}
	\begin{enumerate}[label=(\alph*)]

		\item \label{item:charac_non-simple:regret-free}
		$\strat$ is regret-free.

		\item \label{item:charac_non-simple:efficient}
		$\strat$ is efficient.

		\item \label{item:charac_non-simple:noerror}
		$\strat$ never misses an opportunity or takes a risk.

	\end{enumerate}
\end{namedthm}

\begin{proof}
	We prove the implications depicted in \Cref{fig:charac} (\cpageref{fig:charac}).
	That \ref{item:charac_non-simple:noerror} implies \ref{item:charac_non-simple:efficient} follows from the argument in \cref{app:pf_charac:noerror_efficient}, which applies unchanged to extended strategies.
	That \ref{item:charac_non-simple:efficient} implies \ref{item:charac_non-simple:regret-free} follows from \Cref{lemma:Wefficient_Wunimprovable} (\cpageref{lemma:Wefficient_Wunimprovable}).

	To show that \ref{item:charac_non-simple:regret-free} implies \ref{item:charac_non-simple:noerror}, we prove the contra-positive by augmenting the argument in \cref{app:pf_charac:unimp_noerror}.
	Take an extended strategy $\strat$ that misses an opportunity or takes a risk under some voting profile $(\mathord{\votei})_{i=1}^I$, and let $t$ be the first period in which this occurs.
	Let $\perm$ be the majority will of $(\mathord{\votei})_{i=1}^I$.
	Construct an alternative majority will $\permp$ exactly as in the proof in \cref{app:pf_charac:unimp_noerror}.
	Construct in addition a voting profile $(\mathord{\votepi})_{i=1}^I$ whose majority will is $\permp$, and such that the extended history up to time $t$ under $\strat$ and $(\mathord{\votepi})_{i=1}^I$ is the same as under $\strat$ and $(\mathord{\votei})_{i=1}^I$.
	The argument in \cref{app:pf_charac:unimp_noerror} ensures that the outcome of $\strat$ under $(\mathord{\votepi})_{i=1}^I$ fails to be $\permp$-unimprovable.
	Thus $\strat$ fails to be regret-free.
\end{proof}

\end{appendices}



\printbibliography[heading=bibintoc]


\end{document}